\documentclass[12pt,pdftex,a4paper]{amsart}

\selectfont

\usepackage{qtree}
\usepackage{wrapfig}
\usepackage{forest}
\usepackage{mathtools} 

\usepackage{caption}

\usepackage{etoolbox}
\usepackage{epic,eepic,ecltree}
\usepackage{graphicx}
\usepackage{tikz}
\usepackage{tikz-cd}

\usetikzlibrary{decorations.text}
\usepackage{amssymb,color, euscript, enumerate}
\usepackage{amsthm}
\usepackage{amsmath}
\usepackage{braket}
\usepackage{amscd}
\usepackage{txfonts}
\usepackage{comment}
\usepackage{bm}
\usepackage{amscd}
\numberwithin{equation}{section}

\makeatletter
\let\old@tocline\@tocline
\let\section@tocline\@tocline
\newcommand{\subsection@dotsep}{4.5}
\newcommand{\subsubsection@dotsep}{4.5}
\patchcmd{\@tocline}
  {\hfil}
  {\nobreak
     \leaders\hbox{$\m@th
        \mkern \subsection@dotsep mu\hbox{.}\mkern \subsection@dotsep mu$}\hfill
     \nobreak}{}{}
\let\subsection@tocline\@tocline
\let\@tocline\old@tocline

\patchcmd{\@tocline}
  {\hfil}
  {\nobreak
     \leaders\hbox{$\m@th
        \mkern \subsubsection@dotsep mu\hbox{.}\mkern \subsubsection@dotsep mu$}\hfill
     \nobreak}{}{}
\let\subsubsection@tocline\@tocline
\let\@tocline\old@tocline

\let\old@l@subsection\l@subsection
\let\old@l@subsubsection\l@subsubsection

\def\@tocwriteb#1#2#3{%
  \begingroup
    \@xp\def\csname #2@tocline\endcsname##1##2##3##4##5##6{%
      \ifnum##1>\c@tocdepth
      \else \sbox\z@{##5\let\indentlabel\@tochangmeasure##6}\fi}%
    \csname l@#2\endcsname{#1{\csname#2name\endcsname}{\@secnumber}{}}%
  \endgroup
  \addcontentsline{toc}{#2}%
    {\protect#1{\csname#2name\endcsname}{\@secnumber}{#3}}}%

\newlength{\@tocsectionindent}
\newlength{\@tocsubsectionindent}
\newlength{\@tocsubsubsectionindent}
\newlength{\@tocsectionnumwidth}
\newlength{\@tocsubsectionnumwidth}
\newlength{\@tocsubsubsectionnumwidth}
\newcommand{\settocsectionnumwidth}[1]{\setlength{\@tocsectionnumwidth}{#1}}
\newcommand{\settocsubsectionnumwidth}[1]{\setlength{\@tocsubsectionnumwidth}{#1}}
\newcommand{\settocsubsubsectionnumwidth}[1]{\setlength{\@tocsubsubsectionnumwidth}{#1}}
\newcommand{\settocsectionindent}[1]{\setlength{\@tocsectionindent}{#1}}
\newcommand{\settocsubsectionindent}[1]{\setlength{\@tocsubsectionindent}{#1}}
\newcommand{\settocsubsubsectionindent}[1]{\setlength{\@tocsubsubsectionindent}{#1}}

\renewcommand{\l@section}{\section@tocline{1}{\@tocsectionvskip}{\@tocsectionindent}{}{\@tocsectionformat}}%
\renewcommand{\l@subsection}{\subsection@tocline{2}{\@tocsubsectionvskip}{\@tocsubsectionindent}{}{\@tocsubsectionformat}}%
\renewcommand{\l@subsubsection}{\subsubsection@tocline{3}{\@tocsubsubsectionvskip}{\@tocsubsubsectionindent}{}{\@tocsubsubsectionformat}}%
\newcommand{\@tocsectionformat}{}
\newcommand{\@tocsubsectionformat}{}
\newcommand{\@tocsubsubsectionformat}{}
\expandafter\def\csname toc@1format\endcsname{\@tocsectionformat}
\expandafter\def\csname toc@2format\endcsname{\@tocsubsectionformat}
\expandafter\def\csname toc@3format\endcsname{\@tocsubsubsectionformat}
\newcommand{\settocsectionformat}[1]{\renewcommand{\@tocsectionformat}{#1}}
\newcommand{\settocsubsectionformat}[1]{\renewcommand{\@tocsubsectionformat}{#1}}
\newcommand{\settocsubsubsectionformat}[1]{\renewcommand{\@tocsubsubsectionformat}{#1}}
\newlength{\@tocsectionvskip}
\newcommand{\settocsectionvskip}[1]{\setlength{\@tocsectionvskip}{#1}}
\newlength{\@tocsubsectionvskip}
\newcommand{\settocsubsectionvskip}[1]{\setlength{\@tocsubsectionvskip}{#1}}
\newlength{\@tocsubsubsectionvskip}
\newcommand{\settocsubsubsectionvskip}[1]{\setlength{\@tocsubsubsectionvskip}{#1}}

\patchcmd{\tocsection}{\indentlabel}{\makebox[\@tocsectionnumwidth][l]}{}{}
\patchcmd{\tocsubsection}{\indentlabel}{\makebox[\@tocsubsectionnumwidth][l]}{}{}
\patchcmd{\tocsubsubsection}{\indentlabel}{\makebox[\@tocsubsubsectionnumwidth][l]}{}{}

\newcommand{\@sectypepnumformat}{}
\renewcommand{\contentsline}[1]{%
  \expandafter\let\expandafter\@sectypepnumformat\csname @toc#1pnumformat\endcsname%
  \csname l@#1\endcsname}
\newcommand{\@tocsectionpnumformat}{}
\newcommand{\@tocsubsectionpnumformat}{}
\newcommand{\@tocsubsubsectionpnumformat}{}
\newcommand{\setsectionpnumformat}[1]{\renewcommand{\@tocsectionpnumformat}{#1}}
\newcommand{\setsubsectionpnumformat}[1]{\renewcommand{\@tocsubsectionpnumformat}{#1}}
\newcommand{\setsubsubsectionpnumformat}[1]{\renewcommand{\@tocsubsubsectionpnumformat}{#1}}
\renewcommand{\@tocpagenum}[1]{%
  \hfill {\mdseries\@sectypepnumformat #1}}

\let\oldappendix\appendix
\renewcommand{\appendix}{%
  \leavevmode\oldappendix%
  \addtocontents{toc}{%
    \protect\settowidth{\protect\@tocsectionnumwidth}{\protect\@tocsectionformat\sectionname\space}%
    \protect\addtolength{\protect\@tocsectionnumwidth}{2em}}%
}
\makeatother

\makeatletter
\settocsectionnumwidth{2em}
\settocsubsectionnumwidth{2.5em}
\settocsubsubsectionnumwidth{3em}
\settocsectionindent{1pc}%
\settocsubsectionindent{\dimexpr\@tocsectionindent+\@tocsectionnumwidth}%
\settocsubsubsectionindent{\dimexpr\@tocsubsectionindent+\@tocsubsectionnumwidth}%
\makeatother

\settocsectionvskip{10pt}
\settocsubsectionvskip{0pt}
\settocsubsubsectionvskip{0pt}
    
\settocsectionformat{\bfseries}
\settocsubsectionformat{\mdseries}
\settocsubsubsectionformat{\mdseries}
\setsectionpnumformat{\bfseries}
\setsubsectionpnumformat{\mdseries}
\setsubsubsectionpnumformat{\mdseries}

\let\oldtableofcontents\tableofcontents
\renewcommand{\tableofcontents}{%
  \vspace*{-\linespacing}
  \oldtableofcontents}

\newcommand{\CP}{\C P^1}

\newcommand{\msl}{\mathrm{sl}}

\newcommand{\Harm}{\mathrm{Harm}}

\newcommand{\Emb}{\mathrm{Mfld}_{d,\mathrm{emb}}^{\mathrm{CO}}}
\newcommand{\Embf}{\mathrm{Mfld}}
\newcommand{\Embc}{\mathrm{Mfld}_{d}^{\mathrm{CO}}}
\newcommand{\EmbFc}{\mathrm{Mfld}_{d}^{\mathrm{CF}}}

\newcommand{\Disk}{\mathrm{Disk}_d^{\mathrm{CO}}}

\newcommand{\SO}{\mathrm{SO}}
\newcommand{\PO}{\mathrm{PO}}

\newcommand{\tDisk}{{\mathrm{Disk}}_{d,\mathrm{emb}}^{\mathrm{CO}}}
\newcommand{\Hilb}{{\mathsf{Hilb}}}
\newcommand{\Ind}{{\operatorname{Ind}(\Hilb)}}
\newcommand{\colim}{\mathrm{colim}}

\newcommand{\fS}{\mathbb{S}}

\newcommand{\fG}{\mathbb{G}}

\newcommand{\Hom}{{\mathrm{Hom}}}

\newcommand{\dn}{{\frac{d-2}{2}}}

\newcommand{\da}{\dagger}

\newcommand{\cC}{{\mathcal C}}
\newcommand{\bH}{\mathbb H}

\newcommand{\N}{\mathbb{N}}

\newcommand{\R}{\mathbb{R}}
\newcommand{\C}{\mathbb{C}}

\newcommand{\pr}{{\mathrm{pr}}}

\newcommand{\Om}{\Omega}

\newcommand{\Conf}{\mathrm{Conf}}

\newcommand{\sod}{\mathrm{so}(d+1,1)}

\newcommand{\alg}{\text{alg}}

\newcommand{\std}{{\text{std}}}

\newcommand{\va}{\bm{1}}

\newcommand{\id}{{\mathrm{id}}}

\newcommand{\z}{{\bar{z}}}

\newcommand{\m}{\rho}

\newcommand{\pa}{{\partial}}

\newcommand{\Vect}{{\underline{\text{Vect}}_\C}}
\newcommand{\VectR}{{\underline{\text{Vect}}_\R}}

\newcommand{\al}{\alpha}
\newcommand{\ep}{\epsilon}
\newcommand{\be}{\beta}
\newcommand{\ga}{\gamma}

\newcommand{\de}{\delta}

\newcommand{\la}{\lambda}

\newcommand{\si}{\sigma}

\newcommand{\ft}{\frac{1}{2}}

\newcommand{\Ld}{{\overline{L}}}

\newcommand{\Sym}{\mathrm{Sym}}

\newcommand{\Hf}{{\mathcal{H}}}
\newcommand{\Aut}{\mathrm{Aut}\,}

\newcommand{\End}{\mathrm{End}}

\newcommand{\hotimes}{\hat{\otimes}}

\newcommand{\CE}{\mathbb{CE}_d}
\newcommand{\CEd}{\mathbb{CE}_d}
\newcommand{\CEc}{\mathbb{CE}_{d}^\mathrm{emb}}

\newcommand{\CEt}{\mathbb{CE}_2}
\newcommand{\CEct}{\mathbb{CE}_{2}^\mathrm{emb}}

\newcommand{\bD}{\mathbb{D}}

\newcommand{\bs}{\hat{S}}
\newcommand{\bn}{\hat{N}}

\newcommand{\norm}[1]{\left\lVert #1 \right\rVert}

\newcommand{\op}{{\mathrm{op}}}

\newtheorem{thm}{Theorem}[section]
\newtheorem{dfn}[thm]{Definition}
\newtheorem{lem}[thm]{Lemma}
\newtheorem{prop}[thm]{Proposition}
\newtheorem{cor}[thm]{Corollary}
\newtheorem{rem}[thm]{Remark}

\theoremstyle{remark}

\begin{document}

\begin{center}
{{\LARGE \bf Conformally flat factorization homology in Ind-Hilbert spaces and conformal field theory}
} \par \bigskip

\renewcommand*{\thefootnote}{\fnsymbol{footnote}}
{\normalsize
Yuto Moriwaki \footnote{email: \texttt{moriwaki.yuto (at) gmail.com}}
}
\par \bigskip
{\footnotesize Interdisciplinary Theoretical and Mathematical Science Program (iTHEMS)\\
Wako, Saitama 351-0198, Japan}

\par \bigskip
\end{center}

\noindent


%
%
\begin{center}
\textbf{\large Abstract}
\end{center}

We introduce a metric-dependent geometric variant of factorization homology in conformally flat Riemannian geometry for $d \geq 2$.
Its coefficients are symmetric monoidal functors from a disk category in conformal Riemannian geometry to the ind-category of Hilbert spaces, which we call conformally flat $d$-disk algebras. We prove that their left Kan extensions define symmetric monoidal invariants of conformally flat manifolds. Under suitable positivity and continuity assumptions, the value on the standard sphere recovers the sphere partition function of the associated conformal field theory. For $d\geq 3$, we construct explicit examples
using unitary representations of
$\mathrm{SO}^+(d,1)$ and harmonic analysis, and show that their operadic structures do not extend to bounded operations on the natural Hilbert space completions.

%
%
%

\tableofcontents

\begin{center}
\textbf{\large Introduction}
\end{center}

\emph{Factorization homology}, also known as topological chiral homology, is a homology-type theory
introduced by Lurie and further developed by Ayala--Francis \cite{Lurie2, AF1}.
Taking a $d$-disk algebra $A$ as input, factorization homology assigns invariants of
$d$-dimensional manifolds via the left Kan extension:
\begin{equation}
\label{eq_intro_left}
\begin{tikzcd}
\mathrm{Disk}_d
  \arrow[r,"A"]
  \arrow[d,hook,"i"']
&
\mathcal{C}
\\
\mathrm{Mfld}_d
  \arrow[ru,dashed,"\mathrm{Lan}_{A}"']
&
\end{tikzcd}
\end{equation}
Factorization homology is formulated in an $\infty$-categorical setting, and its output
is closely related to metric-independent field theories, i.e., \emph{topological field
theories} \cite{Lurie3, Sc, CS, AFR, AFT, AF2}.
The topological nature of factorization homology makes it natural to ask whether an analogous local-to-global construction can be formulated in metric-dependent geometries.

In this paper, we develop a conformal Riemannian variant of factorization homology for
$d\geq 2$.
We introduce the category $\Embc$ of germs of $d$-dimensional Riemannian manifolds and
orientation-preserving conformal open embeddings, and its full monoidal subcategory $\Disk$ generated by the flat unit disk.
We call a symmetric monoidal functor 
\begin{align}
A:\Disk \rightarrow \cC
\label{intro_disk_add}
\end{align}
a \emph{conformally flat $d$-disk algebra}. 
It is a conformal-geometric analogue of a $d$-disk algebra and is intended to encode
the local structure of a $d$-dimensional \emph{conformal field theory} (CFT), that is, a
metric-dependent quantum field theory with conformal symmetry.
This viewpoint is closely related to Costello--Gwilliam's formulation of quantum field
theory (QFT) by factorization algebras \cite{CG1,CG2}, especially the functorial construction of
the free scalar field in Riemannian geometry \cite[Chapter 6.3]{CG1}.

While factorization algebras reveal an elegant link between QFT and geometry, 
many constructions of factorization algebras arising from QFT are
perturbative in nature, formulated over formal power series $\C[[h]]$, which typically have zero radius of convergence
\cite{Dy,Jaffe}.
On the other hand, analytic approaches to QFT, such as the G{\aa}rding--Wightman axioms, involve Hilbert spaces, unitary representations of the Poincar\'e group, and operator-valued tempered distributions  (quantum fields) \cite{GW,SW}.
In this analytic setting, non-formal and nontrivial examples are known,
for instance through the probabilistic construction of path integrals by Glimm--Jaffe
\cite{Glimm-Jaffe}.
However, these constructions are usually formulated on flat spacetime.

The aim of
this paper is to construct a categorical framework that combines conformal
Riemannian geometry with the Hilbert-space analytic structures of conformal field
theory.
From the perspective of factorization homology, one is led to regard the local analytic theory as the input datum and the corresponding global geometric theory as its left Kan extension.
A main difficulty in placing analytic QFT into a categorical framework is that, as in
the G{\aa}rding--Wightman axioms, quantum fields remain unbounded operators even after
smearing by test functions. Hence it is not clear a priori whether the target category $\cC$ in
\eqref{intro_disk_add} can be taken to be the category $\Hilb$ of separable Hilbert
spaces.

For $d \geq 3$,
we construct examples of conformally flat $d$-disk algebras using harmonic analysis and unitary representations of $\mathrm{SO}^+(d,1)$. The key observation is that, if one works with the naive category $\Emb$ of Riemannian manifolds and conformal open
embeddings, then unbounded operators arise naturally in the construction of examples,
as one would expect from the G{\aa}rding--Wightman axioms. 
For $d\geq 3$, we characterize the
boundedness of these operators by a geometric condition on conformal
embeddings (Theorem~\ref{thm_unbounded}).

This leads us to replace $\Emb$ by the category $\Embc$ of Riemannian germs and
conformal open embeddings, where morphisms are required to extend to neighborhoods
of the cores. At the same time, on the target side, we take the category $\cC$ in
\eqref{intro_disk_add} to be the ind-category $\Ind$ of the category $\Hilb$ of
separable Hilbert spaces. With these choices, we construct conformally flat $d$-disk
algebras in $\Ind$ (Theorem~\ref{thm_CF}).

The underlying vector spaces of these algebras admit natural Hilbert space completions.
A natural question is whether the algebra structure extends to bounded operations on these Hilbert spaces. We prove that, for $d\geq 3$, such extensions do not exist: some operadic products are unbounded for the Hilbert space norm (Theorem~\ref{thm_app_unbounded}).
Thus, in dimensions $d\geq 3$, the construction requires refinements on both sides:
the source category is replaced by germs around cores, and the target category by
$\Ind$. Together, these refinements make it possible to incorporate the unbounded
operations of analytic QFT into a geometric and categorical framework.



This phenomenon is particularly interesting in comparison with dimension two.
Conformally flat $2$-disk algebras in ind-Hilbert spaces have been constructed using
Bergman spaces \cite{MBergman}, while the unboundedness obstruction exhibited in this
paper for $d\geq 3$ does not arise from the same argument in dimension two. This is consistent with Segal's formulation, where two-dimensional conformal field theories are symmetric monoidal functors from a bordism category of Riemann surfaces to Hilbert spaces \cite{Segal}, and important examples have been constructed in \cite{Te,GKRV2}. Together, these facts suggest that the boundedness of conformally flat $d$-disk algebras with respect to Hilbert space norms may behave differently in dimension two and in dimensions $d\geq 3$.


For a general conformally flat $d$-disk algebra $A$, we prove that the
($1$-categorical) left Kan extension along $\Disk\hookrightarrow \Embc$ defines a
symmetric monoidal functor
\begin{align*}
\mathrm{Lan}_A:\Embc \rightarrow \Ind,
\end{align*}
and hence yields metric-dependent invariants of conformally flat Riemannian manifolds
(Theorem~\ref{prop_ind_CF}). Under suitable positivity and continuity assumptions on
$A$, its value on the standard sphere,
$\mathrm{Lan}_A(S^d,g_\std)$, recovers the sphere partition function of the associated
CFT (Theorem~\ref{cor_left_nonzero}).

The rest of the introduction explains these constructions in more detail. We first describe
the conformal-geometric disk category and its left Kan extension, then introduce Hilbert
space filtrations and positivity conditions, and finally discuss the explicit examples from
harmonic analysis.

\vspace{3mm}

\noindent
\begin{center}
{\bf 0.1. Riemannian manifolds, conformal embeddings, and disk algebras}
\end{center}

As a conformal-geometric analogue of factorization homology \eqref{eq_intro_left},
we consider the category $\Emb$ whose objects are oriented $d$-dimensional
Riemannian manifolds without boundary and whose morphisms are
orientation-preserving conformal open embeddings.
Let $\bD_d=\{x\in\R^d \mid |x|<1\}$ be the unit disk equipped with the standard
Riemannian metric $g_{\std}=dx_1^2+\cdots+dx_d^2$, and denote by $\tDisk$
the full monoidal subcategory of $\Emb$ generated by the unit disk.

\begin{figure}[h]
\begin{minipage}[c]{.44\textwidth}
\centering
    \includegraphics[scale=0.3]{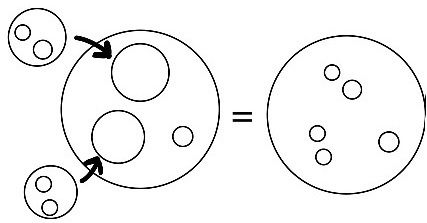}
    \caption{$\CE$-operad}\label{fig_intro_operad}
\end{minipage}
\begin{minipage}[l]{.50\textwidth}
    \centering
    \includegraphics[scale=0.18]{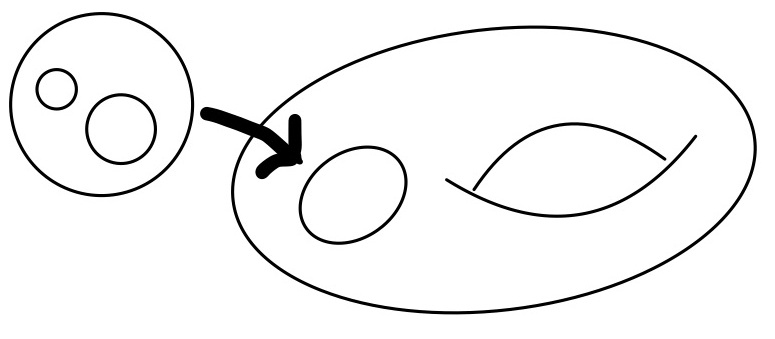}
    \caption{Kan extension}\label{fig_intro_module}
\end{minipage}
\end{figure}
This category defines an operad of conformal open embeddings of disks
(Fig. \ref{fig_intro_operad})
\begin{align*}
\CEc(n)
&=\mathrm{Hom}_{\Emb}(\sqcup_n \bD_d,\bD_d),\qquad\qquad\qquad (n \geq 0)\\
&=\left\{\phi:\underbrace{\bD_d \sqcup \cdots \sqcup \bD_d}_n \rightarrow \bD_d
\;\middle|\;\text{orientation-preserving conformal open embeddings} \right\},
\end{align*}
which gives a naive definition of a $d$-disk operad in conformally flat geometry (to be refined later).

%

In two dimensions, orientation-preserving conformal maps and holomorphic maps with nonzero derivatives are equivalent, and $\CEct(1)$ is the infinite-dimensional monoid of all injective holomorphic maps from disk to disk.
On the other hand, for $d \geq 3$, Liouville's theorem \cite{Liouville} implies that any conformal map from disk to disk can be expressed as a restriction of the group $\Conf(S^d) \cong \mathrm{O}^+(d+1,1)$, which consists of all conformal diffeomorphisms on the standard sphere.
In particular, for $d \geq 3$, the monoid $\CEc(1)$ can be identified as
$\CEc(1)=\{g \in \SO^+(d+1,1)\mid g(\bD_d) \subset \bD_d\}$, and $\CEc(n)$ consists of 
all conformal disk embeddings whose images are mutually disjoint
\begin{align}
\CEc(n)
= \{(f_1,\dots,f_n) \in \CEc(1)^n \mid f_i(\bD_d)\cap f_j(\bD_d)=\emptyset \}.
\label{eq_intro_geom}
\end{align}
Note that \eqref{eq_intro_geom} allows configurations such as $\overline{f_i(\bD_d)}\cap\overline{f_j(\bD_d)} \neq \emptyset$.
Our key observation is that, in the natural construction of a $\CEc$-algebra associated with the free scalar field,
such configurations give rise to unbounded operators on Hilbert spaces (Theorem \ref{thm_unbounded}).


To address this issue, we introduce a category $\Embc$ whose objects are triples $(U,g,\overline{M})$, where $\overline{M}$ is a compact Riemannian manifold which may have boundary (called the \emph{core}),
and $U$ is an open Riemannian neighborhood of $\overline{M}$, \emph{a germ of a Riemannian manifold}.
Morphisms in $\Embc$ are given by orientation-preserving conformal open embeddings
that are defined on neighborhoods of the cores.
The germ of the closed disk $\overline{\bD}_d$ in flat Euclidean space defines an object of $\Embc$.
We denote by $\Disk$ the full monoidal subcategory of $\Embc$ generated by $\overline{\bD}_d$.
The morphisms in $\Disk$ refine \eqref{eq_intro_geom} to
\begin{align}
\CE(n)
= \Hom_{\Embc}(\sqcup_n \overline{\bD}_d,\overline{\bD}_d)
=
\{(f_1,\dots,f_n) \in \CE(1)^n \mid \overline{f_i(\bD_d)}\cap \overline{f_j(\bD_d)} = \emptyset \}.
\label{eq_intro_geom2}
\end{align}

Here, $\CE(1)$ coincides with $\CEc(1)$ for $d \geq 3$, but for $d=2$, it becomes genuinely smaller because an injective holomorphic map $f:\bD_2 \rightarrow \bD_2$ may not extend to a neighborhood of $\overline{\bD}_2$.
Our key result is that for $d \geq 3$ and for the basic examples, the necessary and sufficient condition for the operators associated to elements of $\CEc(n)$ to be bounded is precisely when they lie in $\CE(n) \subset \CEc(n)$ (Theorem \ref{thm_unbounded}). (In dimension two this situation is more subtle; see \cite{MBergman}.)
We call  this operad a \textbf{conformally flat d-disk operad}.
(Our definition is also motivated by the definition of functorial CFT by Stolz and Teichner;
see Remark~\ref{rem_cobordism}.)

Let $\cC$ be a cocomplete symmetric monoidal category with a colimit distributive tensor product.
A $\CE$-algebra in $\cC$ is equivalently a symmetric monoidal functor $\Disk \rightarrow \cC$,
and its (1-categorical) left Kan extension
\begin{equation}
\label{eq_intro_left2}
\begin{tikzcd}
\Disk
  \arrow[r,"A"]
  \arrow[d,hook,"i"']
& 
\mathcal{C}
\\
\Embc
  \arrow[ru,dashed,"\mathrm{Lan}_{A}"']
& 
\end{tikzcd}
\end{equation}
defines a symmetric monoidal functor $\mathrm{Lan}_A: \Embc \rightarrow \cC$ (see Proposition~\ref{prop_Left_Kan}).
In particular, this yields invariants for conformally flat manifolds.
%
When $\cC$ is the category of vector spaces $\Vect$, a $\CE$-algebra gives a collection
of products on a vector space $A$ compatible with the operadic composition (Fig. \ref{fig_intro_operad})
\begin{align}
\rho_n: \CE(n) \rightarrow \Hom_\Vect(A^{\otimes n},A),
\qquad\qquad (n \geq 0).
\label{eq_intro_product}
\end{align}
In this case, for each object $\overline{M} \in \Embc$, giving
$\chi \in \mathrm{Hom}_\Vect\bigl(\mathrm{Lan}_A(\overline{M}),\C\bigr)$
is equivalent to giving a family of maps
\begin{align}
\chi_n:
\mathrm{Hom}_{\Embc}(\sqcup_n \overline{\bD},\overline{M})
\rightarrow
\Hom_{\Vect}(A^{\otimes n},\C),
\qquad\qquad (n \geq 0),
\label{eq_intro_state}
\end{align}
which is compatible with the $\CE$-operad action
(Fig. \ref{fig_intro_module})
\begin{align}
\CE(n) \times
\mathrm{Hom}_{\Embc}(\sqcup_m \overline{\bD},\overline{M})
\rightarrow
\mathrm{Hom}_{\Embc}(\sqcup_{n+m-1} \overline{\bD},\overline{M}).
\label{eq_intro_ope}
\end{align}
We call such a collection $\chi$ a \emph{state} of $A$ on $\overline{M}$.
When $\overline{M}$ has no boundary, we expect a state to recover the \emph{partition function} on $\overline{M}$ in physics,
while for manifolds with boundary, states should encode the boundary degrees of freedom of the CFT.

In the theory of vertex operator algebras, the compatibility condition
\eqref{eq_intro_ope} is naturally interpreted as compatibility with the operator product expansion,
appearing in the definition of conformal blocks \cite{TUY,Zh}.
From this viewpoint, the left Kan extension \eqref{eq_intro_left2} can be viewed as
a Riemannian-geometric analogue of conformal blocks and chiral homology \cite{BD},
without {relying on holomorphicity}.

\vspace{3mm}

\noindent
\begin{center}
{0.2. \bf 
$\CE$-algebras with Hilbert space filtrations
}
\end{center}

A \emph{Hilbert space filtration} of a $\CE$-algebra $A$ in $\Vect$ is a family of subspaces
\begin{align*}
A = \cup_{k=0}^\infty H^k, \qquad H^0 \subset H^1 \subset H^2 \subset \dots
\end{align*}
such that each $H^k$ is a Hilbert space and the inclusion $H^k \subset H^{k+1}$ is isometric, and
each $H^k$ is preserved by the action of the monoid $\CE(1)$, and
for any $n \geq 1$ and $k_1,\dots,k_n \geq 0$, there exists $K \geq 0$ such that the restriction of \eqref{eq_intro_product} to $H^{k_1}\otimes \cdots \otimes H^{k_n} \subset A^{\otimes n}$ has image contained in $H^K$, and defines a bounded operator between Hilbert spaces
\footnote{
In Hilbert spaces, a continuous bilinear map $H_1 \times H_2 \rightarrow H$ and
a bounded operator $H_1 \hotimes H_2 \rightarrow H$ from the completed tensor product are not the same notion.
In this paper we adopt the latter, as required by the definition of the tensor product in the category of ind-Hilbert spaces $\Ind$ (see Definition \ref{def_ind_CF}).}.

Let $\Hilb$ be the category of separable Hilbert spaces and bounded operators, and let $\Ind$ be its ind-category.
A $\CE$-algebra $A$ equipped with a Hilbert space filtration naturally determines a symmetric monoidal functor
$\Disk \rightarrow \Ind$.
The filtration is important because it allows one to impose continuity on \eqref{eq_intro_state} and thereby
to introduce natural notions such as continuous states and closed ideals of $A$.
We call a $\CE$-algebra \emph{simple} if it has no nontrivial closed ideals.

We are interested in when the continuous states (partition functions) obtained via left Kan extension are nontrivial.
Given a continuous state $\chi_M$ on $\overline{M}$, define its \emph{radical} $N(\chi_M)$ to be the subspace of $A$
consisting of those $v \in A$ such that
\begin{align*}
\chi_M(\iota)(a_1,\dots,a_{n-1},v)=0
\end{align*}
for any $n \geq 1$, $\iota \in \mathrm{Hom}_{\Embc}(\sqcup_n \overline{\bD},\overline{M})$ and
$a_1,\dots,a_{n-1} \in A$.
By continuity of the state, $N(\chi_M)$ is a closed ideal. Hence, if $A$ is simple then any non-zero continuous state on a conformally flat manifold has trivial radical (Proposition \ref{prop_null_ideal}).
%
%
Among the states of $A$ on $\overline{\bD}$ there is a distinguished one, called the \emph{vacuum expectation value} in physics.
Under additional assumptions to be stated later, simplicity of $A$ can be determined by the triviality of the radical of the vacuum expectation value (Proposition \ref{prop_criterion_simple}).
\vspace{3mm}

\noindent
\begin{center}
{0.3. \bf Reflection positivity in conformal field theory}
\end{center}

Here we explain why we choose the unit disk as the local geometric model.
In topological field theory, unitarity can be incorporated by equipping the
cobordism category with a dagger structure \cite{Atiyah,Steh2}.
In our setting, it is natural not only to take the target category to be
$\Ind$, but also to impose additional positivity conditions.

In Euclidean QFT, the natural notion of positivity is
Osterwalder--Schrader reflection positivity \cite{OS1,OS2}, which is defined by choosing an involution of $\R^d$,
\begin{align*}
\R^d \rightarrow \R^d,\quad (x_1,x_2,\dots,x_d) \mapsto (-x_1,x_2,\dots,x_d).
\end{align*}
The choice of involution is non-canonical; however, any two such choices are equivalent by
$\SO(d)$-symmetry.
On the other hand, for vertex operator algebras (VOA) \cite{B1,FLM}, which provide a formulation
of two-dimensional chiral CFT
(see also \cite{FMS} from the physics viewpoint),
 unitarity is defined by the involution
$z \mapsto \bar z^{-1}$ \cite{FHL,Li,DL}.
For a unitary VOA $V$, its underlying vector space admits a Hilbert space completion
$\overline{V}^{\mathrm{Hilb}}$, and Carpi-Kawahigashi-Longo-Weiner showed that
(Minkowskian) Wightman fields on $S^1=\{z \in \C \mid z = \z^{-1} \}$ act on $\overline{V}^{\mathrm{Hilb}}$ \cite{CKLW}.
Moreover, Adamo-Tanimoto and the author showed that Osterwalder--Schrader reflection positivity
corresponds to the unitarity of full VOAs \cite{AMT}.
Under the Wick rotation $t\mapsto i\tau$, Minkowski spacetime
$e^{it+i\theta}\in S^1\times S^1$ is mapped to $e^{-\tau+i\theta}\in \bD_2$
\cite{AGT,AMT}.
From this perspective, our local model $\Disk$ generated by $\bD_d$
 is naturally associated with the reflection positivity of the inversion
\begin{align}
J_{\mathrm{RP}}:\R^d \rightarrow \R^d,\quad (x_1,x_2,\dots,x_d) \mapsto
\left(\frac{x_1}{\norm{x}^2},\frac{x_2}{\norm{x}^2},\dots,\frac{x_d}{\norm{x}^2}\right).
\label{eq_intro_reflection}
\end{align}
The choice of reflection positivity \eqref{eq_intro_reflection} also suggests that
Hermitian conditions should be imposed on the products of a $\CE$-algebra.

In this paper, we impose the following conditions on the monoid homomorphism
$\rho_1:\CE(1)\rightarrow \mathrm{End}\,A$:
\begin{enumerate}
\item[(U)]
The restriction of $\rho_1$ to 
$\{f\in \CE(1)\mid J_{\mathrm{RP}} f J_{RP}^{-1} =f \}\cong \SO^+(d,1)$
defines a strongly continuous unitary representation on $H^k$ for each $k\geq 0$.
\item[(D)]
Dilations, shrinking the disk by a factor $r\in(0,1)$, define a strongly continuous,
self-adjoint, contractive semigroup representation by Hilbert-Schmidt operators.\footnote{
The assumption that these operators are Hilbert--Schmidt is quite strong; physically,
it corresponds to restricting attention to so-called \emph{compact CFTs} (see also \cite{M1}).}
\end{enumerate}
Note that the dilation $D(r)$ satisfies $J_{\mathrm{RP}} D(r) J_{RP}^{-1} = D(r^{-1})$.
By condition (D), the $\CE$-algebra $A$ decomposes into the eigenspaces for dilations  (Proposition \ref{prop_semigroup}),
\begin{align}
A=\bigoplus_{\Delta\geq 0} A_\Delta\qquad \text{in $\Ind$},
\label{eq_intro_eigenspace}
\end{align}
where each $A_\Delta$ is finite-dimensional (in CFT, $\Delta$ is
called the conformal dimension).

We prove that if $A$ is a $\CE$-algebra equipped with a Hilbert space filtration satisfying (U) and (D) with $\dim A_0=1$, then the left Kan extension to the sphere $(S^d,g_\std)$ satisfies
\begin{align}
\dim \Hom_{\Ind}(\mathrm{Lan}_A(S^d,g_\std),\C)^{\Conf^+(S^d,g_\std)} = 1.
\label{eq_intro_coh}
\end{align}
Here \eqref{eq_intro_coh} denotes the fixed subspace under the action of the
orientation-preserving conformal diffeomorphism group $\Conf^+(S^d,g_\std)$.
In other words, there exists a unique $\Conf^+(S^d,g_\std)$-invariant partition function on the standard sphere $S^d$,
and it is obtained via left Kan extension.

The eigenspace decomposition \eqref{eq_intro_eigenspace} also plays a role in the construction of
correlation functions. For $n\geq 1$, let
\begin{align*}
\Conf_n(\R^d)= \{(x_1,\dots,x_n) \in (\R^d)^n \mid x_i \neq x_j \text{  for  }i\neq j\}
\end{align*}
be the configuration space of $n$ ordered points.
The products of a $\CE$-algebra give rise to multilinear maps
\begin{align}
Y_{x_1,\dots,x_n}:(A^\alg)^{\otimes n}\longrightarrow \overline{A^\alg},
\qquad (x_1,\dots,x_n)\in \Conf_n(\R^d),\label{intro_Y}
\end{align}
where $A^\alg=\bigoplus_{\Delta\geq 0}^\alg A_\Delta$ and
$\overline{A^\alg}=\prod_{\Delta\geq 0}A_\Delta$ are the algebraic direct sum and the direct product of vector spaces, respectively.
After pairing with elements of $A^\alg$, these maps give correlation functions on
configuration spaces.
We show that the boundedness of the operadic products implies a clustering property: when one
configuration is inserted into another in a separated region, the corresponding
composition of \eqref{intro_Y} is expressed by an absolutely convergent sum with respect to the eigenspace
decomposition (Theorem~\ref{thm_Y}).

\vspace{3mm}

\noindent
\begin{center}
{0.4. \bf 
Example from harmonic analysis
}
\end{center}

We now turn to the examples, which form the central part of the paper. For $d\geq 3$,
we construct explicit $\CEd$-algebras using harmonic polynomials and reproducing kernel
Hilbert space methods. The key point is that the higher operations, which are a priori
unbounded from the viewpoint of QFT, become bounded on each stage
of the Hilbert space filtration precisely under the geometric separation condition encoded
in $\CEd$.

The vector space of harmonic polynomials in $d$ variables admits a natural inner
product, and its Hilbert space completion will be denoted by $\Hf$. This space has a
realization as a reproducing kernel Hilbert space on the disk, and carries a unitary
representation of $\mathrm{SO}^+(d,1)$.
Using the reproducing kernel structure, we obtain a representation of the monoid $\CE(1)$
satisfying conditions (U) and (D),
\begin{align*}
\rho_1: \CE(1) \rightarrow \mathbb{B}(\Hf),
\end{align*}
where $\mathbb{B}(\Hf)$ denotes the $C^*$-algebra of bounded linear operators on $\Hf$ (Proposition \ref{prop_unitary}).

The nontrivial part of the construction lies in defining the higher products and proving their boundedness.
For any conformal disk embedding $(f_1,f_2) \in \CE(2)$, we construct a bounded operator
from the Hilbert space tensor product,
\begin{align}
C_{f_1,f_2}:\Hf \hat{\otimes} \Hf \rightarrow \R,
\label{eq_intro_contraction}
\end{align}
which satisfies, for any $g,h_1,h_2 \in \CE(1)$,
\begin{align}
C_{g f_1, g f_2}=C_{f_1,f_2}, \qquad
C_{f_1 h_1, f_2 h_2}
= C_{f_1,f_2}\circ (\rho_1(h_1)\otimes \rho_1(h_2))
\label{eq_intro_cont_inv}
\end{align}
(see Theorem \ref{thm_wick_contraction}).
The construction proceeds as follows:
(I) we define $C_{f_1,f_2}$ formally as a linear map on a dense subspace of $\Hf$;
(II) we show that, when $f_1$ and $f_2$ are given by compositions of dilations and translations,
this operator admits a bounded closure;
(III) using \eqref{eq_intro_cont_inv}, we extend boundedness to arbitrary elements of $\CE(2)$.

Step (II) is the key observation: it shows that operators which are a priori unbounded become  bounded under the geometric separation conditions \eqref{eq_intro_geom2}. On the other hand, for an element of the naive operad $(f_1,f_2) \in \CEc(2) \setminus \CE(2)$ in \eqref{eq_intro_geom}, we prove that $C_{f_1,f_2}$ is an unbounded operator (Theorem \ref{thm_unbounded}). These results are proved by directly estimating the norm of $C_{f_1,f_2}$ using a combinatorial description of the inner product on harmonic polynomials (Theorem \ref{thm_bounded}).
Using this geometric generalization of Wick contraction, we obtain
a nontrivial $\CE$-algebra structure
\begin{align*}
\rho_n^\Psi:\CE(n) \rightarrow 
\mathrm{Hom}_{\Ind}\left((\Sym\Hf)^{\otimes n},\Sym\Hf\right)
,\qquad n\geq 0
\end{align*}
on $\Sym \Hf = \bigoplus_{p\geq 0}^\alg \Sym^p \Hf$
(Theorem \ref{thm_CF}).
Here $\Sym \Hf$ is the algebraic direct sum of the Hilbert spaces $\Sym^p\Hf$ equipped with the Hilbert space filtration
$\Sym^{\leq k}\Hf=\bigoplus_{p=0}^k\Sym^p\Hf$.
It is not complete as a Hilbert space
and its natural Hilbert space completion is the Fock space
$\mathcal{F}_\Hf=\widehat{\bigoplus}_{p\geq 0}\Sym^p\Hf$.

We then ask whether the operadic products extend to bounded operators on this Hilbert space completion.
Let $B_{r_i}(a_i)$, $i=1,\ldots,N$, be mutually disjoint balls contained in $\bD$, and 
define \(g_i:\bD\to \bD\) by \(g_i(x)=r_i x+a_i\). Then \(\phi=(g_1,\ldots,g_N)\in\CE(N)\).
We show that if the corresponding operation $\rho_N^\Psi(\phi):\Sym \Hf^{\otimes N} \rightarrow \Sym \Hf$ extended to a bounded operator on $\mathcal{F}_\Hf^{\hat{\otimes}N}$, then one would have
\begin{align}
\sum_{i, j=1}^N \frac{r_i^{\frac{d-2}{2}}r_j^{\frac{d-2}{2}}}{(1-2(a_i,a_j)+\norm{a_i}^2\norm{a_j}^2)^{\frac{d-2}{2}}}
+
\sum_{i\neq j}
\frac{r_i^{\frac{d-2}{2}}r_j^{\frac{d-2}{2}}}{\norm{a_i-a_j}^{d-2}}
\leq N.
\label{intro_inequality1}
\end{align}
For $d\geq 3$ this inequality is violated by suitable configurations with $N$ sufficiently large. This gives an explicit obstruction to extending the $\CE$-algebra structure to the Hilbert space completion (Theorem \ref{thm_app_unbounded}).

It is worth noting that the analogous inequality arising from the two-dimensional $\CEt$-algebras is a Grunsky--Nehari type inequality in complex analysis, and holds for every $N\geq 2$ (see Remark \ref{rem_app_2d}). Thus the obstruction used here is absent in dimension two. 
The case $d=2$ is treated in \cite{MBergman}, where $\CEt$-algebras in Ind-Hilbert spaces are constructed using Bergman spaces. These
constructions are closely related to the factorization algebras associated with the
conformal Laplacian \cite{Mfactorization}.

\vspace{3mm}

The organization of this paper is as follows. Section 1 develops the geometric framework: we define the category of Riemannian germs and conformal embeddings, the conformally flat disk operad $\CE$, and the notion of a $\CE$-algebra with Hilbert space filtration. Section 2 constructs the concrete examples from harmonic analysis. After recalling the action of $\mathfrak{so}(d+1,1)$ on harmonic polynomials and proving the fundamental norm estimate, we construct the contraction operators and the corresponding $\CE$-algebras in $\Ind$. We also prove that, for $d\geq 3$, these examples do not extend to $\CE$-algebras on their natural Hilbert space completions. 
Section 3 studies left Kan extensions and continuous states, applies them to the construction of conformally invariant states on the sphere, and finally extracts correlation functions on configuration spaces from $\CE$-algebras, 
and proves their clustering properties.
Appendix~\ref{app_remark} collects basic results on the action of $\Conf^+(S^d)$ on $S^d$, while
Appendix~\ref{app_Hilb} reviews standard facts about the category of Hilbert spaces and its ind
category.

\vspace{3mm}
\begin{center}
\textbf{\large Notations}
\end{center}

\vspace{3mm}
We will use the following notations:
\begin{itemize}
\item[$\text{[}n\text{]}$:] the finite set $\{1,2,\dots,n\}$ for $n \geq 1$.
\item[$(M,g)$:] A pair of a manifold $M$ and a Riemannian metric $g$.
\item[$\Vect$:] the category of $\C$-vector spaces.
\item[$\Hilb$:] the category of Hilbert spaces with bounded linear operators, Appendix \ref{app_Hilb}
\item[$\Ind$:] the category of ind-objects in $\Hilb$, Appendix \ref{app_Hilb}
\item[$\otimes$:] the algebraic tensor product of vector spaces.
\item[$\hotimes$:] the tensor product of Hilbert spaces, Appendix \ref{app_Hilb}.
\item[$T_a$:] the translation $x \mapsto x+a$ for $a \in \R^d$, \S \ref{sec_conf_sphere}
\item[$r^D =D(r)$:] the dilation $x \mapsto rx$ for $r\in (0,1)$, \S \ref{sec_conf_sphere}
\item[$P_\mu,J_{\mu\nu}, K_\mu,D$:] a basis of the Lie algebra $\sod$, \S \ref{sec_conf_sphere}
\item[$V_\Delta$:] the eigenspace of the dilation operators with eigenvalue $r^\Delta$, \S\ref{sec_ind_Hilb}.
\item[$V^\alg$:] the algebraic core $\bigoplus_{\Delta\geq0}^\alg V_\Delta$, \S\ref{sec_ind_Hilb}.
\end{itemize}

\vspace{5mm}

\section{Operad in conformally flat geometry}\label{sec_operad}
In this section we introduce an operad that serves as a local model for conformally flat geometry, and we define its algebras in ind Hilbert spaces.
Throughout this paper, by a manifold we mean a smooth, oriented, second countable manifold, and we write $(M,g)$ for a manifold $M$ equipped with a Riemannian metric $g$. We also assume that the dimension of a manifold is always $d \geq 2$.

In Section \ref{sec_category}, we introduce and study the symmetric monoidal category $\Embc$ of conformal Riemannian manifolds, its full subcategory $\Disk$ generated by disks, and the associated left Kan extensions.
In Section \ref{sec_conf_sphere}, we describe explicitly the conformal transformation group of the standard sphere and investigate its properties. 
Although much of the material in Section \ref{sec_conf_sphere} is standard, we provide self-contained proofs since it plays an important role in this paper. 
Some propositions are proved in Appendix \ref{app_remark}.
In Section \ref{sec_def_CF}, based on Liouville's theorem, we discuss some basic properties of $\Disk$-algebras.
Finally, in Section \ref{sec_ind_Hilb} we introduce the notion of a $\Disk$-algebra equipped with a Hilbert space filtration and impose several additional assumptions on the induced monoid action.
Under these assumptions, we study the behavior of closed ideals.
The properties of the category of Hilbert spaces and of its ind-category are reviewed in Appendix \ref{app_Hilb}.


\subsection{Categories of Riemannian manifolds with conformal open embeddings}\label{sec_category}

\begin{dfn}\label{def_conformal_map}
A local diffeomorphism $f : (M_1,g_1) \rightarrow (M_2,g_2)$ between Riemannian manifolds is called a {\bf conformal map} if there exists a smooth function $\Om_f : M_1 \rightarrow \R$ with $\Om_f>0$ such that $f^*(g_2)= \Om_f^2 g_1$. The positive smooth function $\Om_f$ is called a \textbf{conformal factor}. 
\end{dfn}

For a Riemannian manifold $(M,g)$, set
\begin{align*}
\Conf(M,g) &= \{f \in \mathrm{Diff}(M) \mid f \text{ is a conformal map}\}\\
\Conf^+(M,g) &= \{f \in \mathrm{Diff}(M) \mid f \text{ is an orientation preserving conformal map}\}.
\end{align*}
Here $\mathrm{Diff}(M)$ is the group of all diffeomorphisms on $M$.

Let $\Emb$ be the category whose objects are oriented Riemannian manifolds $(M,g)$ of dimension $d$ without boundary (not necessarily compact), and whose morphisms are smooth open embeddings that are orientation-preserving conformal maps. We denote the morphisms of this category by $\Emb(M,N)$.
The category $\Emb$ is equipped with a symmetric monoidal structure given by disjoint union of manifolds with the empty set as the unit.


Let $g_\std=dx_1^2+\dots+dx_d^2$ be the standard metric on $\R^{d}$.
Set
\begin{align*}
\bD_d = \{(x_1,\dots,x_d) \in\R^d\mid ||x||<1 \},
\end{align*}
the unit disk equipped with the standard Riemannian metric $g_{\std}$.
When the dimension is clear, we often write $(\bD_d, g_\std)$ simply as $\bD$ for brevity.
Let $\tDisk$ denote the full subcategory of $\Emb$ whose objects are disjoint unions of disks $\{\sqcup_n \bD\}_{n \geq 0}$.
The category $\tDisk$ is the full subcategory generated as a monoidal category by the flat disk $\bD$.
\begin{dfn}\label{def_operad}
For $n \geq 0$, set
\begin{align*}
\CEc(n)=\Emb(\underbrace{\bD \sqcup \cdots \sqcup \bD}_n,\bD).
\end{align*}
The collection $\{\CEc(n)\}_{n \geq 0}$ has the structure of a symmetric operad via the operation of composing into the $i$-th component ($i \in \{1,\dots,n\}$)
\begin{align*}
\circ_i:\CEc(n) \times \CEc(m) \rightarrow \CEc(n+m-1)
\end{align*}
The element $(\id_\bD:\bD \rightarrow \bD)\in \CEc(1)$ is the unit, and 
for the empty disk embedding $*=(\emptyset \rightarrow \bD) \in \CE(0)$, the composition
\begin{align*}
\circ_i:\CEc(n) \times \CEc(0)=\{*\} \rightarrow \CEc(n-1),\quad\phi_{[n]} \mapsto \phi_{[n]}\circ_i *
\end{align*}
is the operation that forgets the $i$-th disk of $\phi_{[n]}$ and regards it as an element of $\CEc(n-1)$.
The permutation group $S_n$ acts on $\CEc(n)$ by permuting the disks.
\end{dfn}

For $\phi \in \CEc(n)$, since $\sqcup_{n} \bD$ is a disjoint union, $\phi$ can be
written as the disjoint union of maps $\phi_i: \bD \rightarrow \bD$ ($i=1,\dots n$). Hence,
\begin{align*}
\CEc(n) = \{(\phi_1,\dots,\phi_n) \in \CEc(1)^n \mid \phi_i(\bD) \cap \phi_j(\bD) = \emptyset \text{ for any }i\neq j \}.
\end{align*}

Note that $\phi \in \CEc(1) = \Emb(\bD,\bD)$ is not always surjective. For example, for $a\in\R^d$ and $r>0$, set
\begin{align}
B_r(a) = \{x \in \R^d \mid |x-a|<r\}.\label{eq_ball_trans}
\end{align}
If $B_r(a) \subset \bD$, then the map $\bD \rightarrow \bD$ defined by $x \mapsto rx+a$ is an element in $\CEc(1)$.
More generally, if $r_1,\dots,r_n >0$ and $a_1,\dots,a_n \in \bD$ satisfy $B_{r_i}(a_i) \subset \bD$ and
\begin{align}
B_{r_i}(a_i) \cap B_{r_j}(a_j) = \emptyset \qquad \text{for any }i\neq j\label{eq_sec1_disj}
\end{align}
then this determines an element of $\CEc(n)$.
The condition \eqref{eq_sec1_disj} allows the boundaries of the disks to have an intersection.

Our observation is that in the case of scalar free field, the structure maps defining a $\CEc$-algebra do not give bounded operators
for configurations in which the boundaries have  non-empty intersection.
On the other hand, in order to define a category, it is essential that the identity map
$\id_\bD:\bD \rightarrow \bD$ is included.
For this purpose, we consider the following category
(this definition is motivated by the cobordism category of Stolz--Teichner; see Remark \ref{rem_cobordism}).
\begin{dfn}\label{def_cat_compact}
An object of $\Embc$ is a triple $(U,g,M)$, where $(U,g)$ is an oriented $d$-dimensional Riemannian manifold without boundary, and $M \subset U$ is a compact closed subset equipped with a smooth structure as a $d$-dimensional manifold, possibly with boundary,  such that the inclusion $M \rightarrow U$ is a smooth embedding (see Remark \ref{rem_smooth_embedding}).

A morphism from $(U_1,g_1,M_1)$ to $(U_2,g_2,M_2)$ is an equivalence class of maps $f$, each defined on some neighborhood $V_f$ of $M_1$ in $U_1$,
$M_1 \subset V_f \subset U_1$, 
 where $f:(V_f,g_1|_{V_f}) \rightarrow (U_2,g_2)$ is an orientation-preserving conformal open embedding satisfying $f(M_1) \subset M_2$.
Two such morphisms $(f,V_f)$ and $(f',V_{f'})$ are said to be equivalent if there exists an open set $V \subset U_1$ with $M_1 \subset V \subset V_f \cap V_{f'}$ such that $f|_V = f'|_V$.
We define $\Embc$ to be the category whose morphisms are these equivalence classes $[f]$.
With disjoint union, this category becomes a symmetric monoidal category.
\end{dfn}

\begin{rem}
\label{rem_smooth_embedding}
Let $(U,g,M) \in \Embc$.
Then an interior point $p \in M \setminus \pa M$ has the same local coordinates as those of $U$.
If $p \in \pa M$ is a boundary point, then there exist an open neighborhood $V_p \subset U$ of $p$ and a local coordinate chart
$\phi:V_p \rightarrow \R^d$ of $U$ such that
$\phi: M \cap V_p \rightarrow \{(x_1,\dots,x_d) \in \R^d \mid x_d \geq 0 \} \cap \phi(V_p)$ and
$\phi: \pa M \cap V_p \rightarrow \{x_d=0\} \cap \phi(V_p)$
(codimension zero slice).
The smooth structure on $M$ is given in this way (see, for example, \cite[Theorem 5.51]{Lee}).
\end{rem}


The triple $(U,g,M)$ encodes the information of a neighborhood of the compact Riemannian manifold $M$.
We call $M$ the \textbf{ core}, and the equivalence class determines a Riemannian manifold germ around $M$ (see also \cite{Milnor}).
For example, for a compact Riemannian manifold $(M,g)$ without boundary, $(M,g,M)$ is an object of $\Embc$.
Also, $(\R^d,g_\std,\overline{\bD})$ defines an object of this category, where $\overline{\bD}=\{x\in\R^d\mid |x|\leq 1\}$.
The following proposition is immediate from the definition.
\begin{prop}\label{prop_category_mor}
Let $(U,g,M) \in \Embc$. Then the following statements hold in the category $\Embc$:
\begin{enumerate}
\item
If $M$ has no boundary, then $M \subset U$ is both open and closed, and $(U,g,M)$ is isomorphic to $(M,g|_M,M)$.
\item
If $M=\emptyset$, then $(U,g,\emptyset)$ is an initial object in $\Embc$.
\item
If $M$ is a compact Riemannian manifold without boundary, then the automorphism group of $(M,g,M)$ in $\Embc$
coincides with $\Conf^+(M,g)$.
\end{enumerate}
\end{prop}


Set $M^\circ = M \setminus \pa M$ to be the set of interior points of $M$.
Define a functor $F: \Embc \rightarrow \Emb$ by sending an object $(U,g,M)$ to the open submanifold $(M^\circ, g|_{M^\circ})$ of the core $M$, and by sending morphisms to their restrictions to the interior.
This defines a symmetric monoidal functor.
\begin{prop}\label{prop_faithful}
For any $d \geq 2$, $F: \Embc \rightarrow \Emb$ is a faithful functor.
\end{prop}

To prove this proposition, we need the following classical results.

Any two-dimensional Riemannian manifold without boundary admits local coordinates that are conformally equivalent to the flat disk $\bD_2$, called isothermal coordinates (see for example \cite{Che55}). An orientation-preserving smooth map $f:\bD_2 \rightarrow \bD_2$ is a conformal map if and only if $f$ is a holomorphic map whose derivative never vanishes. Hence, we have:
\begin{lem}\label{lem_Kahler}
An oriented Riemannian 2-manifold without boundary admits a unique complex structure compatible with the Riemannian metric and the orientation, and hence carries the structure of a K\"ahler manifold.
Moreover, for a local diffeomorphism $f:(M_1,g_1) \rightarrow (M_2,g_2)$ between oriented Riemannian 2-manifolds without boundary, $f$ is an orientation-preserving conformal map if and only if $f$ is a holomorphic map with nowhere vanishing derivative.
\end{lem}

%
The structure group of oriented conformal geometry is given by the direct product of the orthogonal group and dilations, $\mathrm{CO}(n) = \mathrm{SO}(n)\times \R_{>0}$.
The following result is called the \emph{2-rigidity} of conformal geometry, and it follows immediately from the fact that the prolongation of $\mathrm{CO}(n)$ vanishes at order two for $d\geq 3$ (see \cite[Example 2.6, Theorem  3.2 and Theorem 5.1 in Chapter 1]{Kobayashi}):
\begin{lem}\label{lem_rigid3}
Let $(M_i,g_i)$ be oriented $d$-dimensional Riemannian manifolds with $d \geq 3$. Let
$f_1,f_2: (M_1,g_1) \rightarrow (M_2,g_2)$ be orientation-preserving conformal local diffeomorphisms.
Assume that $M_1$ is connected. If $f_1$ and $f_2$ agree up to their $2$-jets at a point $p \in M_1$, then $f_1=f_2$.
\end{lem}


\begin{proof}[Proof of Proposition \ref{prop_faithful}]
Let $f_1,f_2 \in \Embc((U_1,g_1,M_1),(U_2,g_2,M_2))$.
We may assume that both $f_1$ and $f_2$ are conformal open embeddings defined on the same open subset $V \subset U_1$.
Note that if $M_1$ is empty, then by definition any $f_1$ and $f_2$ are equivalent.

Assume that $M_1$ is nonempty and that $f_1=f_2$ on $M_1^\circ$.
Let $p \in M_1 \setminus M_1^\circ$, and let $W_p$ be a connected open neighborhood of $p$.
Since $W_p \cap M_1^\circ$ is nonempty (see Remark \ref{rem_smooth_embedding}), it follows from Lemma \ref{lem_Kahler} or Lemma \ref{lem_rigid3} that $f_1|_{W_p} = f_2|_{W_p}$.
Therefore, there exists an open set $W \subset V$ with $M_1 \subset W$ such that $f_1|_W = f_2|_W$.
Hence $f_1 = f_2$ as morphisms in $\Embc$, and thus $F:\Embc \rightarrow \Emb$ is faithful.
%
%
\end{proof}

Let $\Disk$ denote the full subcategory of $\Embc$ whose objects are disjoint unions of disks $\{\sqcup_n (\R^d,g_\std,\overline{\bD_d})\}_{n \geq 0}$. As before, the collections
$\{\Embc (\sqcup_n (\R^d,g_\std,\overline{\bD}),(\R^d,g_\std,\overline{\bD}))\}_{n\geq 0}$ carry a symmetric operad structure.
Via the faithful symmetric monoidal functor $F:\Embc \rightarrow \Emb$, this operad becomes a suboperad of $\CEc$.
We denote it by $\CE$.
 The operads $\CEc$ and $\CE$ are analogues of the little disk operad in conformal geometry.
We call $\CE$ the \textbf{conformally flat d-disk operad}.
An explicit description of $\CE$ will be given in Section \ref{sec_def_CF}.
The operad $\CE$ is a genuinely smaller suboperad of $\CEc$; in particular, $F$ is faithful but not full.
\begin{rem}
\label{rem_cobordism}
Even if a diffeomorphism between the boundaries of two Riemannian manifolds is given,
there is no canonical way to glue the manifolds together while keeping track of the metrics.
In \cite{ST}, this issue is avoided by considering germs of Riemannian manifolds
(or, more generally, germs of $(G,X)$-structures with boundary).
Our definitions of cores and germs are motivated by \cite{ST}.
The relationship between cobordism and $\CE$-algebras is discussed in detail in \cite{Mvertex}.
\end{rem}

The following is the standard definition of an algebra over an operad (see \cite{LoV}):
\begin{dfn}\label{def_operad_algebra}
Let $\cC$ be a strict symmetric monoidal category. A $\CE$-algebra in $\cC$ is
an object $A \in \cC$ equipped with a sequence of maps
\begin{align*}
\rho_n:\CE(n) \rightarrow \Hom_\cC(A^{\otimes n},A)
\end{align*} 
such that:
\begin{enumerate}
\item
$\rho_1(\id_{\bD}) = \id_A$.
\item
For any $\phi_{[n]} \in \CE(n)$ and $\phi_{[m]} \in \CE(m)$ with $n \geq 1, m \geq 0$ and
$i \in \{1,\dots,n\}$,
\begin{align}
\rho_{n+m-1}(\phi_{[n]}\circ_i \phi_{[m]})=
\rho_{n}(\phi_{[n]})\circ_i \rho_m(\phi_{[m]}).\label{eq_def_of_operad}
\end{align}
\item
The map $\rho_n:\CE(n) \rightarrow \Hom(A^{\otimes n},A)$ is $S_n$-equivariant for all $n \geq 1$.
\end{enumerate}
\end{dfn}

In this paper we will mostly work in $\Embc$. In the following, when there is no risk of confusion, we will simply write $(U,g,M)$ as $M$.
Let $G:\Disk \rightarrow \cC$ be a symmetric monoidal functor. Then,
for any $\phi_{[n]} \in \CE(n)=\Embc(\sqcup_n \overline{\bD}_d,\overline{\bD}_d)$, $G$ assigns
\begin{align*}
G(\phi_{[n]}): G(\overline{\bD})^{\otimes n} \rightarrow G(\overline{\bD}).
\end{align*}
Hence, $A = G(\overline{\bD}) \in \cC$ is a $\CE$-algebra. Conversely, since
\begin{align*}
\Embc(\sqcup_n\overline{\bD}, \sqcup_m \overline{\bD}) \cong \sqcup_{i_1+\dots+i_m =n} \Pi_{k=1}^m \Emb(\sqcup_{i_k}\overline{\bD},\overline{\bD}),
\end{align*}
any $\CE$-algebra defines a symmetric monoidal functor in the obvious way.
\begin{prop}\label{prop_CF_algebra}
There is a one-to-one correspondence between $\CE$-algebras (resp. $\CEc$-algebra) in $\cC$ and symmetric monoidal functors $G: \Disk \rightarrow \cC$ (resp. $G:\tDisk \rightarrow \cC$).
Moreover, a $\CEc$-algebra naturally determines a $\CE$-algebra.
\end{prop}

The following proposition holds for the full subcategory $\tDisk$
 of $\Emb$ as well:
\begin{prop}\label{prop_Left_Kan}
Let $A: \Disk \rightarrow \cC$ be a symmetric monoidal functor. Assume that $\cC$ is cocomplete and the tensor product $\otimes_\cC$ distributes over colimits. Then, the left Kan extension of $A:\Disk  \rightarrow \cC$ along $\Disk \rightarrow \Embc$, denoted by
\begin{align*}
\mathrm{Lan}_A: \Embc \rightarrow \cC
\end{align*}
is a symmetric monoidal functor and the restriction of $\mathrm{Lan}_A$ to $\Disk$ is isomorphic to $A$.
Moreover, for any object $(U,g,M) \in \Embc$,
\begin{align*}
\mathrm{Lan}_A(U,g,M) = \mathrm{colim}_{\phi \in \CE/(U,g,M)} \tilde{A}(\phi).
\end{align*}
Here, $\CE/(U,g,M)$ is the comma category whose object is a map
$\phi: \sqcup_n \overline{\bD} \rightarrow (U,g,M)$ in $\Embc$ with $n \geq 0$ and morphism $f:(\phi_1: \sqcup_n \overline{\bD} \rightarrow M) \rightarrow (\phi_2: \sqcup_m \overline{\bD} \rightarrow M)$ is a map $\psi:\sqcup_n \overline{\bD} \rightarrow \sqcup_m \overline{\bD}$ in $\Disk$ such that $\phi_1 =\phi_2\circ \psi$.
The functor $\tilde{A}: \CE/(U,g,M) \rightarrow \cC$ is defined as the composition of the forgetful functor
$\CE/(U,g,M) \rightarrow \Disk$, which sends
$(\phi: \sqcup_n \overline{\bD} \rightarrow M)$ to $\sqcup_n \overline{\bD}$, with the functor
$A:\Disk \rightarrow \cC$.
\end{prop}
\begin{proof}
Since $\cC$ is cocomplete, the left Kan extension is obtained as a pointwise left Kan extension, and can be written as the explicit colimit described in the proposition.
\begin{align*}
\mathrm{Lan}_A(M_1 \sqcup M_2) \cong \mathrm{colim}_{\phi \in \CE/(M_1 \sqcup M_2)} \tilde{A}(\phi).
\end{align*}
Since $M_1 \sqcup M_2$ is disjoint, as categories,
\begin{align*}
\CE/(M_1 \sqcup M_2) \cong 
\CE/M_1 \times
\CE/ M_2
\end{align*}
holds. Since $A$ is symmetric monoidal and the tensor product of $\cC$ distributes over colimit, we obtain
\begin{align*}
\mathrm{colim}_{\phi \in \CE/(M_1 \sqcup M_2)} \tilde{A}(\phi)
&\cong \mathrm{colim}_{(\phi_1,\phi_2) \in \CE/M_1 \times \CE/M_2} \tilde{A}(\phi_1,\phi_2) \\
&\cong \mathrm{colim}_{(\phi_1,\phi_2) \in \CE/M_1 \times \CE/M_2} \tilde{A}(\phi_1)\otimes_\cC \tilde{A}(\phi_2)\\
&\cong \mathrm{colim}_{\phi_1 \in \CE/ M_1}\tilde{A}(\phi_1)\otimes_\cC  \mathrm{colim}_{\phi_2 \in \CE/M_2}\tilde{A}(\phi_2) \\
&\cong \mathrm{Lan}_A(M_1) \otimes_\cC 
\mathrm{Lan}_A(M_2).
\end{align*}
Hence $\mathrm{Lan}_A$ is symmetric monoidal.
\end{proof}

Recall that a Riemannian manifold $(M,g)$ is called \textbf{conformally flat} if for any $x \in M$ there is a smooth conformal open embedding $\phi:(\bD,g_\std) \rightarrow (M,g)$ such that $x \in \phi(\bD)$.
Denote by $\EmbFc$ the full subcategory of $\Embc$ consisting $(U,g,M)$ such that $(V,g|_V)$ is a conformally flat manifold for some open neighborhood $M \subset V$.
In $d=2$, every Riemannian manifold is conformally flat (see Lemma \ref{lem_Kahler}). 
On the other hand, when $d \geq 3$, there exist many Riemannian manifolds that are not conformally flat \cite{Kuiper,Kuiper2}.
For such a manifold, there are not enough disk embeddings, so one cannot access it via left Kan extension.
One may say that the above proposition essentially constructs a symmetric monoidal functor from the category of conformally flat manifolds $\EmbFc$.

The category of vector spaces $\Vect$ and the category of ind Hilbert spaces $\Ind$
are symmetric monoidal categories satisfying the assumptions of Proposition \ref{prop_Left_Kan} (see Appendix \ref{app_Hilb}).
The properties of left Kan extensions in these categories will be studied in more detail in
Section \ref{sec_left}.

\subsection{Conformal transformation on sphere}\label{sec_conf_sphere}
In this section, we give an explicit description of the conformal transformation group $\Conf^+(S^d)$ on the standard sphere $S^d$, and investigate a submonoid of $\Conf^+(S^d)$ that is closely related to $\CE(1)$.  
The part concerning the conformal transformation group is standard; see, for example, \cite{Matsumoto} and the references therein.


Let $d \geq 2$ and $\mathbb{R}^{d+1,1}$ be equipped with a nondegenerate symmetric bilinear form
$\langle\cdot,\cdot\rangle$ of signature $(d+1,1)$. Let $\{e_1,\dots,e_d,e_+,e_-\}$ be a basis of $\R^{d+1,1}$ with
\begin{align}
\langle e_+,e_- \rangle=1,\quad \langle e_i,e_j \rangle=\delta_{i,j},\quad \langle e_+,e_+ \rangle=\langle e_-,e_- \rangle=0= \langle e_\pm,e_i \rangle\label{eq_null_basis}
\end{align}
 for $i,j=1,\dots,d$.
Set
\[
\mathcal N \;=\; \{ v\in \mathbb{R}^{d+1,1}\setminus\{0\} \mid \langle v,v\rangle=0\}
\]
and let $\mathbb{P}(\mathcal N)$ be the projective space of $\mathcal{N}$.
Set
$
\mathcal N_+ = \{v \in \mathcal N \mid \langle v, e_+ - e_- \rangle > 0 \}$ and $\mathcal N_- = \{v \in \mathcal N \mid \langle v, e_+ - e_- \rangle < 0 \}$,
then $\mathcal N$ has two connected components, $\mathcal N = \mathcal N_+ \sqcup \mathcal N_-$.
Since $\mathrm{O}(d+1,1)$ preserves the bilinear form on $\R^{d+1,1}$, it acts on $\mathcal N$.
Set
\begin{align*}
\mathrm{O}^+(d+1,1) = \{g \in \mathrm{O}(d+1,1) \mid g(\mathcal N_+) = \mathcal N_+ \}.
\end{align*}
Then $\mathrm{O}^+(d+1,1)$ is an index-two subgroup of $\mathrm{O}(d+1,1)$.
The projective group $\PO(d+1,1)=\mathrm{O}(d+1,1)/\{\pm 1\}$ acts on the projective space $\mathbb{P}(\mathcal N)$. Since $-1$ sends $\mathcal N_+$ to $\mathcal N_-$, we have an identification
$\PO(d+1,1) \cong \mathrm{O}^+(d+1,1)$.
The purpose of this section is to study the action of $\mathrm{O}^+(d+1,1)$ on $\mathbb{P}(\mathcal N)$.

%

Define maps $\sigma_0,\sigma_\infty:\mathbb{R}^d\to\mathcal N_+$ by
\begin{align*}
\sigma_0(x)&= x- \frac{\norm{x}^2}{\sqrt{2}}e_+ + \frac{1}{\sqrt{2}}e_-,\qquad \sigma_\infty(x)= x- \frac{1}{\sqrt{2}}\,e_+ + \frac{\norm{x}^2}{\sqrt{2}}e_-
\end{align*}
for $x\in \R^d$.
It is clear that $\mathbb{P}(\mathcal{N}) \setminus \si_0 (\R^d)$ (resp. $\mathbb{P}(\mathcal{N}) \setminus \si_\infty (\R^d)$) consists of the projective class of $e_+$ (resp. $e_-$). We denote $[e_+] \in \mathbb{P}(\mathcal{N})$ by $\infty$.
Since 
\begin{align}
J(x)=\si_\infty^{-1}\circ \si_0(x) = \left(\frac{x_1}{\norm{x}^2},\dots, \frac{x_d}{\norm{x}^2} \right),
\label{eq_inversion_si}
\end{align}
$\mathbb{P}(\mathcal{N})= \R^d \cup \{\infty\}$ is a smooth $d$-dimensional manifold with the coordinates $\si_0,\si_\infty$,
which is the $d$-dimensional sphere.
Since $J^*(g_\std) = \frac{4}{\norm{x}^4} g_\std$, it follows that $\mathbb{P}(\mathcal{N})$ admits, up to a conformal factor, a unique  Riemannian metric with respect to which $\si_0$ and $\si_\infty$ are conformal maps.

\begin{rem}\label{rem_standard_sphere}
Let $(S^{d},g_{S^d})$ be the $d$-dimensional sphere, viewed as a submanifold of $(\R^{d+1},g_\std)$.
Let $P_\pm:\R^d \rightarrow S^d$ denote the stereographic projections, that is,
\begin{align*}
(x_1,\dots,x_d) \mapsto \left(\frac{2x_1}{1+|x|^2},\dots,\frac{2x_d}{1+|x|^2}, \pm \frac{1 - |x|^2}{1+|x|^2}\right).
\end{align*}
Since
$P_\pm^*(g_{S^d})= P_\pm^*(dy_1^2+\dots+dy_{d+1}^2) =\frac{4}{(1+|x|^2)^2} g_\std$, $P_\pm:(\R^d, g_\std) \rightarrow (S^d,g_{S^d})$ is a conformal map with conformal factor $\Om_{P_\pm}(x) = \frac{4}{(1+|x|^2)^2}$.
Moreover, since $P_-^{-1}\circ P_+ =J$, it follows that $\mathbb{P}(\mathcal{N})$ is conformally equivalent to $(S^d,g_{S^d})$.
\end{rem}

%
%
Let $g \in \mathrm{O}^+(d+1,1)$ and $x\in \R^d$.
If $g \si_0(x) \neq \infty$, then there exists a unique nonzero scalar $j_g(x) \in\mathbb{R}^\times$ and $y \in \R^d$
such that
\begin{equation}\label{eq:lift-null}
g\sigma_0(x)=j_g(x)\sigma_0(y).
\end{equation}
We write this $y$ as $g\cdot x$. This is the action of $\mathrm{O}^+(d+1,1)$ in the $\si_0$ coordinate system.
The map $g\cdot x$ is a rational polynomial of $x$, and $j_g(x)$ is a polynomial of degree at most $2$.

\begin{lem}\label{lem_norm_all_inv}
For any $g,h \in \mathrm{O}^+(d+1,1)$, 
\begin{align}
j_g(hx)j_h(x) &=j_{gh}(x) \label{eq_cocycle}\\
j_g(x)j_g(y) \norm{gx-gy}^2 &= \norm{x-y}^2
\end{align}
hold.
\end{lem}
\begin{proof}
The cocycle law \eqref{eq_cocycle} is immediate from \eqref{eq:lift-null}.
Since for any $x,y\in\mathbb{R}^d$ $
\big\langle \sigma_0(x), \sigma_0(y)\big\rangle  =(x,y)-\tfrac12\bigl(|x|^2+|y|^2\bigr)
=-\tfrac12\,|x-y|^2$,
we have $-\ft \norm{x-y}^2= \langle \si_0(x), \si_0(y)\rangle=
\langle g(\si_0(x)), g(\si_0(y))\rangle
 = -\ft j(g,x)j(g,y) \norm{gx-gy}^2$.
\end{proof}

Let us consider the pseudo-Riemannian metric $g_{\R^{d+1,1}} = dx_1^2+\dots+dx_d^2 + dx_+dx_-$ on $\R^{d+1,1}$.
Since $\si_0^*(g_{\R^{d+1,1}}) = g_\std$, it follows that the action of $\mathrm{O}^+(d+1,1)$ on $\mathbb{P}(\mathcal{N})$ is conformal, and its conformal factor is given by $\Om_g(x) = j_g(x)^{-1}$.
Conversely, any conformal diffeomorphism of $\mathbb{P}(\mathcal{N})$ is necessarily of this form (see Theorem \ref{thm_Liouville}):
\begin{prop}\cite{Liouville}\label{prop_Liouville_sphere}
If $d \geq 2$, then
\begin{align*}
\Conf(S^d,g_\std) \cong \mathrm{O}^+(d+1,1),\qquad \Conf^+(S^d,g_\std) \cong \SO^+(d+1,1).
\end{align*}
Moreover, for $g \in \mathrm{O}^+(d+1,1)$, in the coordinate system $\si_0$ one has
$g^*(g_\std) = j_g(x)^{-2}g_\std$.
\end{prop}


Set
\begin{align}
J =
\begin{pmatrix}
I_{d} & 0 & 0\\
0 & 0 & -1 \\
0 & -1 & 0
\end{pmatrix} \in \mathrm{O}^+(d+1,1).
\label{eq_I_def}
\end{align}
Then, $J \si_0(x) =x- \ft e_+ +  |x|^2 e_- = |x|^2 \si_0(\frac{x}{|x|^2})$. Hence, $J$ corresponds to the inversion $J(x)$ \eqref{eq_inversion_si} on $\R^d$. Since $\det J = -1$, 
\begin{align*}
\mathrm{O}^+(d+1,1) = \mathrm{SO}^+(d+1,1) \sqcup \mathrm{SO}^+(d+1,1) J
\end{align*}
and $\mathrm{SO}^+(d+1,1)$ is connected.
With respect to our basis \eqref{eq_null_basis}, the group $\SO^+(d+1,1)$ is generated by the following elements:

\begin{description}
\item[Rotation] For any $R\in\mathrm{SO}(d)$,
\[
\tilde{R}=
\begin{pmatrix}
R & 0 & 0\\
0 & 1 & 0\\
0 & 0 & 1
\end{pmatrix},
\qquad
\tilde{R}\,\sigma_0(x)=\sigma_0(Rx),
\qquad
j_{\tilde{R}}(x)=1.
\]
\item[dilation] For any $\lambda >0$,
\[
D(\lambda)=\lambda^D=
\begin{pmatrix}
I_d & 0 & 0\\
0 & \lambda & 0\\
0 & 0 & \lambda^{-1}
\end{pmatrix},
\qquad
D(\lambda)\,\sigma_0(x)=\lambda^{-1}\,\sigma_0(\lambda x),
\qquad
j_{D(\la)}\bigl(x\bigr)=\lambda^{-1}.
\]
\item[translation] For any $a\in\R^d$,
\[
T(a)=
\begin{pmatrix}
I_d & 0 & \sqrt{2}a\\
-\sqrt{2} a^\top & 1 & - |a|^2\\
0 & 0 & 1
\end{pmatrix},
\qquad
T(a)\,\sigma_0(x)=\sigma_0(x{+}a),
\qquad
j_{T(a)}\bigl(x\bigr)=1.
\]
\item[special conformal transformation] For any $c\in \R^d$,
\[
K(c)=
\begin{pmatrix}
I_d & \sqrt{2}c & 0\\
0 & 1 & 0\\
-\sqrt{2}c^\top & -|c|^2 & 1
\end{pmatrix},\quad\quad
j_{K(c)}\bigl(x\bigr)=1-2c\cdot x+|c|^2|x|^2.
\]
\[
K(c)\,\sigma_0(x)
=\bigl(1-2c\!\cdot\!x+|c|^2|x|^2\bigr)\,
\sigma_0\!\left(\frac{x-|x|^2 c}{\,1-2c\!\cdot\!x+|c|^2|x|^2\,}\right).
\]
\end{description}
Such actions are analogous to the action of $\mathrm{PSL}_2\C$ on $\CP$, and are called \textbf{M\"obius transformations}.
As an important property, hyperplanes and hyperspheres in $\R^d$ are mapped again to hyperplanes or hyperspheres by M\"obius transformations \cite{Matsumoto}.
 
Let $E_{\mu,\nu}$ denote the matrix whose $(\mu,\nu)$-entry is $1$ and whose remaining entries are $0$.
For $\mu,\nu \in \{1,\dots,d\}$, set
\begin{align*}
J_{\mu,\nu} = E_{\mu,\nu} - E_{\nu,\mu},\qquad
D &= E_{++}-E_{--},\\
P_\mu = \sqrt{2}E_{\mu,-} - \sqrt{2}E_{+,\mu},\qquad
K_\mu &= \sqrt{2}E_{\mu,+} -  \sqrt{2}E_{-,\mu},
\end{align*}
which are a basis of the Lie algebra $\sod$, which satisfies the following relations:
\begin{align}
\begin{split}
[D,K_\mu]&=  -K_\mu \\
[D,P_\mu]&= P_\mu \\
[D, J_{\mu,\nu}]&= 0 \\
[P_\mu, K_\nu] &= -2(\de_{\mu,\nu}D + J_{\mu,\nu}) \\
[J_{\mu,\nu},P_\rho]&=-\de_{\mu,\rho}P_\nu + \de_{\nu,\rho} P_\mu \\
[J_{\mu,\nu},K_\rho]&= -\de_{\mu,\rho}K_\nu + \de_{\nu,\rho} K_\mu  \\
[J_{\mu \nu}, J_{\rho,\si}]&=
\delta_{\nu,\rho}J_{\mu,\si}
- \de_{\mu,\rho}J_{\nu,\si}
-\de_{\nu,\si}J_{\mu,\rho}
+\de_{\mu,\si}J_{\nu,\rho}.
\end{split}
\label{eq_com_Lie}
\end{align}

Let $\R[x_1,\partial_1,\dots,
x_d,\partial_d]$ be the ring of differential operators on $\R^d$.
Set $$\norm{x}^2=x_1^2+\dots+x_d^2 \in \R[x_1,\partial_1,\dots,
x_d,\partial_d]$$
and $$E_x=\sum_{i=1}^d x_i \pa_i \in \R[x_1,\partial_1,\dots,
x_d,\partial_d],$$ the Euler operator.
Writing the infinitesimal action of $\mathrm{SO}^+(d+1,1)$ on the algebra of functions on $S^d$ in the coordinate system $\si_0$, we obtain the following Lie algebra homomorphism $d:\mathrm{so}(d+1,1) \rightarrow \R[x_1,\partial_1,\dots,
x_d,\partial_d]$:
\begin{align}
\begin{split}
d(D)&= - E_x \\
d(K_\mu) &= \norm{x}^2\pa_\mu -2x_\mu E_x \\
d(P_\mu) &= -\pa_\mu \\
d(J_{\mu,\nu}) &= x_\mu \pa_\nu-x_\nu \pa_\mu
\end{split}
\label{eq_differential_action}
\end{align}
for $\mu,\nu \in \{1,2,\dots,d \}$.
Set 
\begin{align*}
\fG_d = \{g \in \SO^+(d+1,1) \mid g(\bD) = \bD \},\\
\fS_d = \{g \in \SO^+(d+1,1) \mid g(\bD) \subset \bD \}.
\end{align*}
Note that $\fG_d$ is a group, while $\fS_d$ is a monoid which, for example, contains elements corresponding to \eqref{eq_ball_trans}.
In this section, we state several results concerning $\fG_d$ and $\fS_d$ that will be used in this paper. Proofs of these results are given in the appendix.

Note that \eqref{eq_I_def} defines an outer automorphism of order two
$\SO^+(d+1,1) \rightarrow \SO^+(d+1,1)$, given by $g \mapsto J g J^{-1}$, and
\begin{align*}
(J T(a) J)(x)
= \frac{x/|x|^2 + a}{|x/|x|^2 + a|^2}
= \frac{x + a|x|^2}{1 + 2(x,a) + |x|^2|a|^2}
= K(-a)(x).
\end{align*}
Let $\theta:\mathrm{so}(d+1,1) \rightarrow \mathrm{so}(d+1,1)$ denote the automorphism of the Lie algebra induced by $J$.
That is, $\theta$ is an involution of $\mathrm{so}(d+1,1)$ satisfying the following properties:
\begin{align}
\begin{split}
\theta(P_\mu)&= -K_\mu \\
\theta(K_\mu)&= -P_\mu \\
\theta(D)&=-D \\
\theta(J_{\mu,\nu})&=J_{\mu,\nu},
\end{split}
\label{eq_def_theta}
\end{align}
for $\mu,\nu \in \{1,2,\dots,d\}$.

\begin{prop}\label{prop_fixed_point}
The following properties hold:
\begin{enumerate}
\item
For $x \in \mathbb{P}(\mathcal{N})$, the condition $J(x)=x$ is equivalent to $\norm{x}=1$.
\item
The group $\{g \in \mathrm{O}^+(d+1,1) \mid Jg=gJ\}$ contains $\fG_d$ as an index $2$ subgroup and coincides with
$\fG_d \sqcup \fG_d J$.
Moreover, as a Lie group, $\fG_d$ is isomorphic to $\SO^+(d,1)$.
\item
If the restriction of $g \in \fG_d$ to $\partial \bD = \{x \in \R^d\mid \norm{x}=1 \}$ is the identity map, then $g = \id$.
\end{enumerate}
\end{prop}

\begin{prop}\label{prop_decomposition}
Let $g \in \fS_d$. Then, the following properties hold:
\begin{enumerate}
\item
If $g \notin \fG_d$, then there is $1>r>0$, $x_0 \in \bD$ and $h \in \fG_d$ such that
\begin{align*}
g = T_{x_0}r^D h.
\end{align*}
\item
If $\overline{g(\bD)} \subset \bD$, then there is $1>r>0$ and $h_1,h_2 \in \fG_d$ such that 
\begin{align*}
g = h_1 r^D h_2.
\end{align*}
\end{enumerate}
\end{prop}

\begin{prop}\label{prop_conf_identity}
For any $\phi \in \fS_d$ and $x,y \in \bD_d$,
\begin{align}
\norm{\phi(x)-\phi(y)}^2=\Om_\phi(x)\Om_\phi(y)\norm{x-y}^2.
\label{eq_identity_all}
\end{align}
Moreover, for any $\phi \in \fG_d$,
\begin{align}
 (1-2(\phi(x),\phi(y))+|\phi(x)|^2|\phi(y)|^2) =\Om_\phi(x)\Om_{\phi}(y) (1-2(x,y)+|x|^2|y|^2).
\label{eq_identity_G}
\end{align}
\end{prop}
The above identities are higher-dimensional analogues of the well-known identities for the linear fractional transformations $\mathrm{PSL}_2\C$:
\begin{align*}
(cz+d)(cw+d)\left(\frac{az+b}{cz+d} - \frac{aw+b}{cw+d}\right)&= (z-w)\\
((1-\bar{\al}z)(1-\al\bar{w})-(z-\al)(\bar{w}-\bar{\al}))&=(1-|\al|^2)(1-z\bar{w}),
\end{align*}
where $|\al|<1$, which implies that $\frac{z-\al}{1-\bar{\al}z} \in \mathrm{PSU}(1,1)=\fG_2$.


\subsection{Conformally flat d-disk operad}\label{sec_def_CF}
In this section we study the structure of the $\CE$-operad.
First, we recall the following result due to Liouville (see also Lemma \ref{lem_rigid3}):
\begin{thm}\cite{Liouville}\label{thm_Liouville}
Let $U,V \subset \R^d$ be connected open subsets equipped with the standard metric $g_\std$.
\begin{itemize}
\item
Assume $d \geq 3$. A smooth map $\phi:U \rightarrow V$ is an orientation-preserving conformal map if and only if $\phi$ is a restriction of an element in $\Conf^+(S^d)$.
\item
Assume $d=2$. A smooth map $\phi:U \rightarrow V$ is an orientation-preserving conformal map if and only if $\phi$ is a holomorphic map satisfying $\phi'(z)\neq 0$ for any $z\in U$.
\end{itemize}
\end{thm}

It follows immediately that
\begin{align*}
\CEc(1) =
\begin{cases}
\fS_d & \text{ for }d \geq 3\\
\{\phi:\bD_2 \rightarrow \bD_2 \mid \text{$\phi$ is an injective holomorphic map}\} & \text{ for }d =2.
\end{cases} 
\end{align*}
as semigroups. Hence, we have:
\begin{prop}\label{prop_CE_d_explicit}
For $d \geq 3$, we have $\CE(1)=\CEc(1)=\fS_d$, and
\begin{align*}
\CEc(n) &= \{(g_1,\dots,g_n) \in \fS_d^n \mid g_i(\bD_d) \cap g_j(\bD_d) =\emptyset \text{ for any } i\neq j\},\\
\CE(n) &= \{(g_1,\dots,g_n) \in \fS_d^n \mid \overline{g_i(\bD_d)} \cap \overline{g_j(\bD_d)} =\emptyset \text{ for any } i\neq j\}.
\end{align*}
The operadic composition $\circ_i:\CE(n) \times \CE(m) \rightarrow \CE(n+m-1)$ is given by
\begin{align*}
(g_1,\dots,g_n) \circ_i (f_1,\dots,f_m)
= (g_1,\dots,g_{i-1}, g_i \circ f_1, \dots, g_i\circ f_m, g_{i+1},\dots,g_n).
\end{align*}
\end{prop}
\begin{proof}
By Theorem \ref{thm_Liouville}, any element $f \in \CEc(1)$ is an element of $\fS_d$, and the extension $\tilde{f}\in\fS_d$ satisfies $\tilde{f}(\overline{\bD}) \subset \overline{\bD}$, we have $\tilde{f} \in \Embc(\overline{\bD},\overline{\bD})$ and $\CE(1)=\CEc(1)=\fS_d$.

Let $(g_1,\dots,g_n) \in \fS_d^n$ satisfy
$\overline{g_i(\bD_d)} \cap \overline{g_j(\bD_d)} =\emptyset$ for any $i\neq j$.
Then, for sufficiently small $\ep >0$, we have
$g_i(B_{1+\ep}(0)) \cap g_j(B_{1+\ep}(0)) = \emptyset$.
Hence,
\begin{align}
(g_1,\dots,g_n): \sqcup_n B_{1+\ep}(0) \rightarrow \R^d
\label{eq_ep_smooth}
\end{align}
is an open conformal embedding, and thus defines a morphism in $\Embc$.
Conversely, any element of
$\Embc(\sqcup_n (\R^d,g_\std,\bD), (\R^d,g_\std,\bD))$
satisfies
$\overline{g_i(\bD_d)} \cap \overline{g_j(\bD_d)} =\emptyset$.
\end{proof}

%
%
\begin{prop}\label{prop_CE_two}
The set $\CEt(1)$ consists of holomorphic maps $\phi:\bD_2 \rightarrow \bD_2$ such that there exist an open set $U$ with $\bD_2 \subset U$ and an injective holomorphic map $\tilde{\phi}: U \rightarrow \C$ such that $\tilde{\phi}|_{\bD_2}=\phi$.
In particular, $\CEt(1)$ is a proper submonoid of $\CEct(1)$. Moreover,
\begin{align}
\CEt(n)=\{(f_1,\dots,f_n) \in \CEt(1)^n \mid \overline{f_i(\bD)}\cap \overline{f_j(\bD)} = \emptyset \text{ for any }i\neq j \}.
\label{eq_CEt_exp}
\end{align}
Here the closures are taken in $\C$. Furthermore, if $f \in \CEt(1)$ satisfies
$\overline{f(\bD)} = \overline{\bD}$, then $f \in \fG_2$.
\end{prop}

\begin{proof}
The first statement and \eqref{eq_CEt_exp} follow from the proof of Proposition \ref{prop_CE_d_explicit}.
Note that since the intersection of $\overline{f_i(\bD)}$ and $\overline{f_j(\bD)}$ may occur on the boundary of $\overline{\bD}$, the closures must be taken in $\C$.
By the Riemann mapping theorem, there exists a biholomorphic map
$f:\bD \rightarrow \bD \setminus [0,1)$.
This map belongs to $\CEct(1)$, but it clearly cannot be extended to an injective holomorphic map on any neighborhood of $\overline{\bD}$.
Hence $\CEt(1) \subsetneq \CEct(1)$.
Let $f \in \CEt(1)$ satisfy $\overline{f(\bD)} = \overline{\bD}$, and denote by $\tilde{f}$ its extension to $\overline{\bD} \subset U$.
Since $\overline{f(\bD)}=\tilde{f}(\overline{\bD})=\overline{\bD}$ and
$\tilde{f}:\overline{\bD} \rightarrow \overline{\bD}$ is injective,
it follows that $f:\bD \rightarrow \bD$ is a biholomorphic map,
which is an element of $\fG_2=\mathrm{PSU}(1,1)$ by Schwarz's lemma.
\end{proof}

%

Let $\CE(1)^\times$ denote the subgroup of $\CE(1)$ consisting of those elements
$f \in \CE(1)$ that are invertible as elements of $\CE(1)$.
By Proposition \ref{prop_CE_two} and Proposition \ref{prop_CE_d_explicit}, we have the following:
\begin{prop}\label{prop_invertible}
For any $d \geq 2$, the group $\CE(1)^\times$ coincides with $\fG_d$.
\end{prop}

Let $(A,\rho)$ be a $\CE$-algebra in $\Vect$.
Note that for $*=(\emptyset \rightarrow\bD) \in \CE(0)$, $\rho_0(*) \in \mathrm{Hom}_{\Vect}(\C,A)\cong A$, which defines a distinguished element in $A$.
We denote it by $\va_A$, called the \textbf{vacuum vector}.
\begin{prop}\label{prop_vacuum_1}
Let $(A,\rho)$ be a $\CE$-algebra in $\Vect$. Then, $\rho_1=\rho: \CE(1) \rightarrow \End(A)$ is a unital monoid homomorphism which satisfies the following conditions:
\begin{enumerate}
\item
$\va_A \in A$ satisfies $\rho_1(g)\va_A =\va_A$ for any $g \in \CE(1)$.
\item
$\rho_{n+1}(\phi_{[n+1]})(v_1,\dots,v_n,\va_A)=\rho_n(\phi_{[n+1]} \circ_{n+1} *)(v_1,\dots,v_n)$ for any $\phi_{[n+1]} \in \CE(n+1)$ and $(v_1,\dots,v_n)\in A^n$.
\item
For any $(\phi_1,\dots,\phi_n) \in \CE(n)$ and $g\in \CE(1)$,
\begin{align*}
\rho(g)\rho_n(\phi_1,\dots,\phi_n) = \rho_n(g\circ \phi_1,\dots,g\circ\phi_n)
\end{align*}
and for any $i \in \{1,\dots,n\}$,
\begin{align*}
\rho_n(\phi_1,\dots,\phi_n)\circ_i \rho(g) = \rho_n(\phi_1,\dots, \phi_i \circ g, \dots, \phi_n).
\end{align*}
\end{enumerate}
\end{prop}
\begin{proof}
Since $g \circ_1 * = * \in \CE(0)$ for any $g \in \CE(1)$, $\rho(g)\va_A =\va_A$ holds. (2), (3) are clear.
\end{proof}

\begin{dfn}\label{def_ideal}
Let $(A,\rho)$ be a $\CE$-algebra in $\Vect$. An \textbf{ideal} of $A$ is a subspace $I \subset A$ such that
\begin{align*}
\rho(\phi_1,\dots,\phi_n)(a_1,\dots,a_{n-1},v) \in I
\end{align*}
for any $n \geq 1$, $(\phi_1,\dots,\phi_n)\in \CE(n)$, $a_i \in A$, and $v \in I$.
In this case, the quotient $A/I$ as a vector space naturally carries the structure of a $\CE$-algebra.
\end{dfn}

We end this section by showing that a $\CE$-algebra becomes trivial when $\va_A=0$.

\begin{lem}\label{lem_degenerate}
Let $A$ be a $\CE$-algebra in $\Vect$. If $\va_A=0$, then the following hold:
\begin{enumerate}
\item
For any $n \geq 2$, $\rho_n(\phi_{[n]})=0$ for any $\phi_{[n]} \in \CE(n)$.
\item
If $g \in \CE(1) \setminus \fG$, then $\rho_1(g)=0$.
\end{enumerate}
On the other hand, if a $\CE$-algebra $A$ in $\Vect$ satisfies $\rho_1(g)\va_A=0$ for some $g \in \CE(1)$, then $\va_A=0$.
\end{lem}

\begin{proof}
Let $\phi_{[n]}=(\phi_1,\dots,\phi_n) \in \CE(n)$ for $n \geq 1$.
If
\begin{align}
\bD \not\subset \bigcup_{i=1}^n \overline{\phi_i(\bD)}, \label{eq_sature}
\end{align}
then there exist $x_0 \in \bD$ and a sufficiently small $\ep>0$ such that $B_\ep(x_0) \subset \bD$ and
$\overline{B_\ep(x_0)} \cap \overline{\phi_i(\bD)} =\emptyset$ for any $i=1,\dots,n$.
Hence, $(\phi_1,\dots,\phi_n,T_{x_0}D(\ep)) \in \CE(n+1)$.\\
Thus,
$\rho_n(\phi_1,\dots,\phi_n)
=\rho_{n+1}(\phi_1,\dots,\phi_n,T_{x_0}D(\ep))\circ_{n+1}\rho_0(\emptyset)
=0$.
By Proposition \ref{prop_CE_two} and Proposition \ref{prop_CE_d_explicit}, if $n \geq 2$, then \eqref{eq_sature} always holds for any $\phi_{[n]} \in \CE(n)$. This proves (1). Let us consider the case of $n=1$.
If $d \geq 3$, then for $g \in \fS_d \setminus \fG_d$ it is clear that $\bD \not\subset \overline{g(\bD)}$.
If $d =2$, this follows from Proposition \ref{prop_CE_two}.
The final assertion follows immediately from Proposition \ref{prop_vacuum_1} (1).
\end{proof}

%
%

Let $V$ be a representation of $\fG$ in $\Vect$. Then one can define a $\CE$-algebra structure on $V$ so that it satisfies (1) and (2) of Lemma \ref{lem_degenerate}, and such that $\rho_1|_{\fG}$ agrees with the given representation on $V$. Hence, we have:
\begin{prop}\label{prop_}
There is a one-to-one correspondence between $\fG_d$-representations in $\Vect$ and $\CE$-algebras in $\Vect$ satisfying $\va_A=0$.
\end{prop}


\subsection{Hilbert space filtration and monoid action}\label{sec_ind_Hilb}
In this section, we consider the case where a $\CE$-algebra in $\Vect$ is equipped with a filtration consisting of Hilbert spaces.
Such topology and completeness play an important role in the study of $\CE$-algebras.

Throughout this paper, we denote the tensor product in the category of Hilbert spaces by $\hotimes$, in order to distinguish it from the algebraic tensor product $\otimes$ of vector spaces.
Let $H_1,\dots,H_n,H$ be Hilbert spaces, and let
\begin{align}
T: H_1 \otimes \cdots \otimes H_n \rightarrow H \label{eq_vect_linear}
\end{align}
be a linear map.
Denote the inner product on $H_i$ by $(-,-)_{H_i}$.
Then $H_1 \otimes \cdots \otimes H_n$ carries a natural positive-definite inner product $(-,-)_{H_1\otimes \cdots \otimes  H_n}$
(see \eqref{eq_app_tensor_inner}).
The Hilbert space tensor product $H_1 \hotimes \cdots \hotimes H_n$ is defined as the completion of
$H_1 \otimes \cdots \otimes H_n$ with respect to this inner product.
In this situation, the following two conditions are equivalent:
\begin{itemize}
\item
$T$ extends to a bounded linear map between Hilbert spaces
$\tilde{T}: H_1 \hotimes \cdots \hotimes H_n \rightarrow H$.
\item
There exists a constant $C>0$ such that, for any $v \in H_1 \otimes \cdots \otimes H_n$,
\[
\norm{T(v)}_H \leq C \norm{v}_{H_1 \otimes \cdots \otimes H_n}.
\]
\end{itemize}
When these equivalent conditions are satisfied, we say that the vector space morphism
\eqref{eq_vect_linear} is \emph{bounded}.
A crucial remark is that even if there exists $C>0$ such that
\begin{align*}
\norm{T(v_1 \otimes \cdots \otimes v_n)}_H \leq C \norm{v_1}_{H_1}\cdots \norm{v_n}_{H_n}
\end{align*}
holds for any $v_i \in H_i$, $T$ does not necessarily admit such an extension (For counterexamples, see Remark \ref{rem_subtle_bilinear}).
%
%

\begin{dfn}\label{def_ind_CF}
Let $(V,\m)$ be a $\CE$-algebra in $\Vect$.
We say $(V,\m)$ admits \textbf{ a Hilbert space filtration} if
there are subspaces $\{H^k \subset V\}_{k\geq 0}$ equipped with positive Hermitian inner products $(-,-)_k$ such that:
\begin{enumerate}
\item
$\va_V \in H^0$.
\item
$H^k$ forms an increasing sequence $H^0 \subset H^1 \subset H^2 \subset \dots$ and satisfies $V = \cup_{k \geq 0}H^k$.
Moreover, the inner products $(-,-)_k$ are compatible with the inclusions, that is,
$(-,-)_k|_{H^{k-1}} = (-,-)_{k-1}$.
\item
Each $H^k$ is a separable Hilbert space with respect to $(-,-)_k$.
\item
The monoid homomorphism $\rho:\CEd(1) \rightarrow \mathrm{End}(V)$ satisfies
$\rho(g)(H^k) \subset H^k$ for any $g\in\CEd(1)$ and $k \geq 0$.
\item
For any $n \geq 0$ and $k_1,\dots,k_n \geq 0$, there exists $K$ such that
for any $\phi_{[n]} \in \CEd(n)$, the image of
$H^{k_1}\otimes \cdots \otimes H^{k_n}$ under $\m_n(\phi_{[n]}):V^{\otimes n} \rightarrow V$
is contained in $H^K$. Moreover,
\begin{align*}
\m_n(\phi_{[n]})\Bigl|_{H^{k_1}\otimes \cdots \otimes H^{k_n}}:H^{k_1}\otimes \cdots \otimes H^{k_n} \rightarrow  H^{K}
\end{align*}
is a bounded operator in the above sense.
\end{enumerate}
\end{dfn}
Let $\Hilb$ be the category of separable Hilbert spaces and bounded operators.
In Appendix \ref{app_Hilb}, we will see that a $\CE$-algebra equipped with a Hilbert space filtration defines a $\CE$-algebra in the ind-category of the category of Hilbert spaces, $\Ind$.

Let $(V,\rho)$ be a $\CE$-algebra with Hilbert space filtration $V=\cup_k H^k$.

\begin{dfn}\label{def_closed}
A subspace $W \subset V$ is said to be \textbf{closed} if, for every $k \geq 0$, the intersection
$W \cap H^k$ is a closed subspace with respect to the Hilbert space topology of $H^k$.
\end{dfn}

\begin{dfn}\label{def_closed_ideal}
A \textbf{closed ideal} of $V$ is an ideal $I \subset V$ that is a closed subspace.
We say that $V$ is \textbf{simple} if it has no non-trivial closed ideals.
\end{dfn}

\begin{rem}\label{rem_ideal_quotient}
Assume that $I \subset V$ is closed ideal. 
Unfortunately, the quotient $V/I$ does not necessarily admit a Hilbert space filtration in the above sense.
Indeed, $V/I = \cup_k H^k /(H^k\cap I)$, and $H^k/(I\cap H^k)$ carries a Hilbert space structure
as the orthogonal complement of $I\cap H^k$ in $H^k$, denoted by $I_k^\perp$.
However, the naturally induced linear map
$I_k^\perp \rightarrow H^{k+1}/(I\cap H^{k+1}) \rightarrow I_{k+1}^\perp$
is an injective contraction and not necessarily isometric.
Hence condition (2) fails.
In fact, this map is isometric if and only if $I_k^\perp \subset H^k$ is orthogonal to $I\cap H^{k+1}$.

Definition \ref{def_ind_CF} is much stronger than to define a $\CE$-algebra in $\Ind$,
and one could relax them to incorporate such quotients; however, we do not do so for the following reasons.
We expect that $\CE$-algebras arising from conformal field theory are always simple (see Corollary \ref{cor_vertex_simple}),
and throughout this paper we never take quotients by closed ideals.
This corresponds to the fact that simple vertex operator algebras play a central role
in the theory of vertex operator algebras.
\end{rem}

In studying closed ideals and simplicity, it is important to impose unitarity and continuity conditions on the action of the monoid $\CE(1)$ on $V$.
In what follows, we define such conditions and examine their relation to closed subspaces.


%
\begin{dfn}\label{def_semigroup_nice}
For the monoid homomorphism $\rho:\CEd(1) \rightarrow \mathrm{End}(V)$, we consider the following conditions:
\begin{enumerate}
\item[(U)]
$\rho(g)|_{H^k}: H^k \rightarrow H^k$ is unitary and
$\rho|_{\fG_d}: \fG_d \rightarrow U(H^k)$ is strongly continuous for any $k \geq 0$.
\item[(D)]
For any $1>r>0$ and $k \geq 0$, the dilation $D(r)=r^D \in \CEd(1)$, $\rho(r^D)|_{H^k}:H^k \rightarrow H^k$ satisfies
\begin{enumerate}
\item[(1)]
$\rho(r^D)|_{H^k}: \R_{\geq 0} \rightarrow B(H^k)$ is strongly continuous, that is, for any $v \in H^k$ and $\ep>0$, there is $1>\delta>0$ such that for any $1 \geq r >\delta$,
\begin{align*}
\norm{\rho(r^D)v-v} < \ep.
\end{align*}
\item[(2)]
it is a contraction, that is, the operator norm satisfies $\norm{\rho(r^D)|_{H^k}} \leq 1$,
\item[(3)]
it is a self-adjoint operator.
\item[(4)]
it is a Hilbert-Schmidt operator, i.e., there is an orthonormal basis $\{e_i \in H^k\}_{i \in I}$ such that
\begin{align}
\norm{\rho(r^D)|_{H_k}}_2= \sum_{i \in I}\norm{\rho(r^D)e_i}^2 <\infty.\label{eq_trace_class}
\end{align}
Note that in this case \eqref{eq_trace_class} holds for any orthonormal basis.
\item[(5)]
For each $1 \geq r>0$, there is $M(r)>0$ such that $\mathrm{sup}_{k \geq 0} \norm{\rho(r^D)|_{H_k}}_2 \leq M(r)$.
\end{enumerate}
\end{enumerate}
\end{dfn}

Assume that the dilation $\rho(r^D)$ satisfies (D). 
Let $v \in \ker(\rho(e^{-sD}):H^k \rightarrow H^k)$ for some $s>0$.
By $(\rho(e^{-\ft sD})v,\rho(e^{-\ft sD})v) = (v,\rho(e^{-sD})v)=0$, $v \in \ker(\rho(e^{-\frac{s}{2^N}D}))$ for any $N \geq 0$.
By (D-1), this implies $v=0$. Hence, $\ker(\rho(r^D))=0$ for any $k\geq0$ and $1 \geq r>0$.
Since $\rho(r^D)$ is a Hilbert-Schmidt operator, it is a compact operator \cite[Theorem VI.22]{RS}.
Hence, by (D-3), $H^k$ admits an eigenspace decomposition whose eigenvalues are positive real numbers, and since $\ker \rho(r^D)=0$, each eigenspace is finite-dimensional.
Moreover, since $\{\rho(r^D)\}_{1\geq r >0}$ are mutually commuting, they can be simultaneously diagonalized.

Let $v \in H^k$ be a non-zero simultaneous eigenvector, that is, there exists a map $f:(0,1] \rightarrow \R_{>0}$ such that $\rho(r^D)v=f(r)v$.
Since $f(rs)=f(r)f(s)$ and $f(1)=1$, and $f$ is lower semicontinuous at $r=1$, there exists $\Delta \geq 0$ such that $f(r)=r^{\Delta}$.
Hence we obtain an orthogonal decomposition
\begin{align}
H^k = \bigoplus_{\Delta \geq 0} H_\Delta^k.
\label{eq_orth_dec}
\end{align}
such that $\rho(r^D)|_{H_\Delta^k}=r^{\Delta}$.
Since $\rho(r^D)$ is compact, the eigenvalues $\Delta$ have no accumulation points (other than $\infty$).
Moreover, by (D-3), the decomposition \eqref{eq_orth_dec} is compatible with the inclusions $H^k \hookrightarrow H^{k+1}$.

In this situation, condition (D-4) implies that
$\sum_{\Delta \geq 0} (\dim H_\Delta^k) r^{2\Delta} < \infty$
holds for any $1 \geq r>0$.
It follows that $\rho(r^D):H^k \rightarrow H^k$ is a trace class operator, and
\begin{align*}
\mathrm{tr}|_{H^k}\rho(r^D) = \sum_{\Delta \geq 0} (\dim H_\Delta^k) r^{\Delta} = \norm{\rho(r^{\ft D}|_{H_k})}_2 < \infty.
\end{align*}

%
%
%


Hence, for any $N >0$ and $1>r>0$, we have
\begin{align*}
\sum_{N \geq \Delta \geq 0} \dim H_\Delta^k \leq 
\sum_{N \geq \Delta \geq 0} \dim H_\Delta^k r^{\Delta-N} \leq r^{-N} \mathrm{tr}|_{H^k}\rho(r^D).
\end{align*}
In particular, by (D-5), the sum $\sum_{N \geq \Delta \geq 0} \dim H_\Delta^k$ can be bounded by $r^{-N}M(\sqrt{r})$, independently of $k$.
Let $\overline{V}$ be the Hilbert space completion of $V$. By (D-2), $r^D$ defines a bounded linear map
$\rho(r^D): \overline{V} \rightarrow \overline{V}$.
Then, we have:
\begin{prop}\label{prop_semigroup}
Assume that the conditions (D) hold. Then, the map
$\rho(r^D):\R_{\geq 0} \rightarrow B(\overline{V})$
is a monoid homomorphism such that $\rho(r^D)$ is a contraction, self-adjoint, and trace class for any $1 \geq r >0$.
Moreover, $\overline{V}$ admits an orthogonal decomposition as a Hilbert space
\begin{align*}
\overline{V} = \bigoplus_{\Delta \geq 0}V_\Delta
\end{align*}
such that:
\begin{enumerate}
\item
For any $\Delta \geq 0$, $V_\Delta$ is finite-dimensional and satisfies $V_\Delta \subset H^k$ for some $k \geq 0$.
\item
$\rho(r^D)$ is injective.
\item
$\va_V \in V_0$ and $\va_V \neq 0$ if $V \neq 0$.
\item
For any $1\geq r>0$, $\rho(r^D)|_{V_\Delta}=r^\Delta$.
\item
For any $N \geq 0$ and $1>r>0$,
\begin{align*}
\sum_{N \geq \Delta \geq 0} \dim V_\Delta \leq r^{-N}\mathrm{tr}|_{\overline{V}}\rho(r^D).
\end{align*}
\item
The map $\rho(\bullet^D):\R_{\geq 0} \rightarrow B(\overline{V})$ is strongly continuous.
\end{enumerate}
\end{prop}

\begin{proof}
From the above discussion, all properties are clear except for (3) and (6).
(3) follows from (2) and Lemma \ref{lem_degenerate}. We will show (6).
Let $\ep>0$ and $v= \sum_{\Delta \geq 0}v_\Delta  \in \overline{V}$.
Then, there is $N>0$ such that
$\norm{v -v_N} < \ep/3$ where $v_N= \sum_{N \geq \Delta \geq 0}v_\Delta$.
Since $\rho(r^D)$ is a contraction, $
\norm{\rho(r^D)v -v}= \norm{\rho(r^D)(v-v_N+v_N) -(v-v_N+v_N)}
\leq 2 \norm{v-v_N}+\norm{\rho(r^D)v_N-v_N}$.
Since $v_N \in H^k$ for some $k \geq 0$,
by (D1), there is $1>\delta>0$ such that 
$\norm{\rho(r^D)v_N-v_N} <\ep/3$ for any $1\geq r>\delta$.
Hence, the assertion holds.
\end{proof}

\begin{dfn}\label{def_algebraic_core}
Under the assumption (D), set
\begin{align*}
V^\alg = \bigoplus_{\Delta \geq 0}^{\alg}V_\Delta,\qquad \overline{V^\alg}=\prod_{\Delta \geq 0}V_\Delta,
\end{align*}
the algebraic direct sum
and direct product of vector spaces, respectively. We call $V^\alg$ the \textbf{algebraic core} of the $\CEd$-algebra $(V,\m)$.
\end{dfn}

The importance of $V^\alg$ can be seen in the following lemma:
\begin{lem}\label{lem_lowest_argument}
Assume that $V$ satisfies condition (D).
Let $W \subset V$ be a closed subspace such that $\rho(D(r))W \subset W$ holds for any $r \in (0,1]$.
If $v = \sum_{\Delta \geq \Delta_0} v_\Delta \in W$ with $v_\Delta \in V_\Delta$ and $v_{\Delta_0} \neq 0$, then $v_{\Delta_0} \in W$.
In particular, $W=0$ if and only if $W^\alg = \bigoplus_{\Delta \geq 0} W \cap V_\Delta =0$.
\end{lem}
\begin{proof}
Let $v = \sum_{\Delta \geq \Delta_0} v_\Delta \in W$ with $v_{\Delta_0} \neq 0$. By assumption, there exists $k$ such that $v\in H^k \cap W$.
Set
$a_n =  \sum_{\Delta \geq \Delta_0} \frac{1}{n^{\Delta-\Delta_0}} v_\Delta
= n^{\Delta_0} D(1/n) v \in W \cap H^k$ for $n>0$.
Let $\Delta_1$ be the smallest $\Delta > \Delta_0$ such that $v_\Delta \neq 0$.
Then
$\norm{a_n-v_{\Delta_0}}_k \leq \frac{1}{n^{\Delta_1-\Delta_0}} \norm{v}_k$,
so $a_n$ converges to $v_{\Delta_0}$ in $H^k$.
Since $W\cap H^k$ is closed, $v_{\Delta_0} \in W$.
\end{proof}

%

\begin{prop}\label{prop_proper_ideal}
Assume that the conditions (D) hold.
Suppose $\dim V_0=1$. Then, for any closed ideal $I \subset V$, the following are equivalent:
\begin{enumerate}
\item
$I \neq V$.
\item
$(\va_V,v)=0$ for any $v\in I$.
\end{enumerate}
\end{prop}
\begin{proof}
The implication (2)$\Rightarrow$(1) follows from Proposition \ref{prop_semigroup} (3).
We prove (1)$\Rightarrow$(2) by contrapositive.
Assume that $I$ is a closed ideal and that there exists $v \in I$ such that $(\va_V,v)\neq 0$.
Since $\dim V_0 =1$, we can decompose
$v = \va_V+ \sum_{\Delta >0}v_\Delta \in V$
into eigenvectors of $D(r)$.
By Lemma \ref{lem_lowest_argument}, this implies $\va_V \in I$.
Let $r \in (0,1)$.
Choose $x_0 \in \bD$ and a sufficiently small $\ep>0$ such that
\begin{align*}
\overline{B_r(0)} \cap \overline{B_\ep(x_0)} =\emptyset,\quad B_\ep(x_0) \subset \bD.
\end{align*}
Then $(D(r), T_{x_0}D(\ep)) \in \CE(2)$, and
\begin{align*}
\rho_2(D(r), T_{x_0}D(\ep))(a,\va_V) = \rho_1(D(r))(a) \in I
\end{align*}
holds for any $a \in V$.
By assumption, if $a \in H^k$, then $\rho_1(D(r)) a \to a$ strongly as $r \to 1$.
Since $I \cap H^k \subset H^k$ is closed, it follows that $a \in I$.
Hence $I=V$.
\end{proof}

%
%



We end this section by constructing a trivial example that is not simple.
Let $\Hf$ be a Hilbert space
and $\rho:\CE(1) \rightarrow B(\Hf)$
a monoid homomorphism.
For any $p \geq 1$, let
\begin{align*}
\bs^p:\Hf^{\hotimes p} \rightarrow \Hf^{\hotimes p}
\end{align*}
be a linear map defined by
\begin{align}
\bs^p(v_1\otimes \cdots \otimes v_p) = \frac{1}{p!}\sum_{\si \in S_p} v_{\si(1)}\otimes \cdots \otimes v_{\si(n)}.\label{eq_sym_proj}
\end{align}
Then, $\bs^p$ is an orthogonal projection onto the closed subspace $\Sym^p \Hf \subset \Hf^{\hotimes p}$.
Note that here we are considering the tensor product $\hotimes$ in $\Hilb$ (see Appendix \ref{app_Hilb}). We set $\C\cong \Sym^0 \Hf$, and denote by $\va$, the vector corresponding to $1\in\C$.
Set
\begin{align}
\Hf^k = \bigoplus_{p\geq 0}^k \Sym^p(\Hf),\label{eq_def_Hk}\\
\Sym(\Hf) = \bigoplus_{p\geq 0} \Sym^p(\Hf),\label{eq_def_Sym}
\end{align}
which is the algebraic direct sum of the vector spaces.
By endowing $\Sym^p \Hf$ and $\Sym^q \Hf$ with an inner product such that they are orthogonal whenever $p\neq q$,
$\Hf^k$ becomes a Hilbert space in a natural way, and the inclusion map
$\Hf^k \hookrightarrow \Hf^{k+1}$ preserves the inner product.
Note that $\Sym(\Hf)$ is not complete and not a Hilbert space.
The symmetric algebra $\Sym(\Hf)$ becomes a unital commutative associative algebra by defining the product
$a \cdot b = \bs^{p+q}(a\otimes b)$ for $a \in \Sym^p\Hf$ and $b \in \Sym^q\Hf$.
This product is continuous in each degree.
Extending $\rho: \CE(1) \rightarrow B(\Hf)$ onto $\Sym(\Hf)$ by
\begin{align*}
\rho_1(g) (v_1 \otimes \cdots \otimes v_p)= \rho(g)v_1 \otimes \cdots \otimes \rho(g) v_p
\end{align*}
for any $p \geq 1$, define
\begin{align*}
\rho:\CE(n) \rightarrow \mathrm{Hom}(\Sym(\Hf)^{\otimes n}, \Sym(\Hf))
\end{align*}
by setting, for $(g_1,\dots,g_n)\in \CE(n)$ and $a_i \in \Sym(\Hf)$,
\begin{align*}
\rho_{(g_1,\dots,g_n)}(a_1,\dots,a_n)
= \rho(g_1)a_1 \cdot \rho(g_2)a_2 \cdots \rho(g_n)a_n .
\end{align*}
Then the following is clear:
\begin{prop}\label{prop_trivial_example}
$(\mathrm{Sym}\Hf,\rho,\va)$ is a $\CE$-algebra with the Hilbert space filtration
$\Hf^k = \bigoplus_{p\geq 0}^k \Sym^p(\Hf)$.
Moreover, $I = \bigoplus_{p>0} \Sym^p(\Hf)$ is a closed ideal.
In particular, $\mathrm{Sym}\Hf$ is not simple.
\end{prop}

\section{Example of conformally flat d-disk algebra}\label{sec_example}
In this section, using reproducing kernel Hilbert space and harmonic analysis, we construct an example of a $\CE$-algebra with Hilbert space filtration for $d \geq 3$ (for the case of $d=2$, see \cite{MBergman}).
The material in this section is independent of Section \ref{sec_left}.

In this section, except for the proof of Theorem \ref{thm_unbounded}, we work throughout with vector spaces over the real numbers and real Hilbert spaces.
The fact that such a construction is possible reflects that our examples correspond to conformal field theories of real scalar fields.
Note that a $\CE$-algebra in $\VectR$ equipped with a filtration by real Hilbert spaces
naturally determines, by complexification, a $\CE$-algebra in $\Vect$ equipped with a filtration by complex Hilbert spaces.


In Section \ref{sec_infinitesimal}, we recall an irreducible representation of $\sod$ on the space of harmonic polynomials.
We also review several basic results from harmonic analysis.
In Section \ref{sec_reproduce}, using reproducing kernel Hilbert spaces, we construct a representation of the monoid $\CE(1)$ on the Hilbert space completion of harmonic polynomials, denoted by $\Hf$.
Section \ref{sec_inequality} and Section \ref{sec_completion} form the most technically important parts of the construction, where for $(\phi_1,\phi_2) \in \CE(2)$ we construct a bounded linear map
\begin{align*}
C_{\phi_1,\phi_2}:H \hotimes H \rightarrow \R.
\end{align*}
In Section \ref{sec_inequality}, we show the existence and boundedness for special elements.
In Section \ref{sec_completion}, we then construct the map in the general case.
Using these results, in Section \ref{sec_construction} we construct examples of $\CE$-algebras.
In Section \ref{sec_unbounded}, we show that these examples do not extend to the natural Hilbert space completion.



\subsection{Infinitesimal representation of the global conformal group}\label{sec_infinitesimal}
In this section, we review basic results on a representation of $\sod$ on the harmonic polynomials.
The space of harmonic polynomials carries at least three natural inner products:
the representation-theoretic inner product $(-,-)_H$, the combinatorial inner product $(-,-)_F$,
and the $L^2$ inner product $(-,-)_{L^2(S^{d-1})}$.
All three play important roles in our construction.

Although all results in this section are essentially known (see, for example, \cite{ABR,DX}), conventions vary across the literature.
For the convenience of the reader, we therefore provide a detailed review and include proofs of some results 
(we also refer the reader to our review \cite{Mvertex} for the detailed proof of 
Proposition \ref{prop_positive}, Proposition \ref{prop_psi_dn} and Proposition \ref{prop_relation_L2}).

%

Let $\sod^+$ (resp. $\sod^-$) be a subalgebra of $\sod$ spanned by
$\{P_\mu, D, J_{\mu,\nu}\}_{\mu,\nu \in \{1,\dots,d\}}$ (resp. $\{K_\mu, D, J_{\mu,\nu}\}_{\mu,\nu \in \{1,\dots,d\}}$).
For $\al \in \R$, let $\R v_\al$ (resp. $v_\al^\da$) be the one-dimensional representation of $\sod^+$ (resp. $\sod^-$) defined by
\begin{align}
\begin{split}
P_\mu v_\al &=0,\\
D v_\al &= -\al v_\al,\\
J_{\mu,\nu} v_\al &= 0
\end{split}
\label{eq_univ_verma}
\end{align}
and
\begin{align}
\begin{split}
K_\mu v_\al^\da &=0,\\
D v_\al^\da &= \al v_\al^\da,\\
J_{\mu,\nu} v_\al^\da &= 0
\end{split}
\label{eq_dual_Verma}
\end{align}
for $\mu,\nu\in \{1,\dots,d\}$.
Let $V(\al) = \mathrm{Ind}^{\sod} \R v_\al$ (resp.  $V(\al)^\da = \mathrm{Ind}^{\sod} \R v_\al^\da$)
 be the induced representation of $\sod$ from $\R v_\al$ (resp. $\R v_\al^\da$).
Denote by $N(\al)^\da$ the (unique) maximal submodule of $V(\al)^\da$ which does not contain $v_\al^\da$. Then, 
\begin{align*}
L(\al)^\da =V(\al)^\da / N(\al)^\da
\end{align*}
is an irreducible $\sod$-module. 
The irreducible module $L(\al)=V(\al)/N(\al)$ can be defined in the same way as $L(\al)^\da$.
For any $\sod$-module $M$,
we can define a new action by 
$A\cdot_\theta m = \theta(A)\cdot m$
for $A \in \sod$ and $m\in M$ by the automorphism $\theta:\sod \rightarrow \sod$ \eqref{eq_def_theta}.
We denote this module by $\theta^* M$.
Then, it is clear that $\theta^* L(\al) \cong L(\al)^\da$ as $\sod$-modules.
Hence, we have \cite{Hum}:
\begin{prop}\label{prop_bilinear}
There is a unique non-degenerate bilinear form $(-,-):L(\al)^\da \otimes L(\al)^\da \rightarrow \R$ such that
\begin{enumerate}
\item
$(v_\al^\da,v_\al^\da)=1$;
\item
$(Au,v) = -(u,\theta(A)v)$ for any $u,v\in L(\al)^\da$ and $A\in \sod$.
\end{enumerate}
Moreover, $(u,v)=(v,u)$ holds for any $u,v\in L(\al)^\da$.
\end{prop}
The following result can be obtained by computing the Kac determinant (For details, see, for example, \cite[Appendix]{Mvertex}).
\begin{prop}
\label{prop_positive}
Assume $d \geq 3$. Then, $N(\frac{d-2}{2})$ is generated by $\sum_{\rho =1}^d P_\rho^2 v_{\frac{d-2}{2}}^\da$ as a left $\sod$-module
and the bilinear form $(-,-)$ on $L(\frac{d-2}{2})$ is positive-definite.
\end{prop}

As in \eqref{eq_differential_action}, geometrically $P_\mu$ corresponds to $-\pa_\mu$.
In particular, $\sum_{\rho =1}^d P_\rho^2$ corresponds to the Laplacian, and
the representation $L(\frac{d-2}{2})$ admits the following realization on the space of harmonic polynomials.
For any $\al \in \R$,
let $d_\al:\mathrm{so}(d+1,1) \rightarrow \R[x_1,\partial_1,\dots,
x_d,\partial_d]$ be a linear map defined by
\begin{align*}
d_\al(D)&= -E_x - \al \\
d_\al(K_\mu) &= ||x||^2\pa_\mu -2 x_\mu (E_x +\al)\\
d_\al(P_\mu) &= -\pa_\mu \\
d_\al(J_{\mu,\nu}) &= x_\mu \pa_\nu-x_\nu \pa_\mu
\end{align*}
for $\mu,\nu \in \{1,2,\dots,d \}$.
It is easy to show that the linear map $d_\al:\mathrm{so}(d+1,1) \rightarrow \R[x_1,\partial_1,\dots,x_d,\partial_d]$ gives a Lie algebra homomorphism.
Thus, $\R[x_1,\dots,x_d,|x|^\R]$ is a $\sod$-module via $d_\al$.
Since $1 \in \R[x_1,\dots,x_d]$ (resp. $|x|^{-2\al} \in \R[x_1,\dots,x_d]$) satisfies \eqref{eq_univ_verma} (resp. \eqref{eq_dual_Verma}) with respect to $d_\al$,
there exists a unique $\sod$-module homomorphism 
\begin{align}
\Psi_\al^\da: V(\al)^\da \rightarrow \R[x^1,\dots,x^d,|x|^\R] \label{eq_psi_dagger}
\end{align}
and
\begin{align}
\Psi_\al: V(\al) \rightarrow \R[x^1,\dots,x^d]\label{eq_psi_al}
\end{align}
such that $\Psi_\al^\da(v_\al^\da) = ||x||^{-\al}$ and $\Psi_\al(v_\al) = 1$. 
Throughout this paper, we consider only the case $\al = \dn$ with $d \geq 3$.
The differential operator representation of $\sod$ on (rational) polynomials is always understood to be $d_\dn$.


Denote the image of $\Psi_\dn^\da$ (resp. $\Psi_\dn$) by $H\left(\dn\right)^\da$ (resp. $H\left(\dn\right)$).
For $n \geq 0$, let $\R[x_1,\dots,x_d]_n$ be the vector space of polynomials of total degree $n$.
 Set
\begin{align}
\mathrm{Harm}_{d,n} = \left\{P(x)\in \R[x_1,\dots,x_d]_n \mid \left(\sum_{\rho=1}^d \left(\frac{d}{dx_\rho}\right)^2\right) P(x) =0\right\},
\end{align}
which consists of harmonic polynomials of degree $n$.
Then, we have (see \cite[Lemma 3.6 and Proposition 3.8]{Mvertex}):
\begin{prop}\label{prop_psi_dn}
Let $d \geq 3$. Then, $\Psi_\dn^\da$ (resp. $\Psi_\dn$) induces a $\sod$-module isomorphism $L(\dn)^\da \rightarrow H\left(\dn\right)^\da$ (resp. $L(\dn) \rightarrow H\left(\dn\right)$), and 
\begin{align}
H\left(\dn\right) &= \bigoplus_{n \geq 0} \mathrm{Harm}_{d,n}\\
H\left(\dn\right)^\da &= \bigoplus_{n \geq 0} \norm{x}^{-d+2-2n} \mathrm{Harm}_{d,n}.\label{eq_Harm_dagger}
\end{align}
Moreover, define a linear map $T:H\left(\dn\right)\rightarrow \theta^* H\left(\dn\right)^\da$ by
\begin{align}
T_\al(P(x)) =P(x)\norm{x}^{-2\al-2n}= P\left(\frac{x}{\norm{x}^2}\right)\norm{x}^{-2\al}\label{eq_Tal_def}
\end{align}
for any $P(x) \in \Harm_{d,n}$. Then, $T_\al$ is a $\sod$-module isomorphism.
\end{prop}

Let $\Hf$ be the Hilbert space completion of $L\left(\dn\right)^\da$, and denote its inner product by $(-,-)_H$.
By Proposition \ref{prop_psi_dn}, this Hilbert space is the completion of the space of harmonic polynomials
$\Harm_d = \bigoplus_{n \geq 0} \mathrm{Harm}_{d,n}$, and $\sod$ acts via the representation twisted by $\theta^*$.
We simply write this action as $A.P(x) = d_\dn(\theta(A)) P(x)$
for $A \in \sod$ and $P(x) \in \Harm_d$.
Since $d_\dn(\theta(D)) = d_\dn(-D)= E_x+\dn$,
\begin{align*}
\Harm_{d,n} = \left\{P(x) \in \Harm_d \mid D.P(x)= \left(n+\dn\right)P(x)\right\}
\end{align*}
for any $n \geq 0$, which are mutually orthogonal subspaces by Proposition \ref{prop_bilinear}.
Hence,
\begin{align}
\Hf = \bigoplus_{n \geq 0} \Harm_{d,n}\label{eq_harm_orth1}
\end{align}
as a Hilbert space.
We note that
\begin{align}
A_{d,n}=\dim \Harm_{d,n} = \binom{n+d-1}{d-1}-\binom{n+d-3}{d-1}. \label{eq_harm_dim}
\end{align}
For example, $A_{2,n}=2$ with $\Harm_{2,n}\otimes \C =\C z^n \oplus \C\z^n$ and $A_{3,n}=2n+1$, $A_{4,n}=(n+1)^2$.

For $r\in (0,1]$, define a linear map
$D(r):\Hf \rightarrow \Hf$ by
\begin{align*}
D(r)(v) = \sum_{n \geq 0} r^{\dn+n} v_n
\end{align*}
for any $v = \sum_{n\geq 0} v_n \in \Hf$
with $v_n \in \Harm_{d,n}$. Since $s \geq 0$, it is clear that $\norm{D(r)v} \leq \norm{v}$ for any $v\in \Hf$,
and thus, $D(r)$ is a well-defined bounded linear operator.
Since \eqref{eq_harm_orth1} is an orthogonal decomposition of the Hilbert space, the following proposition is immediate:
\begin{prop}\label{prop_dilation}
The family of bounded linear operators $\{D(r)\}_{r \in (0,1]}$ satisfies the following conditions:
\begin{enumerate}
\item
It is a monoid homomorphism $\R_{\geq 0} \rightarrow B(\Hf)$ and $\norm{D(r)} \leq 1$ for any $r \in (0,1]$.
\item
It is strongly continuous, that is, for any $v \in \Hf$ and $\ep >0$, there is $1>\delta >0$ such that 
$\norm{D(r)v - v} < \ep$ for any $1 \geq r >\delta$.
\item
For any $r \in (0,1]$, $D(r)$ is a self-adjoint operator.
\item
For any $1>r > 0$, $D(r)$ is a trace-class map with 
\begin{align*}
\mathrm{tr}|_{\Hf}D(r) = \sum_{n \geq 0}A_{d,n} r^{n+\dn} = \frac{r^{\dn}(1+r)}{(1-r)^{d-1}}.
\end{align*}
\end{enumerate}
\end{prop}

Next, we describe an explicit formula for the inner product $(-,-)_H$ on the space of harmonic polynomials, which will be used in later sections.
Let $d\si$ be the volume form on $S^{d-1}$, normalized by $\int_{S^{d-1}}d\si =1$, and write the inner product on $L^2(S^{d-1})$ as $(-.-)_{L^2(S^{d-1})}$.
By restricting a harmonic polynomial $P(x)$ to the sphere $S^{d-1} \subset \R^d$, we may regard $P(x)$ as a function on $S^{d-1}$.
Then the subspaces $\Harm_{d,n}\subset L^2(S^{d-1})$ ($n \geq 0$) are mutually orthogonal, and
$L^2(S^{d-1}) = \bigoplus_{n \geq 0} \Harm_{d,n}$
gives an orthogonal decomposition as a Hilbert space.
The following can be seen by explicitly computing the inner product of the harmonic polynomial
$(x_1+ i x_2)^n \in \Harm_{d,n}\otimes_\R \C$ (see \cite[Appendix]{Mvertex}):
\begin{prop}\label{prop_relation_L2}
For any $n \geq 0$ and $P_n(x) \in \Harm_{d,n}$, $\norm{P_n(x)}_H^2 = \frac{2n+d-2}{d-2} \norm{P_n(x)}_{L^2(S^{d-1})}^2$.
\end{prop}

\begin{prop}\label{prop_restriction_sphere}
Let $P_n \in \Harm_{d,n}$. Then, for any $\xi \in S^{d-1}$, 
\begin{align*}
|P_n(\xi)| \leq \sqrt{\frac{(d-2)A_{d,n}}{2n+d-2}} ||P_n||_H.
\end{align*}
\end{prop}
\begin{proof}
Let $Y_{n,l}(x)$ be an orthonormal basis of $\Harm_{d,n}$ with $(-,-)_H$. For $\xi \in S^{d-1}$, set
\begin{align*}
E_{\xi}^n(x) = \sum_{l=1}^{A_{d,n}} Y_{n,l}(\xi)Y_{n,l}(x) \in \Harm_{d,n}.
\end{align*}
Then, for any $P_n(x) \in \Harm_{d,n}$, 
\begin{align*}
|P_n(\xi)| = |(P_n(x),E_{\xi}^n(x))_H| \leq ||P_n(x)||_H ||E_{\xi}^n(x)||_H.
\end{align*}
Since the inner product on $\Harm_{d,n}$ is $\SO(d)$-invariant,
 $||E_{\xi}^n(x)||_H^2 = \sum_{l=1}^{A_{d,n}} Y_{n,l}(\xi)^2$ is independent of the choice of $\xi \in S^{d-1}$.
Hence, by Proposition \ref{prop_relation_L2}
\begin{align*}
||E_{\xi}^n(x)||_H^2 &= \int_{S^{d-1}}\sum_{l=1}^{A_{d,n}} Y_{n,l}(\xi)^2 d\si(\xi) =\sum_{l=1}^{A_{d,n}} ||Y_{n,l}(\xi)||_{L^2(S^{d-1})}^2\\
&=\frac{d-2}{2n+d-2}\sum_{l=1}^{A_{d,n}} ||Y_{n,l}(\xi)||_{H}^2 = \frac{(d-2)A_{d,n}}{2n+d-2}.
\end{align*}
\end{proof}

For any index $\al \in \N^d$, $\la \in \R$ and $N \geq 0$, set
\begin{align*}
\al! = \prod_{i=1}^d \al_i!,\quad |\al| =  \sum_{i=1}^d \al_i,\quad (\la)_N=\la (\la+1) \cdots (\la+N-1)
\end{align*}
and for any polynomial $f(x)=\sum_\al a_\al x^\al$, 
\begin{align*}
P^\al=\prod_{\mu=1}^d P_\mu^{\al_\mu},\quad f(P)=\sum_\al a_\al P^\al.
\end{align*}
Define an inner product $(-,-)_F$ on $\R[x_1,\dots,x_d]$ by
\begin{align*}
(x^\al,x^\be)_F = \delta_{\al,\be} \al!
\end{align*}
for any $\al,\be \in \N^d$ (see \cite[Chapter 5]{ABR}).
Here, $x^\al$ and $x^\be$ are monomials.
\begin{lem}
\label{lem_fisher_product}
For any $n,m \geq 0$ and $f \in \R[x_1,\dots,x_d]_n$, $g \in \R[x_1,\dots,x_d]_m$,
\begin{align*}
\norm{fg}_F \leq \sqrt{\binom{n+m}{n}} \norm{f}_F \norm{g}_F.
\end{align*}
\end{lem}
\begin{proof}
Note that for any $\ga \in \N^d$ with $|\ga|=n+m$, by looking at the terms of total degree $n$ in
$(1+x_1)^{\ga_1}\dots (1+x_d)^{\ga_d}$, we have
$\sum_{\substack{|\al|=n,|\be|=m \\ \al+\be = \ga}}\frac{(\al+\be)!}{\al!\be!} = \binom{n+m}{n}$.
Let $f(x)=\sum_{|\al|=n} a_\al x^\al$ and $g(x)=\sum_{|\be|=m} b_\be x^\be$.
By Cauchy-Schwartz,
\begin{align*}
\norm{f(x)g(x)}_F^2 &= \sum_{|\ga|=n+m}\ga! \left(\sum_{\substack{|\al|=n,|\be|=m \\\al+\be=\ga}} a_{\al} b_{\be}\right)^2\\
&= \sum_{|\ga|=n+m}\left(\sum_{\substack{|\al|=n,|\be|=m  \\\al+\be=\ga}} \sqrt{\al!} a_{\al} \sqrt{\be!} b_{\be}\sqrt{\frac{(\al+\be)!}{\al!\be!}}\right)^2\\
&\leq \sum_{|\ga|=n+m} \left(\left(\sum_{\substack{|\al|=n,|\be|=m  \\\al+\be=\ga}} {\al!} a_{\al}^2 {\be!} b_{\be}^2 \right)
\left(\sum_{\substack{|\al|=n,|\be|=m  \\\al+\be=\ga}} \frac{(\al+\be)!}{\al!\be!} \right)\right)\\
&\leq \binom{n+m}{n}\norm{f(x)}_F^2 \norm{g(x)}_F^2.
\end{align*}
\end{proof}

Any polynomial $f \in \R[x_1,\dots,x_d]_m$ of degree $m$ admits a unique decomposition of the form
$f(x) = g(x) + \norm{x}^2 h(x)$
using a harmonic polynomial $g \in \Harm_{d,m}$ and a polynomial $h \in \R[x_1,\dots,x_d]_{m-2}$ of degree $m-2$ (see \cite[Proposition 5.5]{ABR}).
That is, there exists an isomorphism of vector spaces
\begin{align}
\begin{split}
\Harm_d \oplus \norm{x}^2\R[x_1,\dots,x_d] &\rightarrow \R[x_1,\dots,x_d]\\
(g(x),\norm{x}^2 h(x)) &\mapsto g + \norm{x}^2 h(x).
\end{split}
\label{eq_Harm_decomp}
\end{align}
We denote by $\pr_H$ the projection $\R[x_1,\dots,x_d] \rightarrow \Harm_d$ defined by \eqref{eq_Harm_decomp}.
\begin{lem}\label{lem_fisher_adjoint}
The decomposition \eqref{eq_Harm_decomp} is orthogonal with respect to the inner product $(-,-)_F$.
In particular, $\pr_H$ is an orthogonal projection, and for any $f \in \R[x_1,\dots,x_d]$,
the inequality $\norm{\pr_H(f)}_F \leq \norm{f}_F$ holds.
\end{lem}
\begin{proof}
With respect to the inner product $(-,-)_F$, the adjoint operator of multiplication by $x_i$ is $\pa_i$.
That is, for any $f,g \in \R[x_1,\dots,x_d]$,
\begin{align*}
(x_i f,g)_F = (f,\partial_i g)_F
\end{align*}
holds.
Therefore, for $g(x) \in \Harm_d$ and $h(x) \in \R[x_1,\dots,x_d]$, we have
\begin{align*}
(\norm{x}^2 h(x),g(x))_F = (h(x), (\sum_i \pa_i^2)g(x))_F =0.
\end{align*}
Hence $\Harm_d$ and $\norm{x}^2\R[x_1,\dots,x_d]$ are orthogonal to each other.
\end{proof}

%

\begin{lem}\label{lem_harm_proj}
Let $f(x)= \sum_\ga a_\ga x^\ga \in \R[x_1,\dots,x_d]_m$. Then,
\begin{align*}
\pr_H(f(x)) = \frac{1}{2^m \left(\frac{d-2}{2}\right)_m} f(P).1.
\end{align*}
\end{lem}
\begin{proof}
The action of $P_i$ on $\Harm_{d}$ is given by
\begin{align*}
2x_i (E_x+\frac{d-2}{2})-\norm{x}^2 \pa_i.
\end{align*}
The polynomial $f(P).1$ is equal to $2^m \left(\frac{d-2}{2}\right)_m f(x)$ modulo $\norm{x}^2\R[x_1,\dots,x_d]$.
Since $f(P).1$ is a harmonic polynomial, the claim follows from the uniqueness of the decomposition \eqref{eq_Harm_decomp}.
\end{proof}

Although the inner products $(-,-)_H$ and $(-,-)_F$ on the space of harmonic polynomials look completely different, they in fact coincide up to a degree-dependent scaling factor (see \cite[Theorem 5.14]{ABR}).
This fact will be crucial later when estimating the norm of a bilinear form:
\begin{prop}\label{prop_fisher_inner}
For harmonic polynomials of degree $m$,
$f(x)= \sum_\be a_\be x^\be$ and $g(x)=\sum_\ga b_\ga x^\ga \in \Harm_{d,m}$, we have
\begin{align}
{2^m \left(\frac{d-2}{2}\right)_m} (f(x),g(x))_H &= (f(x),g(x))_F,\label{eq_fisher3}\\
(f(P).1,g(P).1)_H &= {2^m \left(\frac{d-2}{2}\right)_m} (f(x),g(x))_F.\label{eq_fisher1}
\end{align}
Moreover, if $h(x)= \sum_\ga a_\ga x^\ga \in \R[x_1,\dots,x_d]_m$ is a polynomial that is not necessarily harmonic, then
\begin{align}
\norm{h(P).1}_H^2 \leq  {2^m \left(\frac{d-2}{2}\right)_m} (h(x),h(x))_F. \label{eq_fisher2}
\end{align}
\end{prop}
\begin{proof}
By \cite[Theorem 5.14]{ABR}, if $f,g$ are harmonic polynomials, then
\begin{align*}
(f(x),g(x))_{L^2(S^{d-1})}=\frac{1}{d(d+2)\cdots (d+2m-2)} (f(x),g(x))_F.
\end{align*}
Hence \eqref{eq_fisher3} follows from Proposition \ref{prop_relation_L2}.
By Lemma \ref{lem_harm_proj}, if $f(x)$ is a harmonic polynomial of degree $m$, then
\begin{align}
f(P).1 = 2^m \left(\dn\right)_m f(x)\label{eq_harm_formula}
\end{align}
holds.
Combining this with \eqref{eq_fisher3} yields \eqref{eq_fisher1}.
If $h(x)= \sum_\ga a_\ga x^\ga \in \R[x_1,\dots,x_d]_m$ is a polynomial that is not necessarily harmonic, then
\begin{align*}
(h(P).1,h(P).1)_H &=  (\pr_H(h)(P).1, \pr_H(h)(P).1)_H\\
&={2^m \left(\frac{d-2}{2}\right)_m} (\pr_H(h),\pr_H(h))_F\\
&\leq {2^m \left(\frac{d-2}{2}\right)_m} \norm{h(x)}_F^2.
\end{align*}
Here, the last inequality follows from Lemma \ref{lem_fisher_adjoint}.
\end{proof}

%

\begin{rem}\label{rem_fisher}
Note here that, in \eqref{eq_fisher3} of Proposition \ref{prop_fisher_inner}, the assumption that $f(x),g(x)$ are harmonic polynomials is essential.
For example, if we take $f(x)= \sum_i x_i^2$, then $f(P).1 = 0$, while $(f(x),f(x))_F = 2d$.
\end{rem}

The following lemma follows from Proposition \ref{prop_psi_dn} (see also \cite[Corollary 5.20]{ABR}):
\begin{lem}
\label{lem_dual_rel}
For any $f(x) \in \R[x_1,\dots,x_d]_m$,
\begin{align*}
\norm{x}^{d-2+2m} f(-\partial)\frac{1}{\norm{x}^{d-2}} = f(P).1 \in \Harm_{d,m}.
\end{align*}
\end{lem}

For any $a \in \bD$ and $N \geq 0$, set 
\begin{align*}
E_a^N = \sum_{\substack{\al \in \N^d\\ |\al|=N}} \frac{1}{\al!}a^\al P^{\al}.1 \in \Harm_{d,N}.
\end{align*}
We will show that $\sum_{N \geq 0} \norm{E_a^N}_H^2 < \infty$ and $E_a = \sum_{N \geq 0}E_a^N$ is an element of $\Hf$.
Set
\begin{align*}
F_a^N(x) &= \sum_{\substack{\al \in \N^d\\ |\al|=N}} \frac{1}{\al!}a^\al x^{\al} = \frac{1}{N!}\left(\sum_{i=1}^d a_i x_i\right)^N
 \in \R[x_1,\dots,x_d]_{N}.
\end{align*}
It is important to note that for any $g(x) = \sum_{|\be|=N} c_\be x^\be  \in \R[x_1,\dots,x_d]_{N}$, 
\begin{align*}
(F_a^N(x), g(x))_F = \sum_{|\be|=N} c_\be a^\be =g(a).
\end{align*}
Hence, the $(-,-)_F$-inner product with $F_a^N$ acts as an evaluation.
By Lemma \ref{lem_harm_proj},
\begin{align*}
E_a^N = 2^N \left(\dn\right)_N \pr_H (F_a^N(x)).
\end{align*}
\begin{lem}\label{lem_evaluation_harmonic}
For any $a\in\bD$ and $g(x) \in \Harm_{d,N}$, $(E_a^N,g)_H=g(a)$.
\end{lem}
\begin{proof}
Since $g$ is harmonic, $(E_a^N,g)_H = 2^N \left(\dn\right)_N (\pr_H (F_a^N(x)), g)_H =(\pr_H (F_a^N(x)), g)_F = (F_a^N(x), g)_F =g(a)$ by \eqref{eq_fisher2}.
\end{proof}

Recall that the Gegenbauer polynomial $C_N(t) \in \R[t]$ is defined by
\begin{align}
(1-2rt+r^2)^{-\dn} = \sum_{N \geq 0} C_N(t) r^N.\label{eq_gegenbauer}
\end{align}
\begin{lem}\label{lem_Gegen_pr}
For any $N \geq 0$,
\begin{align*}
\pr_H(x_1^N) = \frac{N!}{2^N(\dn)_N}\norm{x}^N C_N\left(\frac{x_1}{\norm{x}}\right) \in \Harm_{d,N}.
\end{align*}
\end{lem}
\begin{proof}
The polynomial $C_N(t)$ is of degree $N$ and satisfies $C_N(-t)=(-1)^N C_N(t)$, and its leading term is
$\frac{2^N(\dn)_N}{N!} t^N$.
The Gegenbauer polynomials satisfy a well-known differential equation \cite[Appendix B.2.5]{DX}, and from this it follows immediately that
$\norm{x}^N C_N\left(\frac{x_1}{\norm{x}}\right)$ is a harmonic polynomial
\cite[Lemma 1.4.2 and Theorem 1.4.5]{DX}.\\
By definition,
$\frac{N!}{2^N(\dn)_N}\norm{x}^N C_N\left(\frac{x_1}{\norm{x}}\right)$
is equal to $x_1^N$ modulo $\norm{x}^2$, and hence the claim follows.
\end{proof}

\begin{lem}\label{lem_Gegen}
For any $a,b \in \bD$ and $N \geq 0$,
\begin{align*}
(E_a^N,E_b^N)_H = {\norm{a}^N\norm{b}^N}C_N\left(\frac{(a,b)}{\norm{a}\norm{b}}\right).
\end{align*}
\end{lem}
\begin{proof}
By Proposition \ref{prop_bilinear}, $(-,-)_H$ is invariant under the action of $\mathrm{SO}(d)$ on $\Harm_{d,N}$.
Since for any $g \in \SO(d)$ and $P \in \Harm_{d,N}$, $(E_a^N(g^{-1}x),P(x))_H = (E_a^N(x),P(g(x)))_H = P(g(a))$,
$E_a^N(g^{-1}x) = E_{g(a)}^N(x)$ holds.
Thus, $(E_a^N,E_b^N)_H = (g.E_a^N,g.E_b^N)_H=(E_{g(a)}^N,E_{g(b)}^N)_H = E_{g(a)}^N(g(b))$.
We may assume that $b = \rho e_1 =(\rho,0,\dots,0)\in \bD$ with $0<\rho<1$.
Since $F_{\rho e_1}^N(x) = \frac{1}{N!} \rho^N x_1^N$, by Lemma \ref{lem_Gegen_pr}, $E_{\rho e_1}^N = \rho^N\norm{x}^N C_N\left(\frac{x_1}{\norm{x}}\right)$. Hence, for $b=\rho e_1$,
\begin{align*}
(E_a^N,E_{\rho e_1}^N)_H = {\rho^N \norm{a}^N}C_N\left(\frac{a_1}{\norm{a}}\right) = {\norm{a}^N\norm{b}^N}C_N\left(\frac{(a,b)}{\norm{a}\norm{b}}\right).
\end{align*}
\end{proof}
Since $C_N(1)=\frac{(d-2)_N}{N!}$, by Lemma \ref{lem_Gegen}, $\sum_{N \geq 0} \norm{E_a^N}_H^2 = \sum_{N \geq 0}\frac{(d-2)_N}{N!}\norm{a}^{2N}= \frac{1}{(1-\norm{a}^2)^{d-2}} < \infty$. Hence,
$E_a = \sum_{N \geq 0} E_a^N$ converges in norm and defines an element of the Hilbert space $\Hf$.
\begin{prop}\label{prop_reproducing_kernel}
For $1\geq r >0$, we have $D(r)E_a = r^\dn E_{ra}$, and for $a,b \in \bD$,
\begin{align}
(E_a,E_b)_H = \frac{1}{(1-2(a,b)+|a|^2|b|^2)^{\frac{d-2}{2}}}\label{eq_E_reproduce}
\end{align}
holds.
Moreover, for any $1 \geq \ep>0$, the subspace spanned by $\{E_a\}_{a \in B_\ep(0)} \subset \Hf$ is dense in $\Hf$.
In addition, the family $\{E_a\}_{a \in \bD}$ is linearly independent.
\end{prop}
\begin{proof}
Equation \eqref{eq_E_reproduce} follows immediately from Lemma \ref{lem_Gegen} and \eqref{eq_gegenbauer}.
Since $E_{ra}^N=r^N E_a^N$ for any $1\geq r>0$, the equality $D(r)E_a = r^\dn E_{ra}$ is clear.
To show that $\{E_a\}_{a \in B_\ep(0)}$ spans a dense subspace, it suffices to show that
\begin{align*}
N = \{f \in \Hf\mid (f,E_a)=0\text{ for any }a\in B_\ep(0)\}
\end{align*}
is zero.
Assume that $N$ is non-zero.
If $f =\sum_{N \geq 0}f_N \in N$, then $\sum_{N \geq 0}r^N f_N \in N$ for any $1\geq r >0$, and since $N \subset \Hf$ is a closed subspace, we can choose a non-zero element
$f_n \in N \cap \Harm_{d,n}$ (see the proof of Lemma \ref{lem_lowest_argument}).
This contradicts the fact that $(f_n,E_a)=f_n(a)=0$ holds for all $a \in B_\ep(0)$.
Finally, we show linear independence.
Let $a_i \in \bD$ for $i=1,\dots,M$ be distinct points.
If there exist $c_i \in \R$ such that $\sum_{i=1}^M c_i E_{a_i} =0$, then for any harmonic polynomial $f \in \Harm_d$ we have
$0=(\sum_{i=1}^M c_i E_{a_i},f)= \sum_{i=1}^M c_i f(a_i)$.
Hence $c_i=0$ for all $i=1,\dots,M$.
\end{proof}

%

\subsection{Reproducing Kernel Hilbert Spaces and monoid actions}\label{sec_reproduce}
Let $d \geq 3$. Set
\begin{align*}
K_\bD(x,y)=\frac{1}{(1-2(x,y)+|x|^2|y|^2)^{\dn}}
\end{align*}
for $x,y \in \bD$.
Let $K = \bigoplus_{x\in \bD} \R K_x$ be a vector space with a basis $\{K_x\}_{x \in \bD}$, and define a bilinear form on  $K$ by
\begin{align*}
(K_x,K_y)_K = K_\bD(x,y).
\end{align*}

By Proposition \ref{prop_reproducing_kernel}, the linear map
\begin{align*}
i: K \rightarrow \Hf,\quad K_x \mapsto E_x
\end{align*}
is injective, isometric and the image is dense.
In particular, $\Hf$ is the Hilbert space completion of $K$, and elements of $\Hf$ can be regarded as functions on $\bD$
via the evaluation map given by taking the inner product with $E_x$.
A Hilbert space of this kind is called a {\bf reproducing kernel Hilbert space}.
%
%

Let $\phi \in \CE(1)=\fS_d$. Define an action $\rho_K(\phi)$ on $K=\bigoplus_{x\in \bD} \R K_x$ by
\begin{align*}
\rho_K(\phi) K_x = \Om_\phi(x)^{\dn}K_{\phi(x)}.
\end{align*}
Since, by Lemma \ref{lem_norm_all_inv} and Proposition \ref{prop_Liouville_sphere},
$\rho_K(\phi_1)(\rho_K(\phi_2)K_x) = \Om_{\phi_2}(x)^{\dn} \rho_K(\phi_1)K_{\phi_2(x)}= 
\Om_{\phi_2}(x)^{\dn} \Om_{\phi_1}(\phi_2(x))^\dn K_{\phi_1\phi_2(x)} = \Om_{\phi_1\phi_2}(x)^\dn K_{\phi_1\phi_2(x)}=\rho_K(\phi_1 \phi_2) K_x$, $\rho_K$ is a monoid homomorphism.
Since it is not obvious whether $\rho_K(\phi)$ is bounded, it is non-trivial whether this definition extends to a bounded operator on $\Hf$.

On the other hand, for any $\phi \in \fG_d$, by Proposition \ref{prop_conf_identity} (2),
\begin{align}
(\rho_K(\phi) K_x,\rho_K(\phi) K_y)_K &=  \Om_\phi(x)^{\dn} \Om_\phi(y)^{\dn} K_\bD(\phi(x),\phi(y)) = K_\bD(x,y)\label{eq_unitary_rho}
\end{align} 
for any $x,y \in \bD$. Hence, $\rho_K(\phi)$ defines an isometric isomorphism (an orthogonal operator) on the Hilbert space completion $\Hf$.

\begin{lem}\label{lem_nuclear}
Let $r \in (0,1]$. Then, $\rho(r^D)$ is a contraction map on $K$, that is,
for any $v \in K$,
\begin{align*}
\norm{\rho_K(r^D) v}_K \leq \norm{v}_K.
\end{align*}
Moreover, the bounded linear map on $\Hf$ induced by $\rho_K(r^D)$ coincides with that of Proposition \ref{prop_dilation}.
\end{lem}
\begin{proof}
By Proposition \ref{prop_reproducing_kernel} and Proposition \ref{prop_Liouville_sphere},
$D(r)E_a = r^{\dn}E_{ra} = \Om_{r^D}(a)^\dn E_{r^D a}$.
Hence $\rho_K(r^D)$ admits a unique extension to a bounded operator on $\Hf$.
\end{proof}

\begin{lem}\label{lem_internal_g}
If $g \in \fS$ satisfies $\overline{g(\bD)} \subset \bD$, then $\rho_K(g)$ is a contraction map on $K$.
\end{lem}
\begin{proof}
Suppose that $g \in \fS$ satisfies $\overline{g(\bD)} \subset \bD$. Then, by Proposition \ref{prop_decomposition}, there is $r \in (0,1)$ and $h_1,h_2 \in \fG$ such that $g = h_1r^Dh_2$.
Since $\rho_K(g) =\rho_K(h_1)\circ \rho_K(r^D) \circ \rho_K(h_2)$, and by Lemma \ref{lem_nuclear} and \eqref{eq_unitary_rho} each operator admits an extension as a bounded contraction operator, it follows that $\rho_K(g)$ also admits an extension as a bounded contraction operator.
\end{proof}

\begin{prop}\label{prop_contraction_extension}
For any $g \in \fS$, $\rho_K(g)$ is a contraction map on $V$.
\end{prop}
\begin{proof}
By Lemma \ref{lem_internal_g}, if $g \in \fS$ satisfies $\overline{g(\bD)} \subset \bD$, then a contraction map
$\Hf \rightarrow \Hf$ can be defined as the completion of $\rho_K(g)$.
We denote this map by $\rho(g)$.
Let $g \in \fS$ and assume that $\overline{g(\bD)} \not\subset \bD$.
Set $g_s=gD(e^{-s}) \in \fS$ for $s \geq 0$.
Then, for any $s>0$, we have $\overline{g_s(\bD)} \subset \bD$.
Hence, $\rho(g_s)$ is defined and is a contraction operator on $\Hf$.
By checking equality on the dense subset $K\subset \Hf$, we see that
$\rho(g_{s+t})=\rho(g_s) \circ \rho(D(e^{-t}))$ for any $s,t >0$.

Let $v \in \Hf$.
We claim that $\{\rho(g_s)v\}_{s >0}$ is a Cauchy sequence.
Indeed, if $s>t$, then
\begin{align*}
\norm{\rho(g_s)v- \rho(g_t)v}
&=\norm{\rho(g_t)\rho(D(e^{-(s-t)}))v- \rho(g_t)v}\\
&=\norm{\rho(g_t)(\rho(D(e^{-(s-t)}))v-v)}
\leq \norm{\rho(D(e^{-(s-t)}))v-v}.
\end{align*}
By the strong continuity of $\rho(D(r))$ (Proposition \ref{prop_dilation}),
this sequence is Cauchy.
We denote the limit of this Cauchy sequence by
$A(v) = \lim_{s\to 0} \rho(g_s)v$.
Clearly, $A:\Hf \rightarrow \Hf$ is a linear map, and since
$\norm{A(v)} \leq \norm{v}$ holds for any $v \in \Hf$, it is a bounded contraction operator.
For $x\in \bD$, by definition we have
\begin{align}
&\norm{\rho(g_s)K_x - \rho_K(g)K_x}_K^2\nonumber\\
& = \norm{\rho_K(g_s)K_x - \rho_K(g)K_x}_K
=\norm{\Om_{g_s}(x)^\dn K_{g_s(x)}- \Om_g(x)^\dn K_{g(x)}}_K^2\nonumber\\
&=\norm{e^{-\dn s}\Om_{g}(e^{-s}x)^\dn K_{g(e^{-s}x)}- \Om_g(x)^\dn K_{g(x)}}_K^2\nonumber\\
&=e^{-\dn s}\Om_{g}(e^{-s}x)^\dn
\bigl(e^{-\dn s}\Om_{g}(e^{-s}x)^\dn K_\bD(g(e^{-s}x),g(e^{-s}x))
- \Om_g(x)^\dn K_\bD(g(e^{-s}x),g(x))\bigr)\label{eq_smooth}\\
&\quad
- \Om_{g}(x)^\dn
\bigl(e^{-\dn s}\Om_{g}(e^{-s}x)^\dn K_\bD(g(x),g(e^{-s}x))
- \Om_g(x)^\dn K_\bD(g(x),g(x))\bigr).\nonumber
\end{align}
Since \eqref{eq_smooth} is continuous in $s \geq 0$ and vanishes at $s=0$, we obtain
$A(K_x) = \lim_{s\to 0}\rho(g_s)K_x = \rho_K(g) K_x$.
Thus $A$ is the completion of $\rho_K(g)$ and is a contraction operator.
\end{proof}

Therefore, for any $g \in \fS$, $\rho_K(g)$ extends to a bounded contraction map on $\Hf$.
We denote this extension by $\rho(g)$.
Let $B(\Hf)$ be the space of all bounded linear map on $\Hf$. 
By an argument analogous to \eqref{eq_smooth} and \eqref{eq_unitary_rho}, it follows that the restriction $\rho|_{\fG}$ defines a strongly continuous orthogonal representation of $\fG$. Hence, we have:
\begin{prop}\label{prop_unitary}
$\rho:\fS \rightarrow B(\Hf)$ is a monoid homomorphism such that:
\begin{enumerate}
\item
$\rho(g)E_x = \Om_g(x)^\dn E_{g(x)}$ for any $x\in \bD$ and $g \in \fS$.
\item
$\norm{\rho(g)} \leq 1$ for any $g \in \fS$.
\item
$\rho(g)$ is an isometric isomorphism for any $g\in \fG$, and $\rho|_\fG$ is strongly continuous.
\item
$\rho$ satisfies the conditions (D).
\end{enumerate}
\end{prop}

\subsection{Harmonic polynomials and fundamental inequality}\label{sec_inequality}
Let $a,b \in \bD$ and $1>r,s>0$ satisfy 
\begin{align}
\overline{B_r(a)} \cap \overline{B_s(b)} =\emptyset, 
\qquad {B_r(a)}, {B_s(b)} \subset \bD, \label{eq_geom_sep}
\end{align}
where $\overline{B_r(a)}$ is a closure of $B_r(a)$ in $\R^d$.
This condition corresponds to $(T_ar^D,T_bs^D) \in \CE(2)$ by Proposition \ref{prop_CE_d_explicit}.
Set 
\begin{align}
\si_1 = \frac{r}{\norm{a-b}},\quad \si_2 = \frac{s}{\norm{a-b}}.\label{def_sigma_sep}
\end{align}
Then, $\si_1+\si_2 <1$ holds.
We define a linear map
\begin{align*}
C_{a,r;b,s}:\Harm_{d} \otimes \Harm_{d} \rightarrow \R
\end{align*}
by for any $\be,\ga \in \N^d$
\begin{align}
C_{a,r;b,s}(P^\be. 1,P^\ga.1)
= r^{|\be|+\frac{d-2}{2}} s^{|\ga|+\frac{d-2}{2}}
\left(\partial_x^\be \partial_y^\ga \frac{1}{\norm{x-y}^{d-2}}\right)\Bigl|_{(x,y)=(a,b)}.
\label{eq_def_contraction_trans}
\end{align}
Here $\Bigl|_{(x,y)=(a,b)}$ denotes evaluation at $(x,y)=(a,b)$.
Note that $\frac{1}{\norm{x-y}^{d-2}}$ is a Green function of the Laplacian and 
$\sum_i \partial_{x_i}^2 \frac{1}{\norm{x-y}^{d-2}}=0$ holds for $x\neq y$. Hence, the map $C_{a,r;b,s}$ is well-defined as a linear map.
The goal of this section is to show that $C_{a,r;b,s}$ is a bounded linear operator with respect to $(-,-)_H$, and hence induces a map on the completion,
\begin{align*}
C_{a,r;b,s}:\Hf \hotimes \Hf \rightarrow \R.
\end{align*}
This constitutes the most important analytic part in the construction of a $\CE$-algebra.

\begin{lem}\label{lem_minus_translation}
Let $f_n \in \Harm_{d,n}$ and $g_m \in \Harm_{d,m}$. Then,
\begin{align}
C_{a,r;b,s}(f_n(P).1,g_m(P).1)&=  \si_1^{n+\dn}\si_2^{m+\dn} \left(f_n(-P)g_m(P).1 \right)\Bigl|_{x=\frac{a-b}{\norm{a-b}}},
\label{eq_minus_translation}
\end{align}
where $\frac{a-b}{\norm{a-b}} \in S^{d-1}$.
\end{lem}
\begin{proof}
By definition, we have
\begin{align*}
C_{a,r;b,s}(f_n(P).1,g_m(P).1) &= r^{n+\dn}s^{m+\dn} \left(f_n(\pa_x)g_m(\pa_y) \frac{1}{\norm{x-y}^{d-2}} \right)\Bigl|_{(x,y)=(a,b)}\\
&= r^{n+\dn}s^{m+\dn} \left(f_n(\pa_x)g_m(-\pa_x) \frac{1}{\norm{x}^{d-2}} \right)\Bigl|_{x=a-b}.
\end{align*}
By Lemma \ref{lem_dual_rel}, this equal to
\begin{align*}
&= r^{n+\dn}s^{m+\dn} \left(\norm{x}^{-(d-2+2(n+m))} f_n(-P)g_m(P).1 \right)\Bigl|_{x=a-b}\\
&=  \si_1^{n+\dn}\si_2^{m+\dn} \left(\norm{x}^{-(n+m)} f_n(-P)g_m(P).1 \right)\Bigl|_{x=a-b}\\
&=  \si_1^{n+\dn}\si_2^{m+\dn} \left(f_n(-P)g_m(P).1 \right)\Bigl|_{x=\frac{a-b}{\norm{a-b}}}.
\end{align*}
\end{proof}


\begin{lem}\label{lem_binomial_bound}
For any $\la \geq 1$ and $n,m \geq 0$,
\begin{align*}
\frac{\left(\la\right)_{n+m}}{\left(\la\right)_{n}\left(\la\right)_{m}} \leq \binom{n+m}{n}.
\end{align*}
and for any $\la \geq \ft$ and $n,m \geq 0$
\begin{align*}
\frac{\left(\la\right)_{n+m}}{\left(\la\right)_{n}\left(\la\right)_{m}} \leq (n+m) \binom{n+m}{n}.
\end{align*}
\end{lem}
\begin{proof}
Since $\frac{\left(\la\right)_{n+m}}{\left(\la\right)_{n}\left(\la\right)_{m}}= \frac{(\la + n)_m}{(\la)_m}=\prod_{k \geq 0}^{m-1}(1+\frac{n}{\la+k})$,
$\frac{\left(\la\right)_{n+m}}{\left(\la\right)_{n}\left(\la\right)_{m}}$ is a monotonically decreasing function with respect to $\la$.
Hence the assertion follows from $\frac{\left(1\right)_{n+m}}{\left(1\right)_{n}\left(1\right)_{m}} =\binom{n+m}{n}$.
If $n \geq m$,
\begin{align*}
\frac{\left(\ft\right)_{n+m}}{\left(\ft\right)_{n}\left(\ft\right)_{m}} &=
\frac{(\ft+n)\cdots (\ft+n+m-1)}{\ft (\ft+1)\cdots (\ft +m-1)}\\
&\leq 2 \frac{(n+m)!}{n!(m-1)!}\leq (n+m)\binom{n+m}{n}.
\end{align*}
Hence,  the assertion holds.
%
%
\end{proof}

\begin{prop}\label{prop_bounded}
Let $d \geq 3$. Then, there exists a function $Q_d(t)$ of polynomial growth such that
\begin{align*}
|C_{a,r;b,s}(f_n,g_m)| \leq Q_d(n+m) \binom{n+m}{n}\si_1^n\si_2^m \norm{f_n}_H\norm{g_m}_H
\end{align*}
holds for any $n,m\geq 0$ and $f_n \in \mathrm{Harm}_{d,n}$, $g_m \in \mathrm{Harm}_{d,m}$.
\end{prop}

\begin{proof}
Let $f_n(x)=\sum_\al a_\al x^\al \in \mathrm{Harm}_{d,n}$ and $g_m(x) = \sum_\be b_\be x^\be  \in \mathrm{Harm}_{d,m}$.
be non-zero polynomials.

By Lemma \ref{lem_minus_translation} and Proposition \ref{prop_restriction_sphere}, we have
\begin{align*}
|C_{a,r;b,s}&\left(f_n(P).1,g_m(P).1\right)|\\
&=\si_1^{n+\frac{d-2}{2}} \si_2^{m+\frac{d-2}{2}} \left|\left(f_n(-P)g_m(P).1 \right)\Bigl|_{x=\frac{a-b}{\norm{a-b}}}\right|\\
&\leq \si_1^{n+\frac{d-2}{2}} \si_2^{m+\frac{d-2}{2}}
\sqrt{\frac{(d-2)A_{d,n+m}}{2n+2m+d-2}}
\norm{f_n(-P)g_m(P).1}_H.
\end{align*}

Hence, by \eqref{eq_fisher1} and \eqref{eq_fisher2} in Proposition \ref{prop_fisher_inner} and Lemma \ref{lem_fisher_product},
\begin{align*}
&\left(\norm{f_n(-P)g_m(P).1}_H / \norm{f_n(P).1}_H\norm{g_m(P).1}_H\right)^2\\
&\leq \frac{\left(\frac{d-2}{2}\right)_{n+m}}{\left(\frac{d-2}{2}\right)_{n}\left(\frac{d-2}{2}\right)_{m}}
\left(\norm{f_n(-x)g_m(x)}_F / \norm{f_n(x)}_F\norm{g_m(x)}_F\right)^2\\
&\leq \frac{\left(\frac{d-2}{2}\right)_{n+m}}{\left(\frac{d-2}{2}\right)_{n}\left(\frac{d-2}{2}\right)_{m}}
\binom{n+m}{n},
\end{align*}
where we used $\norm{f_n(-x)}_F = \norm{f_n(x)}_F$.
Thus, we have
\begin{align*}
&\norm{f_n(-P)g_m(P).1}_H \leq \sqrt{\frac{\left(\frac{d-2}{2}\right)_{n+m}}{\left(\frac{d-2}{2}\right)_{n}\left(\frac{d-2}{2}\right)_{m}}
\binom{n+m}{n}}\norm{f_n(P).1}_H\norm{g_m(P).1}_H.
\end{align*}

Here, by Lemma \ref{lem_binomial_bound},
$\sqrt{\frac{(d-2)A_{d,n+m}}{2n+2m+d-2}}
\sqrt{\frac{\left(\frac{d-2}{2}\right)_{n+m}}{\left(\frac{d-2}{2}\right)_{n}\left(\frac{d-2}{2}\right)_{m}}
\binom{n+m}{n}}$ can be bounded by $\binom{n+m}{n}Q_d(n+m)$ for a function $Q_d(t)$ of polynomial growth.

\end{proof}

%
%
%


%

By Proposition \ref{prop_bounded}, we have:
\begin{thm}\label{thm_bounded}
Let $a,b \in \bD$ and $1>r,s>0$ satisfy \eqref{eq_geom_sep}.
There is a unique bounded linear map
\begin{align*}
\hat{C}_{a,r;b,s}:\Hf \hotimes \Hf \rightarrow \R
\end{align*}
such that the restriction of $\hat{C}_{a,r;b,s}$ on $\Harm_d \otimes \Harm_d \subset \Hf\hotimes\Hf$
coincides with $C_{a,r;b,s}$.
\end{thm}
\begin{proof}
Since $\Harm_d \subset \Hf$ is dense, the uniqueness is obvious. Recall that $A_{d,n}=\dim\Harm_{d,n}$ is a polynomial in $n$ of degree $d-2$ (see \eqref{eq_harm_dim}).
Let $\{Y_{n,i} \in \Harm_{d,n} \}_{i=1,\dots, \dim A_{d,n}}$ be the orthonormal basis of $\Harm_{d,n}$ with respect to $(-,-)_H$. Set 
\begin{align*}
c_{n,m}^{i,j} = C_{a,r;b,s}(Y_{n,i},Y_{m,j}).
\end{align*}
By the Riesz representation theorem, the existence of a bounded extension of $C_{a,r;b,s}$ is equivalent to the condition
$\sum_{\substack{n,m \geq 0 \\ i,j}} |c_{n,m}^{i,j}|^2 < \infty$.
By Proposition \ref{prop_bounded}, we have
$|c_{n,m}^{i,j}| \leq Q_d(n+m)\binom{n+m}{n}\si_1^n \si_2^m$,
and since
$\binom{n+m}{n}\si_1^n \si_2^m \leq (\si_1+\si_2)^{n+m}$,
setting $\si=\si_1+\si_2$ we obtain
\begin{align*}
\sum_{n,m \geq 0}\sum_{i =1}^{A_{d,n}}\sum_{j =1}^{A_{d,m}} |c_{n,m}^{i,j}|^2
&\leq \sum_{n,m \geq 0} Q_d(n+m)^2 A_{d,n}A_{d,m} \si^{2(n+m)}\\
&=\sum_{N \geq 0} \si^{2N} Q_d(N)^2\left(\sum_{k =0}^N A_{d,k}A_{d,N-k}\right).
\end{align*}
Here $Q_d(N)^2\left(\sum_{k =0}^N A_{d,k}A_{d,N-k}\right)$ grows at most polynomially in $N$.
Since $\si<1$ by \eqref{eq_geom_sep}, the series converges, and the proposition follows.
\end{proof}

%

This map satisfies the following property:
\begin{lem}\label{lem_evaluate_C}
For any $p,q \in \bD$,
\begin{align*}
\hat{C}_{a,r;b,s}(E_p,E_q) = r^{\frac{d-2}{2}}s^{\frac{d-2}{2}} \frac{1}{\norm{T_a r^D(p)-T_b s^D(q)}^{d-2}}
\end{align*}
\end{lem}
\begin{proof}
Since $\norm{a}^{-2} \norm{a-x}^2 = (1-\frac{2(a,x)}{\norm{a}^2}+\frac{\norm{x}^2}{\norm{a}^2})$,
by applying $r=\frac{\norm{x}}{\norm{a}}$ and $t = \frac{(a,x)}{\norm{a}\norm{x}}$ to \eqref{eq_gegenbauer}, we obtain
\begin{align}
\frac{1}{\norm{a-x}^{d-2}}\Bigl|_{|a|>|x|} = 
\norm{a}^{-d+2}\sum_{n \geq 0} \frac{\norm{x}^n}{\norm{a}^n}C_n\left(\frac{(a,x)}{\norm{a}\norm{x}}\right).\label{eq_expansion2}
\end{align}
Here the sum on the right-hand side converges uniformly on any compact subset of $\{|a|>|x|\}$ \cite[Corollary 5.34]{ABR}.
Since $\sum_{k=0}^N E_a^k \to E_a$ $(N \to \infty)$ converges in norm, and $C_{a,r;b,s}$ is continuous,
\begin{align}
\begin{split}
C_{a,r;b,s}(E_p,E_q)&= C_{a,r;b,s}\left(\sum_{n=0}^\infty E_p^n, \sum_{m=0}^\infty E_q^m\right)
=\sum_{n,m \geq 0} C_{a,r;b,s}\left(E_p^n,E_q^m\right)\\
&=\lim_{R \to \infty}\left(\sum_{N = 0}^R \sum_{\substack{n,m \geq 0 \\ n+m=N}} \left(\sum_{|\al|=n,|\be|=m} r^{n+\frac{d-2}{2}} s^{m+\frac{d-2}{2}} p^\al q^\be \frac{1}{\al!\be!}\partial_x^\al \partial_y^\be \frac{1}{\norm{x-y}^{d-2}}\Biggl|_{x-y=a-b}\right)\right)\\
&= r^{\frac{d-2}{2}} s^{\frac{d-2}{2}}\lim_{R \to \infty} \left(\sum_{N = 0}^R \sum_{\substack{n,m \geq 0 \\ n+m=N}} \left( \sum_{|\al|=n,|\be|=m} \frac{1}{\al!\be!} (rp)^\al \partial_x^\al (sq)^\be\partial_y^\be \frac{1}{\norm{x-y}^{d-2}}\Biggl|_{x-y=a-b}\right)\right)\\
&= r^{\frac{d-2}{2}} s^{\frac{d-2}{2}}\lim_{R \to \infty} \left(\sum_{N = 0}^R \frac{1}{N!}\left((r p\cdot \pa_x + s q\cdot \pa_y)^N \frac{1}{\norm{x-y}^{d-2}}\Biggl|_{x-y=a-b}\right)\right)\\
&= r^{\frac{d-2}{2}} s^{\frac{d-2}{2}}\lim_{R \to \infty} \left(\sum_{N = 0}^R \frac{1}{N!}\left(((r p-sq)\cdot \pa_x)^N \frac{1}{\norm{x}^{d-2}}\Biggl|_{x=a-b}\right)\right),\\
&= r^{\frac{d-2}{2}} s^{\frac{d-2}{2}}\lim_{R \to \infty} \norm{a-b}^{-d+2}\sum_{N = 0}^R \frac{1}{N!}\left(
 \frac{\norm{rp-sq}^N}{\norm{a-b}^N}C_N\left(\frac{(a-b,rp-sq)}{\norm{a-b}\norm{rp-sq}}\right)\right),\\
&=\frac{ r^{\frac{d-2}{2}} s^{\frac{d-2}{2}}}{\norm{a-b+rp-sq}^{d-2}},
\end{split}
\label{eq_expansion3}
\end{align}
where we used $\norm{rp - sq} < r+s < \norm{a-b}$ in the last line.
Since $a+rp = T_a r^D (p)$ and $b+sq = T_b s^D (q)$, the assertion follows.

\end{proof}

%

%

Finally, we give a lower bound for the operator norm of $\hat{C}_{a,r;b,s}$, and prove that the operator norm becomes unbounded as $\si_1+\si_2 \to 1$. As we will see in the next section, this estimate motivates our introduction of the category $\Embc$.
\begin{prop}\label{prop_unbounded}
Let $a,b \in \bD$ and $1>r,s>0$ satisfy \eqref{eq_geom_sep}. Set
\begin{align*}
\si = \frac{r+s}{\norm{a-b}}<1.
\end{align*}
Then the operator norm of the bounded linear map $\hat{C}_{a,r;b,s}:\Hf \hotimes \Hf \rightarrow \R$ satisfies
\begin{align}
\norm{\hat{C}_{a,r;b,s}}\geq
\begin{cases}
 \frac{1}{\sqrt{2}} \left(\si_1\si_2\sum_{N \geq 0}\frac{\si^{2N}}{2N+1}\right)^\ft & d =3,\\
(\si_1\si_2)^{\frac{d-2}{2}}\left(\sum_{N \geq 0}\frac{\si^{2N}}{N+1}\right)^\ft & d\geq 4,
\end{cases}
\label{eq_estimate_below}
\end{align}
where $\si_1=\frac{r}{\norm{a-b}}$ and $\si_2=\frac{s}{\norm{a-b}}$.
\end{prop}


\begin{proof}
Let $\Hf_\C$ be the complex Hilbert space obtained by complexifying the real Hilbert space $\Hf$. Throughout, the inner product is taken to be conjugate-linear in the first variable.
It suffices to consider the operator norm of the complexified bounded linear operator
$\hat{C}_{a,r;b,s}:\Hf_\C \hotimes \Hf_\C \rightarrow \C$
and to show that there exists a sequence of vectors
$v_n \in \Harm_d \otimes_{\R} \C$ ($n \geq 0$) such that
$(v_n,v_m)_H=\delta_{n,m}$ and, setting
\begin{align*}
A_{n,m} = C_{a,r;b,s}(v_n,v_m),
\end{align*}
the sum $(\sum_{n,m \geq 0} |A_{n,m}|^2)^\ft$ is bounded from below by \eqref{eq_estimate_below}.
By the rotational invariance of \eqref{eq_def_contraction_trans}, we may assume without loss of generality that
$a-b = \eta(1,0,\dots,0)=\eta e_1 \in \R^d$ with $\eta >0$.
Then, $\si_1 = \frac{r}{\eta}$ and $\si_2= \frac{s}{\eta}$ and $\si =\si_1+\si_2<1$.
%

Set $z=x_1+ix_2$, $\pa_z =\ft (\pa_1-i\pa_2)$ and
\begin{align*}
L(-1)=\ft (P_1-iP_2), \qquad \Ld(-1)=\ft (P_1+iP_2)\\
L(1)=\ft (K_1+iK_2),\qquad \Ld(1)=\ft (K_1-iK_2)\\
L(0)=\ft (D+iJ_{12}),\qquad \Ld(0)=\ft (D-iJ_{12}).
\end{align*}
Then, $\{\Ld(n)\}_{n=-1,0,1}$ forms the $\mathrm{sl}_2$-triple, $[\Ld(n),\Ld(m)]=(n-m)\Ld(n+m)$, and 
\begin{align*}
C_{a,r;b,s}(L(-1)^n .1,L(-1)^m.1) &= (-1)^{m} r^{n+\dn}s^{m+\dn} \pa_z^{n+m} \frac{1}{\norm{z\z+x_3^2+\dots+x_d^2}^{\dn}}\Bigl|_{x=\eta e_1}\\
 &=(-1)^{n} r^{n+\dn}s^{m+\dn}\frac{(\dn)_{n+m}\z^{n+m}}{\norm{z\z+x_3^2+\dots+x_d^2}^{\dn+n+m}}\Bigl|_{x=\eta e_1}\\
  &=(-1)^{n} r^{n+\dn}s^{m+\dn}\frac{(\dn)_{n+m}}{\eta^{d-2+n+m}}
\end{align*}
On the other hand, using the Hermitian property of $(-,-)_H$ and Proposition \ref{prop_bilinear}, we obtain
\begin{align*}
\norm{L(-1)^n.1}_H^2 &=2^{-2n}((K_1-iK_2)^n \va,(K_1-iK_2)^n 1)_H\\
&=2^{-2n}(1, (P_1+iP_2)^n(K_1-iK_2)^n 1)_H\\
&=(1, \Ld(-1)^n\Ld(1)^n 1)_H = n! (\dn)_n,
\end{align*}
where we used $\Ld(0)1 =- \frac{d-2}{4}1$ and $\Ld(-1)1=0$.

Set $v_n = \frac{1}{\sqrt{n!(\dn)_n}} L(-1)^n.1 \in \Harm_{d,n}$. Then, $(v_n,v_m)_H=\delta_{n,m}$ and set
\begin{align*}
A_{n,m} = C_{a,r;b,s}(v_n,v_m) =
(-1)^{n} {\si_1^{\dn}\si_2^{\dn}}\frac{(\dn)_{n+m}}{\sqrt{n!m!(\dn)_n(\dn)_m}} \si_1^n \si_2^m.
\end{align*}
Set $r_n= \frac{(\dn)_n}{n!}$. Then, $\frac{(\dn)_{n+m}}{\sqrt{n!m!(\dn)_n(\dn)_m}} = \binom{n+m}{n}\frac{r_{n+m}}{\sqrt{r_nr_m}}$.
Assume $d \geq 4$. 
Since $r_n = \frac{\dn(\dn+1)\cdots (\dn+n-1)}{1\cdot 2 \cdots n}$ is monotone increasing in $n$,
we have $\frac{r_{n+m}}{\sqrt{r_n r_m}} \geq 1$, and hence
$|A_{n,m}|\geq \si_1^{\dn}\si_2^{\dn}\binom{n+m}{n}\si_1^n \si_2^m$.
Set $B_{n,m}= \binom{n+m}{n} \si_1^n \si_2^m$.
By the Cauchy--Schwarz inequality,
\begin{align*}
\si^N=(\si_1+\si_2)^N = \sum_{k \geq 0}^N B_{k,N-k}
\leq (N+1)^\ft \left(\sum_{k \geq 0}^N B_{k,N-k}^2\right)^\ft,
\end{align*}
and therefore
\begin{align*}
\sum_{n,m \geq 0} |B_{n,m}|^2
=\sum_{N \geq 0} \sum_{k=0}^N B_{k,N-k}^2
\geq \sum_{N \geq 0}\frac{\si^{2N}}{N+1}.
\end{align*}

Next, we consider the case $d=3$.
Using $(\ft)_k=\frac{(2k)!}{4^k k!}$, we compute
\begin{align*}
\sum_{n,m \geq 0} \frac{(\ft)_{n+m}^2\si_1^{2n}\si_2^{2m}}{n!m!(\ft)_n(\ft)_m}
&= \sum_{n,m \geq 0}
\frac{(2n+2m)!^2\si_1^{2n}\si_2^{2m}}{4^{n+m}(n+m)!^2 (2n)!(2m)!} \\
&= \sum_{N \geq 0}
4^{-N} \binom{2N}{N}\sum_{k=0}^N \binom{2N}{2k}\si_1^{2k}\si_2^{2N-2k} \\
&\geq \ft \sum_{N \geq 0}\si^{2N} 4^{-N}\binom{2N}{N}
\geq \ft \sum_{N \geq 0}\frac{\si^{2N}}{2N+1}.
\end{align*}
Here we used
$\sum_{k=0}^N \binom{2N}{2k}\si_1^{2k}\si_2^{2N-2k}
=\ft\bigl((\si_1+\si_2)^{2N}-(\si_1-\si_2)^{2N}\bigr)\geq \ft \si^{2N}$
and
$4^N= \sum_{k=0}^{2N} \binom{2N}{k} \leq (2N+1)\binom{2N}{N}$.
%
%
\end{proof}

\subsection{Completion of contraction operator}\label{sec_completion}
Recall $\CEc(2)$ consists of $g_1,g_2\in\fS_d$ with
\begin{align}
{g_1(\bD)} \cap {g_2(\bD)} =\emptyset \label{eq_B_disj}
\end{align}
and $\CE(2)$ consists of
\begin{align}
\overline{g_1(\bD)} \cap \overline{g_2(\bD)} =\emptyset. \label{eq_B_disj2}
\end{align}
For $(g_1,g_2)\in \CEc(2)$, define a bilinear map
\begin{align*}
B_{g_1,g_2}: K \otimes K \rightarrow \R
\end{align*}
by
\begin{align*}
B_{g_1,g_2}\bigl(K_{a},K_{b}\bigr)= \Om_{g_1}(a)^{\frac{d-2}{2}}\Om_{g_2}(b)^{\frac{d-2}{2}} \frac{1}{\|g_1a-g_2b\|^{d-2}}
\end{align*}
for $a,b \in \bD$.
By Lemma \ref{lem_evaluate_C}, we have:
\begin{lem}\label{lem_extension}
Let $a,b\in \bD$ and $1>r,s>0$ satisfy \eqref{eq_geom_sep}.
The restriction of $\hat{C}_{a,r;b,s}:\Hf \hotimes \Hf \rightarrow \R$ on the dense subspace $K \otimes K \subset \Hf \hotimes \Hf$ coincides with $B_{T_ar^D,T_bs^D}: K \otimes K \rightarrow \R$. In particular, there is a constant $M_{a,r;b,s}>0$ such that
\begin{align*}
|B_{T_ar^D,T_bs^D}(v)| \leq M_{a,r;b,s}\norm{v}_{\Hf\hotimes \Hf}
\end{align*}
for any $v \in K \otimes K$.
\end{lem}

\begin{lem}
\label{lem_B_inv}
Let $(g_1,g_2) \in \CEc(2)$. Then, for any $f,h \in \fS_d$,
\begin{align*}
B_{hg_1,hg_2} = B_{g_1,g_2}
\end{align*}
and
\begin{align*}
B_{g_1f,g_2} = B_{g_1,g_2}\circ (\rho_K(f) \otimes \id_K),\qquad B_{g_1,g_2f} = B_{g_1,g_2} \circ (\id_K \otimes \rho_K(f))
\end{align*}
holds as linear maps $K\otimes K \rightarrow \R$.
\end{lem}
\begin{proof}
By Proposition \ref{prop_conf_identity}, we have
\begin{align*}
B_{hg_1,hg_2}(K_x,K_y) &= \Om_{hg_1}(x)^{\dn} \Om_{hg_2}(y)^{\dn}\frac{1}{ \norm{hg_1 x-hg_2 y}^{d-2}}\\
&= 
\Om_{hg_1}(x)^{\dn} \Om_{hg_2}(y)^{\dn}
\Om_{h}(g_1x)^{-\dn} \Om_{h}(g_2y)^{-\dn}\frac{1}{ \norm{g_1 x-g_2 y}^{d-2}}\\
&= \Om_{g_1}(x)^{\dn} \Om_{g_2}(y)^{\dn}\frac{1}{ \norm{g_1 x-g_2 y}^{d-2}}= B_{g_1,g_2}(K_x,K_y).
\end{align*}
By setting $fx = x'$, we have
\begin{align*}
B_{g_1f,g_2}(K_x,K_y) &=
\Om_{g_1f}(x)^{\dn} \Om_{g_2}(y)^{\dn}\frac{1}{\norm{g_1f x - g_2 y}^{d-2}}\\
&=\left(\Om_{g_1}(f(x))\Om_{f}(x)\right)^{\dn} \Om_{g_2}(y)^{\dn} \frac{1}{\norm{g_1x' - g_2 y}^{d-2}}\\
&=B_{g_1,g_2}(\Om_f(x)^{\dn}K_{x'},K_y)=B_{g_1,g_2}(\rho_K(f)K_{x},K_y).
\end{align*}
\end{proof}


By Theorem \ref{thm_bounded}, Lemma \ref{lem_extension} and Lemma \ref{lem_B_inv}, we have:
\begin{prop}\label{prop_B_bounded}
Let $(g_1,g_2) \in \CE(2)$. Then, there is a constant $M_{g_1,g_2}$ such that
\begin{align*}
|B_{g_1,g_2}(v)| \leq M_{g_1,g_2} \norm{v}_{\Hf\hotimes \Hf}
\end{align*}
for any $v \in K\otimes K$.
\end{prop}
\begin{proof}
Let $(g_1,g_2) \in \CE(2)$. Then, $g_i \notin \fG$. Hence, by Proposition \ref{prop_decomposition}, there are $a_i \in \bD$, $1>r_i>0$ and $h_i \in \fG$ such that $g_i = T_{a_i} r_i^D h_i$. 
Since $a_i$ is the center of $g_1(\bD)$ and $r_i$ is its radius, \eqref{eq_B_disj2} implies that the assumptions of Lemma \ref{lem_extension} are satisfied. Hence, by Lemma \ref{lem_extension} and Proposition~\ref{prop_unitary},
\begin{align*}
|B_{g_1,g_2}(v)| &= |B_{T_{a_1} r_1^D h_1,T_{a_2} r_2^D h_2}(v)| = |B_{T_{a_1} r_1^D,T_{a_2} r_2^D}((\rho_K(h_1)\otimes \rho_K(h_2))v)|\\
& \leq M_{a_1,r_1;a_2,r_2} \norm{(\rho_K(h_1)\otimes \rho_K(h_2))v)}_{\Hf\hotimes \Hf}
 = M_{a_1,r_1;a_2,r_2}\norm{v}_{\Hf\hotimes \Hf}.
\end{align*}
\end{proof}

Therefore, if $(g_1,g_2) \in \CE(2)$, then $B_{g_1,g_2}: K \otimes K \rightarrow \R$ extends onto a bounded linear map on its completion,
\begin{align*}
\hat{C}_{g_1,g_2}: \Hf \hotimes \Hf \rightarrow \R.
\end{align*}
\begin{thm}\label{thm_wick_contraction}
For any $(g_1,g_2) \in \CE(2)$, there is a unique bounded bilinear map $\hat{C}_{g_1,g_2}: \Hf \hotimes \Hf \rightarrow \R$ such that
\begin{align*}
\hat{C}_{g_1,g_2}(E_a\otimes E_b) = \Om_{g_1}(a)^{\frac{d-2}{2}}\Om_{g_2}(b)^{\frac{d-2}{2}} \frac{1}{\|g_1a-g_2b\|^{d-2}}.
\end{align*}
for any $a,b \in \bD$.
 Moreover, it satisfies
\begin{align*}
C_{hg_1,hg_2} = C_{g_1,g_2},\quad
C_{g_1f,g_2} = C_{g_1,g_2}(\rho(f) \otimes \id),\quad C_{g_1,g_2f} = C_{g_1,g_2} (\id \otimes \rho(f))
\end{align*}
for any $h,f \in \fS$ as bounded linear maps $\Hf \hotimes \Hf \rightarrow \R$
and
\begin{align}
C_{g_1,g_2} \circ (12) = C_{g_2,g_1},
\end{align}
where $(12): \Hf \hotimes \Hf \rightarrow \Hf \hotimes \Hf$ is the permutation.
\end{thm}

Finally, we end this section by giving a necessary and sufficient condition for the existence of a bounded extension:
\begin{thm}\label{thm_unbounded}
Let $(g_1,g_2) \in \CEc(2)$. Then,  $B_{g_1,g_2}: K \otimes K \rightarrow \R$ extends onto a bounded linear map
$\Hf \hotimes \Hf \rightarrow \R$ if and only if $(g_1,g_2) \in \CE(2)$.
\end{thm}
\begin{proof}
By Proposition \ref{prop_B_bounded}, if $(g_1,g_2) \in \CE(2)$, then it admits the bounded linear extension.
Let $(g_1,g_2) \in \CEc(2) \setminus \CE(2)$. Assume that $B_{g_1,g_2}$ admits a bounded linear extension,
that is, there is $M>0$ such that $|B_{g_1,g_2}(v)| \leq M \norm{v}_K$ for any $v\in K\otimes K$.
Let $g_i = T_{a_i} r_i^D h_i$ be the decompositions in Proposition \ref{prop_decomposition}, which satisfies
\begin{align}
\frac{r_1+r_2}{\norm{a_1-a_2}}=1 \label{eq_unbounded_sep}
\end{align}
by $(g_1,g_2) \in \CEc(2) \setminus \CE(2)$.
By the proof of Proposition \ref{prop_B_bounded}, we have $|B_{T_{a_1}r_1^D, T_{a_2}r_2^D}(v)| \leq M\norm{v}_K$.
Hence, for any $1\geq s >0$ and $v \in K\otimes K$, we obtain
\begin{align*}
|B_{T_{a_1}(r_1s)^D, T_{a_2}(r_2s)^D}(v)| = |B_{T_{a_1}r_1^D, T_{a_2}r_2^D}((\rho_K(s^D)\otimes \rho_K(s^D)) v)| \leq M \norm{(\rho_K(s^D)\otimes \rho_K(s^D)) v}_K.
\end{align*}
On the other hand, if $s<1$, then $(T_{a_1}(r_1s)^D, T_{a_2}(r_2s)^D) \in \CE(2)$, thus
$B_{T_{a_1}(r_1s)^D, T_{a_2}(r_2s)^D}$ admits a bounded extension $\hat{C}_{a_1,r_1s; a_2,r_2s}$ in Theorem \ref{thm_bounded}.
Moreover, since $\rho_K(s^D)\otimes \rho_K(s^D)$ also admits a bounded extension
$\rho(s^D)\otimes \rho(s^D)$, it follows that for any $v\in \Hf\hotimes \Hf$
\begin{align*}
|\hat{C}_{a_1,r_1s; a_2,r_2s}(v)| \leq M \norm{(\rho(s^D)\otimes \rho(s^D)) v} \leq M\norm{v}.
\end{align*}
By Proposition \ref{prop_unbounded} and \eqref{eq_unbounded_sep}, for $d \geq 4$, there exists $v_0 \in \Hf \hotimes \Hf$ such that
\begin{align*}
|\hat{C}_{a_1,r_1s; a_2,r_2s}(v_0)| \geq \left(\sum_{N \geq 0}\frac{s^{2N}}{N+1}\right)^\ft
\end{align*}
and $\norm{v_0}=1$. Hence, $M\geq \left(\sum_{N \geq 0}\frac{s^{2N}}{N+1}\right)^\ft$ for any $1>s$, which yields a contradiction. The case $d=3$ leads to a similar contradiction.
Thus, $B_{g_1,g_2}$ admits no bounded extension.
\end{proof}

\subsection{Construction of conformally flat $d$-disk algebras}\label{sec_construction}
In this section, we construct a $\CE$-algebra using Theorem \ref{thm_wick_contraction}.
We use the notation introduced in Section \ref{sec_ind_Hilb}; in particular,
$\bs^p$, $\Sym^p(\Hf)$, $\Hf^{\hotimes k}$, and $\Sym(\Hf)$ are as defined in \eqref{eq_sym_proj}, \eqref{eq_def_Hk}, and \eqref{eq_def_Sym}.
Set $V=\Sym(\Hf)$.

Let $(g_1,g_2) \in \CE(2)$.
In Section \ref{sec_completion}, we defined a bounded linear map
$\hat{C}_{g_1,g_2}:\Hf\hotimes\Hf \rightarrow \R$.
Using this operator, for $p,q \geq 0$ we define
\begin{align}
\hat{C}_{g_1,g_2}:\Hf^{\hotimes p} \hotimes \Hf^{\hotimes q}
\rightarrow
\Hf^{\hotimes p-1} \hotimes \Hf^{\hotimes q-1}
\label{eq_hat_C_def}
\end{align}
as follows.
If $p,q \geq 1$, then on the algebraic tensor product
$\Hf^{\otimes p}\otimes \Hf^{\otimes q}$ we set
\begin{align*}
\hat{C}_{g_1,g_2}&(v_1 \otimes \dots \otimes v_p,\,
w_1 \otimes \dots \otimes w_q)\\
&=
\sum_{\substack{p \geq i \geq 1\\ q \geq j \geq 1}}
\hat{C}_{g_1,g_2}(v_i,w_j) (v_1 \otimes \dots \hat{v_i} \dots \otimes v_p)
\otimes
(w_1 \otimes \dots \hat{w_j} \dots \otimes w_q),
\end{align*}
for any $v_1 \otimes \dots \otimes v_p \in \Hf^{\otimes p}$ and
$w_1 \otimes \dots \otimes w_q \in \Hf^{\otimes q}$.
Here $v_1 \otimes \dots \hat{v_i} \dots \otimes v_p$ denotes the vector in $\Hf^{\otimes (p-1)}$ obtained by omitting $v_i$.
By Theorem \ref{thm_wick_contraction}, this operator is bounded and therefore induces
a bounded linear map on the completions, as in \eqref{eq_hat_C_def}.
If $p=0$ or $q=0$, we set $\hat{C}_{g_1,g_2}=0$.

\begin{lem}\label{lem_wick_symmetric}
As a bounded linear operator
$\Hf^{\hotimes p} \hotimes \Hf^{\hotimes q}
\rightarrow
\Hf^{\hotimes p-1} \hotimes \Hf^{\hotimes q-1}$,
the following identity holds:
\begin{align*}
\hat{C}_{g_1,g_2} \circ (\bs^p \hotimes \bs^q)
=
(\bs^{p-1}\hotimes \bs^{q-1}) \circ \hat{C}_{g_1,g_2}.
\end{align*}
In particular, the image of the restriction of $\hat{C}_{g_1,g_2}$ to
$\Sym^p(\Hf) \hotimes \Sym^q(\Hf)$ is contained in
$\Sym^{p-1}(\Hf) \hotimes \Sym^{q-1}(\Hf)
\subset \Hf^{\hotimes p-1}\hotimes \Hf^{\hotimes q-1}$.
\end{lem}
%
\begin{proof}
Let  $v_1 \otimes \dots \otimes v_p \in \Hf^{\otimes p}$ and
$w_1 \otimes \dots \otimes w_q \in \Hf^{\otimes q}$. Since
\begin{align*}
&\hat{C}_{g_1,g_2} \circ (\bs^p \hotimes \bs^q)(v_1 \otimes \dots \otimes v_p, w_1 \otimes \dots \otimes w_q)\\
&= \frac{1}{p!q!}\sum_{\substack{\si \in S_p,\\ \tau \in S_q}}\sum_{i \in [p], j \in [q]}
\hat{C}_{g_1,g_2}(v_{\si(i)},w_{\tau(j)})
(v_{\si(1)} \otimes \dots \hat{v_{\si(i)}} \dots \otimes v_{\si(p)}) \otimes
(w_{\tau(1)} \otimes \dots \hat{w_{\tau(j)}} \dots \otimes w_{\tau(q)})\\
&= \frac{pq}{p!q!}\sum_{k \in [p], l \in [q]} \sum_{\substack{\si \in S_p,\si(k)=k\\ \tau \in S_q,\tau(l)=l}}
\hat{C}_{g_1,g_2}(v_{k},w_{l})
(v_{\si(1)} \otimes \dots \hat{v_{k}} \dots \otimes v_{\si(p)}) \otimes
(w_{\tau(1)} \otimes \dots \hat{w_{l}} \dots \otimes w_{\tau(q)})\\
&=(\bs^{p-1} \hotimes \bs^{q-1})\circ \hat{C}_{g_1,g_2}(v_1 \otimes \dots \otimes v_p, w_1 \otimes \dots \otimes w_q).
\end{align*}
Hence, the assertion holds.
\end{proof}
%
%

Since $\hat{C}_{g_1,g_2}$ is a degree-lowering operator, for each $p,q \geq 0$ the action of
\begin{align*}
\exp(\hat{C}_{g_1,g_2}) = \sum_{k \geq 0} \frac{1}{k!} \hat{C}_{g_1,g_2}^k
\end{align*}
on $\Sym^p \Hf \otimes \Sym^q \Hf$ is a finite sum.
Therefore, $\exp(\hat{C}_{g_1,g_2}):V \otimes V \rightarrow V$ is well defined, and its restriction
to $\Sym^p \Hf \otimes \Sym^q H$ is a bounded linear operator.


%

%

Let $g \in \CE(1)=\fS$. Define a linear map $\rho_1(g) :V \rightarrow V$ by
\begin{align*}
\rho_1(g)|_{\Sym^p \Hf} = \rho(g) \otimes \dots \otimes \rho(g).
\end{align*}
for $p \geq 1$ and $\rho_1(g)|_{\Sym^0 \Hf} = \id$.

Let $(g_1,\dots,g_n) \in \CE(n)$ for $n \geq 2$.
For $i,j \in \{1,\dots,n\}$ with $i\neq j$, we denote by
$\exp(\hat{C}_{g_i,g_j}^{i,j})$ the operator on $V^{\otimes n}$ obtained by letting
$\exp(\hat{C}_{g_i,g_j})$ act on the $i$-th and $j$-th tensor factors.
In this situation, the operators $\exp(\hat{C}_{g_i,g_j}^{i,j})$ and
$\exp(\hat{C}_{g_k,g_l}^{k,l})$ commute with each other.
The operator
$\exp(\sum_{n \geq i > j \geq 1} \hat{C}_{g_i,g_j}^{i,j})$
is the composition of all these mutually commuting operators.
We also define $\bs:V^{\otimes n} \rightarrow V$ so that its restriction to
$\Sym^{p_1}\Hf \otimes \cdots \Sym^{p_n} \Hf$ coincides with the restriction of $\bs^{p_1+\dots+p_n}:\Hf^{p_1+\dots+p_n} \rightarrow \Sym^{p_1+\dots+p_n}(\Hf)$.

We define a linear map
$\m_{g_1,\dots,g_n}: V^{\otimes n} \rightarrow V$
by the following composition:
\begin{align*}
V^{\otimes n}
\overset{\exp(\sum_{n \geq i > j \geq 1} \hat{C}_{g_i,g_j}^{i,j})}{\longrightarrow}
V^{\otimes n}
\overset{\rho_1(g_1)\otimes \dots \otimes \rho_1(g_n)}{\longrightarrow}
V^{\otimes n}
\overset{\bs}{\longrightarrow}
V.
\end{align*}
Finally, for $*=(\emptyset \rightarrow \bD) \in \CE(0)$, we define
\begin{align*}
\rho_0(*):\R \rightarrow V,\quad 1 \mapsto \va.
\end{align*}

%
%
\begin{thm}\label{thm_CF}
For any $d \geq 3$, $(V,\m,\va)$ is a $\CE$-algebra in $\Vect$ with the Hilbert space filtration $\Hf^k = \bigoplus_{p = 0}^k \Sym^p \Hf$. Moreover, it satisfies the conditions (U) and (D) in Definition \ref{def_semigroup_nice}
and $\dim V_0=1$.
\end{thm}
\begin{proof}
We first verify Definition \ref{def_operad_algebra}.
The $S_n$-invariance is clear. The equality $\rho_1(\id_\bD)=\id_V$ also follows from the definition.
Let $(g_1,\dots,g_n,g_{n+1}) \in \CE(n+1)$ and $(h_1,\dots,h_m) \in \CE(m)$.
For simplicity, we write $\hat{C}_{g_i,g_j}^{i,j}$ as $\hat{C}_{g_i,g_j}$, and we write $\rho_1(g_i)$ simply as $g_i$.
The case $m=0$ in which \eqref{eq_def_of_operad} holds is clear, since $\hat{C}$ and $\rho_1$ act trivially on $\va$.
If $m \geq 1$, then by Theorem \ref{thm_wick_contraction} we have
\begin{align*}
&\m_{g_1,\dots,g_n,g_{n+1}h_1,\dots,g_{n+1}h_m}\\
&=\bs\left((g_1 \otimes \dots \otimes g_n \otimes g_{n+1}h_1 \otimes \dots \otimes g_{n+1} h_m)
e^{\sum \hat{C}_{g_i,g_j}}
e^{\sum \hat{C}_{g_i,g_{n+1}h_j}}
e^{\sum \hat{C}_{g_{n+1}h_i,g_{n+1}h_j}}\right)\\
&=\bs\left((g_1 \otimes \dots \otimes g_n \otimes g_{n+1}h_1 \otimes \dots \otimes g_{n+1} h_m)
e^{\sum \hat{C}_{g_i,g_j}}
e^{\sum \hat{C}_{g_i,g_{n+1}}(\id\otimes h_j)}
e^{\sum \hat{C}_{h_i,h_j}}\right)\\
&=\bs\left((g_1 \otimes \dots \otimes g_n \otimes g_{n+1} \otimes \dots \otimes g_{n+1})
(\id \otimes h_1 \otimes \dots \otimes h_m)
e^{\sum \hat{C}_{g_i,g_j}}
e^{\sum \hat{C}_{g_i,g_{n+1}}(\id\otimes h_j)}
e^{\sum \hat{C}_{h_i,h_j}}\right)\\
&=\bs\left((g_1 \otimes \dots \otimes g_n \otimes g_{n+1} \otimes \dots \otimes g_{n+1})
e^{\sum \hat{C}_{g_i,g_j}}
e^{\sum \hat{C}_{g_i,g_{n+1}}}
(\id \otimes h_1 \otimes \dots \otimes h_m)
e^{\sum \hat{C}_{h_i,h_j}}\right)\\
&=\bs\left((g_1 \otimes \dots \otimes g_n \otimes g_{n+1} \otimes \dots \otimes g_{n+1})
e^{\sum \hat{C}_{g_i,g_j}}
e^{\sum \hat{C}_{g_i,g_{n+1}}}\bs
(\id \otimes h_1 \otimes \dots \otimes h_m)
e^{\sum \hat{C}_{h_i,h_j}}\right)\\
&=\rho_{g_1,\dots,g_n,g_{n+1}}\circ_{n+1} \rho_{h_1,\dots,h_m}.
\end{align*}
Here we have used the fact that, for a linear operator symmetric in its inputs, one may insert the symmetrization operator $\bs$.
By Proposition \ref{prop_unitary} and Proposition \ref{prop_dilation}, conditions (D-1), (D-2), (D-3), and (U) are clear.
By Proposition \ref{prop_dilation},
\begin{align*}
\mathrm{tr}|_{V} \rho(q^D)
= \prod_{n \geq 0} \frac{1}{(1-q^{\dn+n})^{\dim \Harm_{d,n}}}\Bigl|_{|q|<1}
\end{align*}
converges absolutely for $|q|<1$, since $\dim \Harm_{d,n}$ is a polynomial of degree $d-2$ in $n$.
Therefore, conditions (D-4) and (D-5) follow.
%
\end{proof}

\begin{rem}
It is shown in \cite{Mvertex}, using the results of Section \ref{sec_conti_state}, that the $(V,\rho,\va_V)$ constructed here
is a simple $\CE$-algebra.
\end{rem}

Define a linear isomorphism $\Psi:V \rightarrow V$ by
\begin{align*}
\Psi|_{\Sym^p \Hf} = \sqrt{p!}\,\id_{\Sym^p \Hf}.
\end{align*}
Then $\Psi|_{H^k}$ is a bounded linear map for each $k \geq 0$, while $\Psi$ itself is not bounded.
Using this linear isomorphism, we can endow $V$ with a $\CE$-algebra structure by setting
\begin{align}
\rho_{g_1,\dots,g_n}^\Psi
= \Psi \circ \rho_{g_1,\dots,g_n} \circ (\Psi^{-1}\otimes \cdots \otimes \Psi^{-1}).\label{eq_sqrt}
\end{align}
The $\CE$-algebra $(V,\rho^\Psi,\va)$ is of course isomorphic to the original one.
Since $\rho_1$ preserves degree, we have $\rho_1=\rho_1^\Psi$, and hence conditions (U) and (D) are still satisfied.
On the other hand, for $n\geq 2$, the operations $\rho_n^\Psi$ differ from $\rho_n$.

We shall regard the normalized products $\rho^\Psi=(\rho_n^\Psi)_{n\geq 0}$ as
the $\CE$-algebra structure associated with the free scalar field. 
Thus, in the rest of this paper, the $\CE$-algebra constructed in this section means
$(\Sym\Hf,\rho^\Psi)$.

\begin{rem}\label{rem_Hermite}
The unnormalized products $\rho_n$ are useful in the construction, but the
Hermitian properties expected from the scalar field require the normalized
products $\rho_n^\Psi$. The factor appearing in \eqref{eq_sqrt} is the usual
normalization in the Fock space realization of a Hermitian scalar field, and also
arises naturally from Gaussian measures on Hilbert spaces; see, for example,
\cite{Arai,Glimm-Jaffe,Janson}.
In dimension two, the same normalization appears in the comparison with the
affine Heisenberg vertex operator algebra. Namely, after identifying the unitary
affine Heisenberg VOA with the symmetric algebra of the Bergman space by the
isometry, $\rho_2^\Psi$ is identified with the
radius-regularized vertex operator
of the affine Heisenberg VOA \cite{MBergman}.
\end{rem}
%
%
%
%

\subsection{Failure of the Hilbert space completion}\label{sec_unbounded}
Let $d\geq 3$.
We now show that the $\CE$-algebra $(\Sym \Hf,\rho^\Psi)$ does not extend to the natural Hilbert space completion.
Let
\begin{align*}
\mathcal{F}_H = \widehat{\bigoplus}_{p\geq 0} \Sym^p \Hf
\end{align*}
be the direct sum as Hilbert spaces, which is called the \emph{Fock space} in axiomatic QFT \cite{SW,Glimm-Jaffe}.
Then $V$ is naturally a dense subspace of $\mathcal{F}_H$, and $\mathcal{F}_\Hf$ is the Hilbert space completion of $V$.
In this section, we derive a contradiction under the assumption that, for every $n \geq 1$, the product
$\rho_n^\Psi:\CE(n) \rightarrow \mathrm{Hom}_{\Ind}(V^{\otimes n},V)$
is bounded and extends to a bounded linear operator on the Hilbert space completion,
$\mathcal{F}_\Hf^{\hat{\otimes}n}\rightarrow \mathcal{F}_\Hf$.

We first prove the following elementary lemma.
\begin{lem}\label{app_inequality_fail}
For $d \geq 3$, there exist $N \geq 2$, $a_i \in \bD$, and $1>r_i>0$ ($i=1,\dots,N$) satisfying
\begin{align}
\begin{split}
B_{r_i}(a_i) \subset B_1(0)\qquad \text{ for } i=1,\dots,N,\\
\overline{B_{r_i}(a_i)}\cap \overline{B_{r_j}(a_j)}=\emptyset\qquad \text{ for } i\neq j,
\end{split}
\label{app_geom_sep}
\end{align}
such that the following inequality holds:
\begin{align}
\sum_{i \neq j}\frac{r_i^{\frac{d-2}{2}}r_j^{\frac{d-2}{2}}}{\norm{a_{i}-a_{j}}^{d-2}}
>N.
\end{align}
\end{lem}

%
\begin{proof}
Let \(L\geq 4\) be an integer to be chosen later and put
\[
I_L=\{0,1,\ldots,L-1\}^d,\qquad N=L^d .
\]
Set
\[
r=\frac{1}{16\sqrt d\,L},
\]
and for \(\mathbf m=(m_1,\dots,m_d)\in I_L\), define
\[
a_{\mathbf m}
=
4r\left(\mathbf m- \frac{L-1}{2}(1,\ldots,1)\right)\in \mathbb R^d .
\]
Then $|a_{\mathbf m}| \leq 4r\frac{\sqrt d(L-1)}{2} < \frac18$ and, thus, $a_\mathbf{m} \in \bD$.
For distinct \(\mathbf m,\mathbf n\in I_L\), we have
\[
|a_{\mathbf m}-a_{\mathbf n}|\geq 4r.
\]
Hence the closed balls \(\overline B_r(a_{\mathbf m})\) ($\mathbf m \in I_L$) are pairwise disjoint. Moreover,
$|a_{\mathbf m}|+r < \frac18+\frac{1}{16\sqrt d\,L}<1.$
Thus this configuration of balls satisfies the assumption \eqref{app_geom_sep}.

For this configuration, we will show the following inequality:
\begin{align}
\sum_{\mathbf m \neq \mathbf n}\frac{r^{d-2}}{\norm{a_{\mathbf m}-a_{\mathbf n}}^{d-2}}>N.
\label{app_violate}
\end{align}
Since
$\frac{r^{d-2}}{\norm{a_{\mathbf m}-a_{\mathbf n}}^{d-2}}
=
\frac{1}{4^{d-2}\norm{\mathbf m-\mathbf n}^{d-2}}$,
we estimate the sum
$\sum_{\mathbf m \neq \mathbf n}\frac{1}{4^{d-2}\norm{\mathbf m-\mathbf n}^{d-2}}$
from below.
Set $K=\lfloor L/2\rfloor$ and let $\mathbf q=(q_1,\ldots,q_d)\in \{1,\ldots,K\}^d$.
Then the number of ordered pairs whose difference is exactly $\mathbf q$ is
\[
\#\{(\mathbf m,\mathbf n)\in I_L\times I_L \mid \mathbf n-\mathbf m=\mathbf q \} =
\prod_{j=1}^d (L-q_j)
\geq
\left(\frac L2\right)^d .
\]
Moreover, $|\mathbf q|\leq \sqrt d\,K$,
and hence $|\mathbf q|^{-(d-2)}
\geq d^{-\frac{d-2}{2}}K^{-(d-2)}$.
Since there are \(K^d\) choices of \(\mathbf q\), we obtain
\[
\sum_{\mathbf m\neq \mathbf n}
|\mathbf m-\mathbf n|^{-(d-2)}
\geq
\left(\frac L2\right)^d K^d d^{-\frac{d-2}{2}}
K^{-(d-2)}=\left(\frac L2\right)^d d^{-\frac{d-2}{2}}K^{2}.
\]
Since $N=L^d$, we have
\begin{align*}
\frac{1}{N}\sum_{\mathbf m \neq \mathbf n}\frac{1}{4^{d-2}\norm{\mathbf m-\mathbf n}^{d-2}}\geq 2^{-3d+4}d^{-\frac{d-2}{2}}K^2.
\end{align*}
Since $K=\lfloor L/2\rfloor$, by taking \(L\) sufficiently large, $\frac{1}{N}\sum_{\mathbf m \neq \mathbf n}\frac{1}{4^{d-2}\norm{\mathbf m-\mathbf n}^{d-2}}>1$.
Hence, the assertion follows.
%
\end{proof}

We shall show that boundedness of the $\CE$-algebra structure leads to a contradiction with Lemma \ref{app_inequality_fail}.
For $v \in \Hf$, set
\begin{align*}
E(v) = \sum_{p \geq 0} \frac{1}{\sqrt{p!}} v^{\otimes p}.
\end{align*}
This vector is called a \emph{coherent vector}, and coherent vectors have been used to study linear operators on Hilbert Fock spaces (see, for example, \cite{Arai}).
Since 
\begin{align}
\norm{E(v)}^2 = \sum_{p \geq 0} \frac{1}{p!}(v,v)^p = \exp(\norm{v}^2),
\label{app_Ev_norm}
\end{align}
we have $\norm{E(v)} < \infty$, and hence $E(v) \in \mathcal{F}_\Hf$.

Let $n \geq 2$ and $(g_1, \dots,g_n) \in \CE(n)$.
If $\rho_{(g_1, \dots,g_n)}^\Psi$ is bounded with respect to the norm of $\mathcal{F}_\Hf^{\hat{\otimes}n}$, then we denote its bounded extension by
\begin{align}
\overline{\rho}_{(g_1, \dots,g_n)}^\Psi:\mathcal{F}_\Hf^{\hat{\otimes} n}\rightarrow \mathcal{F}_\Hf.
\label{app_two_product}
\end{align}

\begin{lem}\label{lem_app_formula}
If $\rho_{(g_1, \dots,g_n)}^\Psi$ admits a bounded extension,
then for any $v_1,\dots,v_n \in \Hf$, 
\begin{align}
\overline{\rho}_{(g_1,\dots,g_n)}^\Psi(E(v_1)\otimes\cdots\otimes E(v_n)) =  \exp\left(\sum_{n \geq i>j \geq 1} C_{g_i,g_j}(v_i,v_j)\right)  E\left( \sum_{i=1}^n\rho(g_i)(v_i)\right).
\label{eq_app_formula}
\end{align}
\end{lem}

%
%
%
%
%
%
\begin{proof}
We first prove \eqref{eq_app_formula} formally. Let $t_1,\dots,t_n$ be formal variables, and set
\begin{align*}
E(t_iv_i) = \sum_{p \geq 0}\frac{1}{\sqrt{p!}}v_i^{\otimes p}t_i^p \in V[[t_i]].
\end{align*}
Since $\rho_{(g_1,\dots,g_n)}^\Psi$ is well-defined on each finite-degree truncation, we have
\begin{align}
\begin{split}
&\rho_{(g_1,\dots,g_n)}^\Psi(E(t_1v_1),\dots,E(t_nv_n))\\
&=\sum_{p_1,\dots,p_n \geq 0} 
\frac{1}{\sqrt{p_1!\dots p_n!}}
\rho_{(g_1,\dots,g_n)}^\Psi(v_1^{\otimes p_1},\dots, v_n^{\otimes p_n})
t_1^{p_1}\dots t_n^{p_n}\\
&=\sum_{p_1,\dots,p_n \geq 0} 
\frac{1}{{p_1!\dots p_n!}}
\Psi \circ \rho_{(g_1,\dots,g_n)}(v_1^{\otimes p_1},\dots, v_n^{\otimes p_n})
t_1^{p_1}\dots t_n^{p_n}\\
&=\sum_{p_1,\dots,p_n \geq 0} 
\frac{1}{{p_1!\dots p_n!}}
\Psi \circ (\rho(g_1)\otimes \cdots \otimes \rho(g_n))\circ \exp\left(\sum_{i>j}\hat{C}_{g_i,g_j}\right)
(v_1^{\otimes p_1},\dots, v_n^{\otimes p_n})
t_1^{p_1}\dots t_n^{p_n}.
\end{split}
\label{eq_add_Et1}
\end{align}
Here, for $k_{ij}\geq 0$, put $k_i=\sum_{s\neq i}k_{is}$, where we set $k_{ij}=k_{ji}$. Then
\begin{align*}
&\frac{1}{p_1!\dots p_n!}\prod_{i>j}\frac{1}{k_{ij}!}\hat{C}_{g_i,g_j}^{k_{ij}}(v_1^{\otimes p_1},\dots, v_n^{\otimes p_n})t_1^{p_1}\dots t_n^{p_n}\\
&=\frac{1}{(p_1-k_1)!\dots (p_n-k_n)!}\prod_{i>j}\frac{1}{k_{ij}!}
{C}_{g_i,g_j}(v_i,v_j)^{k_{ij}}(t_it_j)^{k_{ij}}(v_1^{\otimes (p_1-k_1)},\dots, v_n^{\otimes (p_n-k_n)})t_1^{p_1-k_1}\dots t_n^{p_n-k_n},
\end{align*}
whenever $p_i\geq k_i$ for all $i$, and the term is zero otherwise. By setting $l_i = p_i-k_i \geq 0$, \eqref{eq_add_Et1} is equal to 
\begin{align*}
&\exp\left(\sum_{i>j}t_it_j {C}_{g_i,g_j}(v_i,v_j)\right)
\sum_{l_1,\dots,l_n \geq 0}
\frac{1}{l_1!\dots l_n!}\Psi\hat{S}(\rho(g_1)v_1^{\otimes l_1} \otimes \cdots \otimes \rho(g_n)v_n^{\otimes l_n})t_1^{l_1}\dots t_n^{l_n}\\
&=\exp\left(\sum_{i>j}t_it_j {C}_{g_i,g_j}(v_i,v_j)\right)
\sum_{l_1,\dots,l_n \geq 0}
\frac{\sqrt{(l_1+\dots+l_n)!}}{l_1!\dots l_n!}\hat{S}(\rho(g_1)v_1^{\otimes l_1}\otimes \cdots \otimes \rho(g_n)v_n^{\otimes l_n})t_1^{l_1}\dots t_n^{l_n}\\
&=\exp\left(\sum_{i>j}t_it_j {C}_{g_i,g_j}(v_i,v_j)\right)
\sum_{N \geq 0}
\frac{1}{\sqrt{N!}}
\sum_{\substack{l_1,\dots,l_n\geq 0\\l_1+\dots+l_n=N}}
\frac{N!}{l_1!\dots l_n!}\hat{S}(\rho(g_1)v_1^{\otimes l_1} \otimes \cdots \otimes \rho(g_n)v_n^{\otimes l_n})t_1^{l_1}\dots t_n^{l_n}.
\end{align*}
Since, by the $S_N$-invariance of $\hat{S}$, $(a_1+\dots+a_n)^{\otimes N}= \hat{S}((a_1+\dots+a_n)^{\otimes N}) = \sum_{\substack{l_1,\dots,l_n \geq 0\\l_1+\dots+l_n=N}}\frac{N!}{l_1!\dots l_n!} \hat{S}(a_1^{\otimes l_1}\otimes \cdots \otimes a_n^{\otimes l_n})$ for any $a_1,\dots,a_n \in \Hf$, the above sum is equal to
\begin{align}
\exp\left(\sum_{n \geq i>j \geq 1} t_i t_j {C}_{g_i,g_j}(v_i,v_j)\right)  E\left( \sum_{i=1}^n t_i \rho(g_i)(v_i)\right) \in V[[t_1,\dots,t_n]].
\label{eq_app_series_fin}
\end{align}

By \eqref{app_Ev_norm}, $E(t_iv_i)$ can be regarded as an $\mathcal F_\Hf$-valued holomorphic function of $t_i \in \C$. Hence, by the boundedness assumption on $\overline{\rho}_{(g_1,\dots,g_n)}^\Psi$,
\begin{align}
\C^n \rightarrow \mathcal F_\Hf,\qquad
(t_1,\dots,t_n)\mapsto \overline{\rho}_{(g_1,\dots,g_n)}^\Psi(E(t_1v_1)\otimes \cdots \otimes E(t_nv_n))
\label{eq_app_entire}
\end{align}
is an $\mathcal{F}_\Hf$-valued entire function on $\C^n$. The formal computation above shows that the Taylor expansion of \eqref{eq_app_entire} is given by \eqref{eq_app_series_fin}.
Since both $\exp\left(\sum_{n \geq i>j \geq 1} t_i t_j {C}_{g_i,g_j}(v_i,v_j)\right)$ and $E\left( \sum_{i=1}^n t_i \rho(g_i)(v_i)\right)$ converge and define entire functions on $\C^n$, by the identity theorem,
\begin{align*}
\overline{\rho}_{(g_1,\dots,g_n)}^\Psi(E(v_1)\otimes \cdots \otimes E(v_n)) = \left.\exp\left(\sum_{n \geq i>j \geq 1} t_i t_j {C}_{g_i,g_j}(v_i,v_j)\right)  E\left( \sum_{i=1}^n t_i \rho(g_i)(v_i)\right)\right|_{t_1=\dots =t_n=1}.
\end{align*}
\end{proof}

\begin{prop}\label{prop_app_inequality}
If $\rho_{(g_1, \dots,g_n)}^\Psi$ admits a bounded extension,
then the following inequality holds
for any $v_1,\dots,v_n \in \Hf$:
\begin{align}
\norm{\sum_{i=1}^n \rho(g_i)(v_i)}^2 + 2  \sum_{i>j} C_{g_i,g_j}(v_i,v_j) - \sum_{i=1}^n \norm{v_i}^2 \leq 0.
\label{app_inequality}
\end{align}
\end{prop}
\begin{proof}
Assume that $\rho_{(g_1,\dots,g_n)}^\Psi$ admits a bounded extension. Then, for any non-zero vectors $v_1,\dots,v_n \in \Hf$ and any $t\in\R$, boundedness implies
\begin{align}
\sup_{t \in \R} \frac{\norm{\overline{\rho}_{{(g_1,\dots,g_n)}}^\Psi(E(tv_1)\otimes\cdots\otimes E(tv_n))}}{\norm{E(tv_1)}\dots \norm{E(tv_n)}}<\infty.
\label{eq_app_infty}
\end{align}
On the other hand, by Lemma \ref{lem_app_formula} and \eqref{app_Ev_norm}, we have
\begin{align*}
\frac{\norm{\overline{\rho}_{{(g_1,\dots,g_n)}}^\Psi(E(tv_1)\otimes\cdots\otimes E(tv_n))}^2}{\norm{E(tv_1)}^2\dots \norm{E(tv_n)}^2} &=
\exp\left(t^2\left(
\norm{\sum_{i=1}^n \rho(g_i)v_i}^2
+ 2\sum_{i > j} C_{g_i,g_j}(v_i,v_j)
- \sum_{i=1}^n \norm{v_i}^2
\right)\right).
\end{align*}
Hence, the assertion holds.
\end{proof}

\begin{cor}\label{cor_app_inequality}
Let $n \geq 2$, $a_i \in \bD$, and $1>r_i>0$
($i=1,\dots,n$) satisfying \eqref{app_geom_sep}.
Assume that $\rho_{(T_{a_i}D(r_i), \dots,T_{a_n}D(r_n))}^\Psi$  admits a bounded extension,
then the following inequality holds
\begin{align}
\sum_{i,j=1}^n \frac{r_i^{\frac{d-2}{2}}r_j^{\frac{d-2}{2}}}{\left(1-2(a_i,a_j)+\norm{a_i}^2\norm{a_j}^2\right)^{\frac{d-2}{2}}} + \sum_{i\neq j}
\frac{r_i^{\frac{d-2}{2}}r_j^{\frac{d-2}{2}}}{\norm{a_i-a_j}^{d-2}}\leq n.
\label{app_inequality_geom}
\end{align}
\end{cor}
\begin{proof}
Set $g_i=T_{a_i}D(r_i)$. Then $(g_1,\dots,g_n)\in \CE(n)$ by \eqref{app_geom_sep}.
Applying Proposition \ref{prop_app_inequality} to
$v_i=E_0=1\in \Harm_d\subset \Hf$, $i=1,\dots,n$, since $\norm{E_0}=1$ and $\rho(g_i)E_0=r_i^{\frac{d-2}{2}}E_{a_i}$,
by Theorem \ref{thm_wick_contraction}, we obtain
\begin{align*}
\norm{\sum_{i=1}^n r_i^{\frac{d-2}{2}}E_{a_i}}^2 + \sum_{i\neq j}
\frac{r_i^{\frac{d-2}{2}}r_j^{\frac{d-2}{2}}}{\norm{a_i-a_j}^{d-2}}\leq n.
\end{align*}
Hence, the assertion follows from
Proposition \ref{prop_reproducing_kernel}.
\end{proof}

\begin{thm}\label{thm_app_unbounded}
Let $d \geq 3$. Then the $\CE$-algebra $(\Sym \Hf,\rho^\Psi)$ in $\Ind$ constructed in Section \ref{sec_construction} is unbounded with respect to the norm of the Fock space $\mathcal{F}_\Hf$. 
In particular, the \(\CE\)-algebra structure does not extend to the natural Hilbert space completion.
\end{thm}
\begin{proof}
Let $N \geq 2$, $a_1,\dots,a_N \in \bD$, and $1>r_i>0$ be as in Lemma \ref{app_inequality_fail}, so that
\begin{align}
\sum_{i \neq j}\frac{r_i^{\frac{d-2}{2}}r_j^{\frac{d-2}{2}}}{\norm{a_{i}-a_{j}}^{d-2}}
>N.
\label{eq_app_assume_fail}
\end{align}
By \eqref{app_geom_sep}, $(T_{a_1}D(r_1),\dots,T_{a_N}D(r_N))\in \CE(N)$.
Since 
\[
\sum_{i,j=1}^n \frac{r_i^{\frac{d-2}{2}}r_j^{\frac{d-2}{2}}}{\left(1-2(a_i,a_j)+\norm{a_i}^2\norm{a_j}^2\right)^{\frac{d-2}{2}}}
=\norm{\sum_{i=1}^n r_i^{\frac{d-2}{2}}E_{a_i}}^2
 \geq 0,\]
\eqref{eq_app_assume_fail} together with Corollary \ref{cor_app_inequality} implies that
$\rho_{(T_{a_1}D(r_1),\dots,T_{a_N}D(r_N))}^\Psi$ does not admit a bounded extension.

\end{proof}

\begin{rem}\label{rem_d3}
In dimension $d=3$, the inequality \eqref{app_inequality_geom} already has a counterexample for $N=2$.
Indeed, set
\begin{align*}
N=2,\qquad r=s=\frac{1}{3}, \qquad a=\left(\frac{2}{3},0,0\right),\qquad b=\left(-\frac{2}{3},0,0\right).
\end{align*}
Then the left-hand side of \eqref{app_inequality_geom} is $\frac{281}{130}>2$. Hence, in dimension $d=3$, $\rho_2^\Psi$ is already unbounded.
\end{rem}

\begin{rem}
\label{rem_app_2d}
In \cite{MBergman}, an example of a $\CEt$-algebra in dimension $d=2$ was constructed using the symmetric algebra $\Sym A^2(\bD)$ of the Bergman space $A^2(\bD)$ as an Ind-Hilbert space.
For the product on $\Sym A^2(\bD)$, one can consider the following inequality analogous to \eqref{app_inequality_geom}:
\begin{align}
\begin{split}
\sum_{i,j=1}^N \frac{r_ir_j}{(1-a_i\bar{a}_j)^{2}} + \sum_{i\neq j}
\mathrm{Re}\,\frac{r_i r_j}{(a_i-a_j)^2} \leq N.
\end{split}
\label{app_inequality_geom_2d}
\end{align}
Here $a_i \in \bD$ are points in the unit disk in the complex plane $\C$, and they satisfy the separation condition \eqref{app_geom_sep}.
The inequality \eqref{app_inequality_geom_2d} is a special case of inequalities of Grunsky--Nehari type, and it follows from \cite[Corollary 1]{De} by taking $F_i(\zeta)=r_i \zeta+a_i$.
In particular, unlike the case $d \geq 3$, the inequality obtained from the boundedness assumption holds for every $N$.
Although \eqref{app_inequality_geom_2d} alone does not imply the boundedness of the $\CEt$-algebra $\Sym A^2(\bD)$, this suggests that, in dimension $d=2$, a Hilbert space completion may be possible.
This is consistent with Segal's definition of two-dimensional conformal field theory as a Hilbert-space-valued symmetric tensor monoidal \cite{Segal}, and with the construction of concrete examples \cite{Te,GKRV2}.
\end{rem}

\section{Left Kan extension of conformally flat d-disk algebra}\label{sec_left}
By Proposition \ref{prop_CF_algebra}, giving a $\CE$-algebra in $\Vect$ is equivalent to giving a symmetric monoidal functor
$\Disk \rightarrow \Vect$.
In Proposition \ref{prop_Left_Kan}, we see that the left Kan extension along the inclusion
$\Disk \hookrightarrow \Embc$ defines a symmetric monoidal functor $\Embc \rightarrow \Vect$.
In this section, we assume that $d \geq 2$ and introduce the notion of a state and study general properties of left Kan extensions.
In Section \ref{sec_left_state}, we consider the case $\cC=\Vect$, while in Section \ref{sec_conti_state} we treat the case $\cC=\Ind$.
In Section \ref{sec_correlation}, we explicitly construct continuous states on the standard sphere $S^d$ and show that the left Kan extension is nontrivial under the assumption (U) and (D).
Finally, in Section~\ref{sec_cor}, we extract correlation functions
on configuration spaces from a $\CEd$-algebra and prove their clustering property.

%

\subsection{Left Kan extension and states}\label{sec_left_state}

%
%


Let $(A,\rho)$ be a $\CE$-algebra in $\Vect$. 
\begin{dfn}\label{def_character}
Let $(U,g,M) \in \Embc$. A \textbf{state} of $A$ on $(U,g,M)$ is a sequence of maps $n \geq 0$
\begin{align}
\chi_{M}:\Embc(\sqcup_n \overline{\bD},M) \rightarrow \mathrm{Hom}_\C(A^{\otimes n}, \C)\label{eq_char_def}
\end{align}
such that:
\begin{enumerate}
\item
For any $n \geq 1$, $m\geq 0$ and $\iota_{[n]} = (\iota_1,\dots,\iota_n) \in \Embc(\sqcup_n \overline{\bD},M)$
and $\phi_{[m]} =(\phi_1,\dots,\phi_m) \in \CE(m)$,
\begin{align*}
\chi_{M}(\iota_{[n]} \circ_i \phi_{[m]})= \chi_{M}(\iota_{[n]}) \circ_i \rho_m(\phi_{[m]})
\end{align*}
as linear maps $A^{\otimes (n+m-1)} \rightarrow \C$, where $\iota_{[n]} \circ_i \phi_{[m]}$ is the composition of maps in $\Embc$.
\item
For any $n \geq 1$ and $\iota_{[n]} = (\iota_1,\dots,\iota_n)  \in \Embc(\sqcup_n \overline{\bD},M)$,
$\chi_{M}(\iota_1,\dots,\iota_n)\circ \si_{A}^{-1} = \chi_{M}(\iota_{\si(1)},\dots,\iota_{\si(n)})$ for any permutation $\si \in S_n$. Here, $\si_{A}$ is the permutation of $A^{\otimes n}$.
\end{enumerate}
The {state} is called \textbf{conformally invariant} if 
\begin{align*}
\chi_{M}(\iota_1,\dots,\iota_n) = \chi_{M}(\phi \iota_1,\dots, \phi \iota_n) 
\end{align*}
for any $\phi \in \Aut(U,g,M)$, the automorphism group in $\Embc$.
\end{dfn}
If $\chi_M$ is a state, then any scalar multiple of $\chi_M$ and any sum of states are also states. Hence the set of all states forms a vector space.

Let $\chi_{M}$ be a state of a $\CE$-algebra $(A,\rho)$ in $\Vect$ on $M \in \Embc$.
By the definition of the pointwise Kan extension, there is a linear map $f_\chi:\mathrm{Lan}_A(M) \rightarrow \C$ such that
\[
\begin{tikzcd}[row sep=normal, column sep=normal]
A^{\otimes n} \arrow[r, "\mathrm{Lan}(\iota_{[n]})"] \arrow[dr, "\chi_{M}(\iota_{[n]})"'] 
& \mathrm{Lan}_A(M) \arrow[d, "f_\chi"] \\
& \C
\end{tikzcd}
\]
for any $n \geq 0$ and $\iota_{[n]} \in \Embc(\sqcup_n \overline{\bD},(U,g,M))$.
Conversely, for any linear map $f:\mathrm{Lan}_A(M) \rightarrow \C$, define a state $\chi_f$  by the composition
\begin{align*}
\mathrm{Lan}_A(\sqcup_n \overline{\bD}) \cong A^{\otimes n} \overset{\mathrm{Lan}_A(\iota_{[n]})}{\longrightarrow} \mathrm{Lan}_A(M) \overset{f}{\rightarrow} \C
\end{align*}
for $\iota_{[n]} \in \Embc(\sqcup_n \overline{\bD},M)$.
Moreover, for any permutation $\si \in S_n$, for the composite
$\sqcup_n \bD \overset{\si}{\rightarrow} \sqcup_n \overline{\bD}
\overset{\iota_{[n]}}{\rightarrow} M$,
using that $\mathrm{Lan}_A$ is symmetric monoidal, (2) follows. Hence, we have:
\begin{prop}\label{prop_universal_character}
The space of states of $A$ on $M$ is naturally isomorphic, as a vector space, to
$\mathrm{Hom}_\C(\mathrm{Lan}_A(M),\C)$.
\end{prop}

\begin{rem}\label{rem_empty_character}
Let $(U,g,M) \in \Embc$ such that $\Embc(\overline{\bD},(U,g,M)) =\emptyset$.
Then, the comma category $\CE / M$ consists of the single object $\emptyset \rightarrow M$. Hence, $\mathrm{Lan}_A(M) \cong A(\emptyset)=\C$. Moreover, in this case the map \eqref{eq_char_def} exists only for $n=0$, and the space of states is a one-dimensional vector space.
\end{rem}

The following lemma will be used later:
\begin{lem}\label{lem_state_non_trivial}
Let $(U,g,M) \in \Embc$ satisfy $\Embc(\bD,(U,g,M)) \neq \emptyset$,
and let $\chi_M$ be a non-zero state.
Then there exist $n \geq 1$ and $\phi_{[n]} \in \Embc(\sqcup_n \overline{\bD},(U,g,M))$ such that
$\chi_M(\phi_{[n]}): A^{\otimes n} \rightarrow \C$ is a non-zero linear map.
\end{lem}
\begin{proof}
Suppose that $\chi_M(\phi_{[n]})=0$ for all $n \geq 1$ and all
$\phi_{[n]} \in \Embc(\sqcup_n \bD,(U,g,M))$.
Since $\Embc(\overline{\bD},(U,g,M)) \neq \emptyset$, take $f \in \Embc(\overline{\bD},(U,g,M))$.
Then
\[
\chi_M(\emptyset) = \chi_M(f \circ *) = \chi_M(f)(\va_A)=0,
\]
which implies $\chi_M=0$, a contradiction.
\end{proof}



\subsection{Continuous state and closed ideal}\label{sec_conti_state}
Let $(V,\rho)$ be a $\CE$-algebra in $\Vect$ equipped with a Hilbert space filtration $V= \cup_k H^k$.

Recall that $\Ind$ is the category of ind objects of the category of Hilbert spaces $\Hilb$.
In this setting, $(V,\rho)$ can naturally be regarded as a $\CE$-algebra in $\Ind$
(see Appendix \ref{app_Hilb}).
Moreover, by Corollary \ref{cor_ind_Hilb}, the category $\Ind$ is cocomplete and carries
a tensor product that is distributive over colimits.
Therefore, Proposition \ref{prop_CF_algebra} yields the following result.

\begin{thm}\label{prop_ind_CF}
Let $(V,\m)$ be a $\CE$-algebra in $\Vect$ equipped with a Hilbert space filtration
$V=\cup_{k \geq 0} H^k$.
Then $(V,\m)$ naturally determines a symmetric monoidal functor
$V:\Disk \rightarrow \Ind$,
and its left Kan extension defines a symmetric monoidal functor
$\mathrm{Lan}_V:\Embc \rightarrow \Ind$.
\end{thm}

\begin{dfn}\label{def_continuous_character}
A state $\chi_M$ of $V$ on $M \in \Embc$ is called \textbf{continuous} if
for any $n \geq 1$, $\iota_{[n]} \in \Embc(\sqcup_n \overline{\bD},M)$ and $k_1,\dots,k_n \geq 0$,
\begin{align*}
\chi_M(\iota_{[n]})|_{H^{k_1}\otimes \cdots \otimes H^{k_n}} H^{k_1}\otimes \cdots \otimes H^{k_n} \rightarrow \C
\end{align*}
is a bounded linear map.
\end{dfn}
This continuity condition is equivalent to the requirement that the state defines a morphism in $\Ind$ (see \eqref{eq_ind_hom_dual}).
Hence, similarly to Proposition \ref{prop_universal_character}, we obtain the following:
\begin{prop}\label{prop_universal_character_ind}
The space of continuous states of $V$ on $M \in \Embc$ is naturally isomorphic, as a vector space, to
$\mathrm{Hom}_{\Ind}(\mathrm{Lan}_V(M),\C)$.
\end{prop}


\begin{prop}\label{prop_null_ideal}
Let $\chi_M$ be a continuous state of $V$ on $M \in \Embc$.
Let $N(\chi_M)$ be a subspace of $V$ consisting of all $v \in V$ such that
\begin{align}
\chi_M(\iota_{[n]})(a_1,\dots,a_{n-1},v)=0\label{eq_null_set_state}
\end{align}
for any $n \geq 1$ and $a_i \in V$ and $\iota_{[n]} \in \Embc(\sqcup_n \overline{\bD},M)$.
Then, $N(\chi_M)$ is a closed ideal. 
\end{prop}
\begin{proof}
By Definition \ref{def_character} (2), $N(\chi_M)$ is an ideal.
Since $N(\chi_M) \cap H^k$ is the intersection, taken over all $a_i \in V$ and $\iota_{[n]} \in \Embc(\sqcup_n \bD,M)$, of
\begin{align*}
\ker\left(\chi_M(\iota_{[n]})(a_1,\dots,a_{n-1},\bullet):H^k \rightarrow \C\right) \subset H^k
\end{align*}
the continuity of the state implies that this is a closed subspace.
\end{proof}
\begin{rem}
If $\Embc(\bD,M)=\emptyset$, then \eqref{eq_null_set_state} always holds and we have $N(\chi_M)=V$.
\end{rem}

As a special case, we consider continuous states on $\bD$.
In this case, $\Embc(\sqcup_n \overline{\bD},\overline{\bD})$ is nothing but $\CE(n)$.
A sequence $(u_k \in H^k)_{k\geq 0}$ is said to be compatible if the restriction of
$(u_k,\bullet) \in (H^k)^*$ to $H^{k-1}$ coincides with $(u_{k-1},\bullet) \in (H^{k-1})^*$.
We denote by $V^\vee$ the vector space consisting of all compatible sequences.
This space is equal to the projective limit
$(H^k)^* \rightarrow (H^{k-1})^* \rightarrow \dots \rightarrow (H^1)^* \rightarrow (H^0)^*$,
which is isomorphic to $\mathrm{Hom}_{\Ind}(V,\C)$ by definition.
For any $u=(u_n) \in V^\vee$, the map
\begin{align*}
\chi_\bD^u:\CE(n) \rightarrow \mathrm{Hom}_\C(V^{\otimes n},\C),\quad
\phi_{[n]} \mapsto ( u , \rho_n(\phi_{[n]})(\bullet))
\end{align*}
defines a continuous state on $\bD$.
In particular, the choice $u = \va$ is called the \textbf{vacuum expectation value} in physics.

Assume that $V$ satisfies the condition (D).
Let $N^\alg(V)$ be the subspace of $V^\alg$ consisting of all $v \in V^\alg$ such that
\begin{align*}
\langle \va,\rho_n(\phi_{[n]})(a_1,\dots,a_{n-1},v)\rangle=0
\end{align*}
for any $n \geq 1$, $a_i \in V^\alg$, and $\phi_{[n]} \in \CE(n)$.
The following proposition is useful:
\begin{prop}\label{prop_criterion_simple}
Assume that $V$ satisfies condition (D) and that $\dim V_0=1$.
Then the following are equivalent:
\begin{enumerate}
\item
$V$ is simple.
\item
$N(\chi_\bD^{\va})=0$.
\item
$N^\alg(V) =0$.
\end{enumerate}
\end{prop}

%
%
\begin{proof}
Assume (1).
By the definition of $\chi_{\bD}^{\va}$, we have $\va \notin N(\chi_\bD^{\va})$, and hence
$N(\chi_\bD^{\va})=0$ by Proposition \ref{prop_null_ideal}.
Conversely, assume (2).
Let $I \subset V$ be a closed ideal.
Assume that $I\neq V$.
Then, by Proposition \ref{prop_proper_ideal}, every element of $I$ is orthogonal to $\va$.
It follows that $I \subset N(\chi_\bD^{\va})$, and hence $I=0$.
Therefore, $V$ is simple.

We next show the equivalence of (2) and (3).
Assume that $N(\chi_\bD^{\va})\neq 0$.
Then, by Lemma \ref{lem_lowest_argument}, we have $N(\chi_\bD^{\va})^\alg \neq 0$.
Hence $N^\alg(V) \neq 0$.
Conversely, suppose that $N^\alg(V) \neq 0$.
Let $v \in N^\alg(V)$.
Then, for any $\phi_{[n]} \in \CE$,
\begin{align*}
\langle \va,\rho_n(\phi_{[n]})(\bullet,\dots,\bullet,v)\rangle: V^{\otimes n-1} \rightarrow \C
\end{align*}
restricts to zero on $(V^\alg)^{\otimes n-1}$.
Since $V^\alg \cap H^k$ is dense in $H^k$, the continuity of the state implies that
$v \in N(\chi_\bD^{\va})$.
%
%
\end{proof}

\begin{cor}\label{cor_vertex_simple}
Assume that $V$ satisfies condition (D) and that $\dim V_0=1$. Set
\begin{align*}
N_{2\text{-pt}}^\alg(V) = \{v \in V^\alg \mid  \langle \va, \rho_{T_x D(r), D(s)}(w,v)\rangle = 0 \text{ for any }w\in V^\alg, (T_x D(r), D(s)) \in \CE(2)\}.
\end{align*}
If $N_{2\text{-pt}}^\alg(V)=0$, then $V$ is simple.
\end{cor}
\begin{proof}
It is clear that $N^\alg(V) \subset N_{2\text{-pt}}^\alg(V)$. Hence, by Proposition \ref{prop_criterion_simple}, $V$ is simple.
\end{proof}
In \cite{Mvertex}, the simplicity of the $\CE$-algebra constructed in Theorem \ref{thm_CF} will be proved by showing that
$N_{2\text{-pt}}^\alg(V) = 0$.
By Lemma \ref{lem_state_non_trivial}, we have:
\begin{prop}\label{prop_state_non_trivial}
Let $M \in \Embc$ satisfy $\Embc(\overline{\bD},M) \neq \emptyset$.
If $V$ is simple and $\chi_M$ is a non-zero state of $V$ on $M$, then $N(\chi_M)=0$.
\end{prop}
Hence, any non-zero continuous state on a simple $\CE$-algebra is ``non-degenerate.''

%
%

\subsection{Continuous state on sphere}\label{sec_correlation}

Let $(V,\m)$ be a $\CE$-algebra equipped with a Hilbert space filtration $V= \cup_{k=0}^\infty H^k$.
In this section, we assume that
the monoid homomorphism $\rho:\CE(1) \rightarrow \mathrm{End}(V)$ satisfies the assumptions (U) and (D) in Definition \ref{def_semigroup_nice}.
We will show that the state $\chi_\bD^{\va}$ on the disk extends uniquely to a conformally invariant continuous state on the sphere.
\begin{lem}\label{lem_vacuum_shrink_inv}
Let $g \in \fS$ such that $\overline{g(\bD)} \subset \bD$. Then, $(\va,\rho(g)v)=(\va,v)$ for any $v \in V$.
\end{lem}
\begin{proof}
By Proposition \ref{prop_decomposition}, $g=h_1r^Dh_2$. Since the action of $\fG$ is unitary and $r^D$ is self-adjoint, the assertion follows from $\rho(\phi)\va=\va$ for any $\phi \in \fS$.
\end{proof}

Let $(M,g)$ be a compact Riemannian manifold, which is regarded as $(M,g,M) \in \Embc$.
Then, similarly to Proposition \ref{prop_CE_d_explicit},
\begin{align*}
\Embc(\sqcup_n \overline{\bD}, M)
&=\left\{(f_1,\dots,f_n) \in \Embc(\overline{\bD},M)
\mid \overline{f_i(\bD)}\cap \overline{f_j(\bD)}=\emptyset \text{ for any }i\neq j \right\},
\end{align*}
where $\overline{f_i(\bD)}$ denotes the closure of $f_i(\bD)$ in $M$.
Since the automorphism group of $M$ in $\Embc$ coincides with $\Conf^+(M,g)$ by Proposition \ref{prop_category_mor}, $\Conf^+(M,g)$ acts naturally on this space of disk configurations.

We set $\Embf_n^\bD(S^d) = \Embc(\sqcup_n \overline{\bD}, S^d)$.
By Remark \ref{rem_standard_sphere}, the map $\bD \overset{\si_0}{\hookrightarrow} S^d$ is in $\Embc$. Thus, $\CE$ can be regarded as a subset of $\Embf_n^\bD(S^d)$. In fact,
\begin{align*}
\CE(n) =\left\{(f_1,\dots,f_n) \in \Embf_n^\bD(S^d) \mid f_i(\bD) \subset \si_0(\bD) \right\}.
\end{align*}

%
%
%
Set
\begin{align*}
\Embf_n^\bD(\R^d)&=\{f_{[n]}=(f_1,\dots,f_n) \in \Embf_n^\bD(S^d) \mid \infty \notin \overline{f_i(\bD)} \text{ for any }i \}.
\end{align*}
\begin{lem}\label{lem_shrink}
For any $f_{[n]} \in \Embf_n^\bD(\R^d)$, there exists $r>0$ such that
\begin{align*}
r^D f_{[n]} \in \CE(n),
\end{align*}
that is, $r^D f_i(\bD) \subset \bD$ for any $i=1,\dots,n$.
\end{lem}
\begin{proof}
By assumption, there exists a sufficiently large $R>0$ such that
$\overline{f_i(\bD)} \subset B_R(0)$ for any $i=1,\dots,n$.
Hence $R^{-D}f_i(\bD) \subset B_1(0)=\bD$.
\end{proof}
\begin{rem}
The Cayley transform $\cC$ that sends $\bD$ to the upper half-space
$\mathbb{H}_d=\{x_d > 0\}$ is an element of $\Embc(\overline{\bD},S^d)$ (see Appendix \ref{app_remark}).
Although $\infty \notin \cC(\bD)$, the Cayley transform cannot be made an element of $\CE(1)$
by composing with dilation $r^D$, that is, $r^D \cC \notin \CE(1)$ for all $r$.
Therefore, the closure in the definition of $\Embf_n^\bD(\R^d)$ is necessary.
\end{rem}

%

\begin{lem}\label{lem_reach}
For any $f_{[n]} \in \Embf_n^\bD(S^d)$, there exists $g \in \Conf^+(S^d)$ such that
\begin{align*}
g f_{[n]} \in \Embf_n^\bD(\R^d).
\end{align*}
\end{lem}
\begin{proof}
Since $\overline{f_i(\bD)}$ ($i=1,\dots,n$) are disjoint closed subsets of $S^d$,
there exists a point $x_0 \in S^d$ such that $x_0 \notin \overline{f_i(\bD)}$ for all $i$.
Choose a conformal transformation $g \in \Conf^+(S^d)$ that sends $x_0$ to $\infty$.
Then $\infty \notin \overline{g \circ f_i(\bD)}$ for all $i$, and hence the proposition follows.
\end{proof}

\begin{prop}\label{thm_character}
Let $(V,\m)$ be a $\CE$-algebra in $\Vect$ equipped with a Hilbert space filtration $V= \cup_{k=0}^\infty H^k$
which satisfies (D) and (U) in Definition \ref{def_semigroup_nice}.
Then, there exists a unique state $\chi_S$ on $S^d$ such that:
\begin{enumerate}
\item
$\chi_S$ is conformally invariant.
\item
If $f_{[n]} \in \Embf_n^\bD(S^d)$ satisfies $f_i(\bD) \subset \bD$, that is, $f_{[n]} \in \CE(n)$ then
\begin{align*}
\chi_S(f_{[n]})(a_1,\dots,a_n) = \langle \va,\m_{f_{[n]}}(a_1,\dots,a_n)\rangle
\end{align*}
for any $a_1,\dots,a_n \in V$.
\end{enumerate}
Moreover, the state is continuous and $\chi_{S}(\emptyset)\neq 0$.
\end{prop}
By Lemma \ref{lem_shrink} and Lemma \ref{lem_reach}, the uniqueness of a state satisfying (1) and (2) of Proposition \ref{thm_character} is clear.
In the following, we show the existence.
Note that the subgroup
\begin{align}
\Conf_\infty^+(S^d)=
\{g \in \Conf^+(S^d) \mid g(\infty)=\infty\} \cong \R^d \rtimes (\SO(d)\times \R_{>0})\label{eq_decomp_Matsu}
\end{align}
that fixes $\infty \in S^d$ acts on $\Embf_n^\bD(\R^d)$ \cite[Proposition 1.9]{Matsumoto}.
Let $f_{[n]} = (f_1,\dots,f_n) \in \Embf_n^\bD(\R^d)$.
Then, by Lemma \ref{lem_shrink} for sufficiently small $r>0$,
\begin{enumerate}
\item[(R1)]
$r^D \circ f_{[n]} \in \CE(n)$
holds.
\end{enumerate}
Define a linear map
\begin{align*}
\tilde{\chi}_{f_{[n]}}: V^{\otimes n} \rightarrow \C,\quad \tilde{\chi}_{f_{[n]}}(v_1,\dots,v_n) = (\va, \m_{r^D f_{[n]}}(v_1,\dots,v_n)).
\end{align*}
For any $1>s>0$, since $\rho(s^D)$ is a self-adjoint operator and $\rho(s^D) \va=\va$,
\begin{align*}
(\va, \m_{r^D f_{[n]}}(v_1,\dots,v_n))&=
(\rho(s^D) \va, \m_{r^D f_{[n]}}(v_1,\dots,v_n)) \\
&= (\va, \rho(s^D) \m_{r^D f_{[n]}}(v_1,\dots,v_n))\\
&=(\va, \m_{(sr)^D f_{[n]}}(v_1,\dots,v_n)).
\end{align*}
Hence, $\tilde{\chi}_{f_{[n]}}$ is independent of the choice of $r>0$ which satisfies (R1).

\begin{lem}\label{lem_invariance}
For any $g \in \Conf_\infty^+(S^d)$, $\tilde{\chi}_{f_{[n]}} = \tilde{\chi}_{gf_{[n]}}$ for any $f_{[n]} \in \Embf_n^\bD(\R^d)$.
Moreover, if $g \in \fG$ and $f_{[n]} \in \CE(n)$, then $\tilde{\chi}_{f_{[n]}} = \tilde{\chi}_{gf_{[n]}}$.
\end{lem}
\begin{proof}
Let $g \in \Conf_\infty^+(S^d)$ and $f_{[n]} \in \Embf_n^\bD(\R^d)$.
Let $r>0$ be such that both $r^D f_{[n]}$ and $r^D g f_{[n]}$ are in $\CE(n)$.
If $g \in \SO(d)$ or $g = s^D$, then $\tilde{\chi}_{f_{[n]}} = \tilde{\chi}_{g f_{[n]}}$ clearly holds.
Assume that $g = T_x$, a translation by $x\in \R^d$.
Choose $1>s>0$ sufficiently small so that $s^D T_{r(x)} \in \fS$.
Then, by Lemma \ref{lem_vacuum_shrink_inv},
\begin{align*}
\tilde{\chi}_{T_x \circ f_{[n]}}(v_1,\dots,v_n)&= (\va, \m_{r^Ds^DT_x f_{[n]}}(v_1,\dots,v_n))\\
&=(\va, \m_{s^DT_{rx}r^D f_{[n]}}(v_1,\dots,v_n))\\
&=(\va, \rho(s^DT_{rx})\m_{r^D f_{[n]}}(v_1,\dots,v_n))\\
&=(\va, \m_{r^D f_{[n]}}(v_1,\dots,v_n)).
\end{align*}
Hence, the assertion follows from \eqref{eq_decomp_Matsu}


If $g \in \fG$ and $f_{[n]} \in \CE(n)$, then 
\begin{align*}
\tilde{\chi}_{g\circ f_{[n]}}(v_1,\dots,v_n)&= (\va, \m_{g \circ f_{[n]}}(v_1,\dots,v_n))\\
&= (\va, \rho(g)\m_{f_{[n]}}(v_1,\dots,v_n)) = (\rho(g)^{-1}\va, \m_{f_{[n]}}(v_1,\dots,v_n)) \\
&= \tilde{\chi}_{f_{[n]}}(v_1,\dots,v_n).
\end{align*}
\end{proof}

\begin{lem}\label{lem_decomposition}
Let $g \in \Conf^+(S^d)$.
\begin{enumerate}
\item
If $\infty \notin g(\pa\bD)$, then $g= T_xr^{D}h$ for some $x \in \R^d$, $r >0$ and $h \in \fG_d$.
\item
If $\infty \in g(\pa\bD)$, then $g= h_1T_xr^{D}h_2$ for some $x \in \R^d$, $r >0$ and $h_1,h_2 \in \fG_d$.
\end{enumerate}
\end{lem}
\begin{proof}
Assume $\infty \notin g(\pa\bD)$. Then, $g(\pa \bD)$ is a circle in $\R^d$ with the center $x$ and the radius $r>0$.
Then, $r^{-D}T_{-x}g(\bD)=\bD$. Hence, $h=r^{-D}T_{-x}g \in \fG_d$.

Assume $\infty \in g(\pa\bD)$. Then, $g(\pa\bD)$ is a hyperplane in $\R^d$ (which pass through $\infty$).
It is obvious that there exists $h \in \fG_d$ such that $\infty \notin h g(\pa\bD)$.
Hence, by (1), $hg = T_xr^{D}h'$ for some $x \in \R^d$, $r>0$ and $h' \in \fG$.
\end{proof}

\begin{lem}\label{lem_char_disk}
Let $f_{[n]} \in \CE(n)$ and $g \in \Conf^+(S^d)$. Assume $g f_{[n]} \in \CE(n)$. Then, $\tilde{\chi}_{f_{[n]}} = \tilde{\chi}_{gf_{[n]}}$ holds.
\end{lem}
\begin{proof}
First, assume that $\infty \notin g(\pa \bD)$. Then, by Lemma \ref{lem_decomposition}, $g=T_x r^D h$ for some $x\in \R^d$, $r>0$ and $h \in \fG$. 
Since $h f_{[n]} \in \CE(n)$ and $T_x r^D \in \Conf_\infty^+(S^d)$, by Lemma \ref{lem_invariance}, $\tilde{\chi}_{gf_{[n]}} = \tilde{\chi}_{T_x r^D h f_{[n]}}=\tilde{\chi}_{h f_{[n]}} = \tilde{\chi}_{f_{[n]}}$.
Next, assume that $\infty \in g(\pa \bD)$ and let $g = h_1 T_x r^D h_2$ be the decomposition in Lemma \ref{lem_decomposition}.
By the assumption, $g f_{[n]} \in \CE(n)$, which implies $h_1^{-1} g f_{[n]} \in \CE(n)$.
Hence, by Lemma \ref{lem_invariance},
$ \tilde{\chi}_{gf_{[n]}} = \tilde{\chi}_{h_1T_x r^D h_2 f_{[n]}} = \tilde{\chi}_{T_x r^D h_2 f_{[n]}}=\tilde{\chi}_{h_2 f_{[n]}}=\tilde{\chi}_{f_{[n]}}$. Hence, the assertion holds.
\end{proof}

\begin{lem}\label{lem_invariance_all}
Let $f_{[n]} \in\Embf_n^\bD(\R^d)$ and $g \in \Conf^+(S^d)$. Assume $gf_{[n]} \in \Embf_n^\bD(\R^d)$.
Then, $\tilde{\chi}_{f_{[n]}} = \tilde{\chi}_{gf_{[n]}}$ holds.
\end{lem}
\begin{proof}
Let $r>0$ such that both $r^D f_{[n]}$ and $r^D gf_{[n]}$ are in $\CE(n)$.
By Lemma \ref{lem_invariance}, 
$\tilde{\chi}_{g f_{[n]}} = \tilde{\chi}_{r^D g f_{[n]}} =  \tilde{\chi}_{r^D gr^{-D} r^Df_{[n]}}$.
Since both $r^D f_{[n]}$ and $(r^D g r^{-D})\, r^D f_{[n]}$ are in $\CE(n)$, by applying Lemma \ref{lem_char_disk} we obtain $\tilde{\chi}_{r^D g r^{-D} r^D f_{[n]}}=\tilde{\chi}_{r^D f_{[n]}}$.
Hence, the assertion holds.
\end{proof}

\begin{proof}[Proof of Proposition \ref{thm_character}]
Let $f_{[n]} \in \Embf_n^\bD(S^d)$ and define $\chi_{f_{[n]}}:V^{\otimes n} \rightarrow \C$
by $\chi_{f_{[n]}}(a_1,\dots,a_n)=\tilde{\chi}_{g f_{[n]}}(a_1,\dots,a_n)$,
where $g \in \Conf^+(S^d)$ satisfies $g f_{[n]} \in \Embf_n^\bD(\R^d)$.
By Lemma \ref{lem_reach}, such an element $g \in \Conf^+(S^d)$ always exists.
Suppose that we take another element $g' \in \Conf^+(S^d)$ with $g' f_{[n]} \in \Embf_n^\bD(\R^d)$.
Applying Lemma \ref{lem_invariance_all} to $g' g^{-1} \in \Conf^+(S^d)$, we obtain
$\tilde{\chi}_{g f_{[n]}}= \tilde{\chi}_{(g'g^{-1}) g f_{[n]}}$.
Hence, $\chi$ is independent of the choice of $g$ and is well defined.
Conformal invariance follows as well: for any $h \in \Conf^+(S^d)$, if we choose $g'$ for $h f_{[n]}$ to be $g h^{-1}$ using the above $g$,
then we obtain $\chi_{h f_{[n]}} = \chi_{f_{[n]}}$.
Continuity is clear from the boundedness of $\m$.
\end{proof}

\begin{thm}\label{cor_left_nonzero}
Let $(V,\m)$ be a $\CE$-algebra equipped with a Hilbert space filtration.
Assume that $(V,\m)$ satisfies the conditions (U), (D) and $\dim V_0 =1$. 
Then, the space of conformally invariant continuous states on $S^d$
\begin{align*}
\mathrm{Hom}_{\Ind}\left(\mathrm{Lan}_V(S^d),\C\right)^{\Conf^+(S^d)}
\end{align*}
is a one-dimensional vector space spanned by $\chi_S$.
\end{thm}
\begin{proof}
Let $\chi$ be a conformally invariant continuous states on $S^d$.
Then, the restriction of $\chi$ on $\si_0:\bD \hookrightarrow S^d$ is a continuous state on $\bD$, which is given by 
$\chi(\si_0):V \rightarrow \C$ for $\si_0 \in \Embc(\overline{\bD},S^d)$. By definition of $r^D$, and the conformal invariance,
\begin{align*}
\chi(\si_0) = \chi(r^D \circ \si_0) = \chi(\si_0 \circ r^D) = \chi(\si_0) \circ \rho(r^D)
\end{align*}
as a linear map $V \rightarrow \C$ for any $r\in (0,1)$. Hence, $\chi(\si_0)=(a\va,\bullet)$ with some $a\in\C$ by $\dim V_0=1$
and $\chi|_{\si_0}= a\chi_\bD^{\va}$. By Lemma \ref{lem_shrink} and Lemma \ref{lem_reach}, $\chi = a \chi_S$.
\end{proof}

\subsection{Correlation functions on configuration spaces}
\label{sec_cor}
Let $(V,\m)$ be a $\CE$-algebra equipped with a Hilbert space filtration $V= \cup_{k=0}^\infty H^k$
and assume that $(V,\m)$ satisfies the condition (D).
In this section, we extract from a $\CEd$-algebra a family of multilinear maps
parametrized by configuration spaces. After pairing with vectors in $V^\alg$, these maps
give the correlation functions used in physics.

For $d=2$, such systems of functions were introduced in \cite{HK1} and used to formulate two-dimensional, non-chiral conformal field theory (see also \cite{MSwiss,AMT}, where, under several assumptions, it is proved that such systems of correlators satisfy the Osterwalder--Schrader axioms).
Based on the results of this section, \cite{Mvertex} extends \cite{HK1} to higher dimensions and gives a formulation of $d$-dimensional conformal field theory in terms of systems of correlation functions.

Set
\begin{align*}
V^\alg = \bigoplus_{\Delta \geq 0}V_\Delta,\qquad \overline{V^\alg}= \prod_{\Delta \geq 0}V_\Delta,
\end{align*}
where these denote the algebraic direct sum and the direct product of vector spaces, respectively (see Definition \ref{def_algebraic_core}).
Then $V$ is naturally regarded as a subspace of $\overline{V^\alg}$, and satisfies
$V^\alg \subset V \subset \overline{V^\alg}$.
For any $s>0$, define
\begin{align*}
D^\alg(s):\overline{V^\alg} \rightarrow \overline{V^\alg}
\end{align*}
by $D^\alg(s)(v_\Delta)_{\Delta \geq 0}=(s^{\Delta}v_\Delta)_{\Delta \geq 0}$ for $(v_\Delta)_{\Delta \geq 0} \in \prod_{\Delta \geq 0}V_\Delta$.
If $1\geq s>0$, then $D^\alg(s)|_V = D(s)$ holds.
On the other hand, if $s>1$, then $D^\alg(s)$ is an unbounded operator, and hence the action does not preserve $V \subset \overline{V^\alg}$.
We denote by $(-,-)_V$ the inner product on $V$ defined by the Hilbert space filtration,
which induces a natural pairing
\begin{align*}
\langle-,-\rangle_V:V^\alg \otimes \overline{V^\alg} \rightarrow \C.
\end{align*}
For a ball $B_r(a) \subset \bD$, set
\begin{align*}
T_{a,r}:\bD \rightarrow \bD,\qquad x\mapsto rx+a.
\end{align*}
Let
\begin{align*}
\Conf_n(\R^d) = \{(x_1,\dots,x_n) \in (\R^d)^n \mid x_i \neq x_j \},
\end{align*}
be the configuration space of points.
Let $(x_1,\dots,x_n) \in \Conf_n(\R^d)$.
Choose $R >0$ and $r_i>0$ ($i=1,\dots,n$) so that the following condition is satisfied:
\begin{itemize}
\item[(C1)]
$\overline{B_{r_i}(x_i)} \cap \overline{B_{r_j}(x_j)}=\emptyset$ and $B_{R^{-1}r_i}(R^{-1}x_i) \subset \bD$.
\end{itemize}
Then define a linear map
\begin{align*}
Y_{x_1,\dots,x_n}(\bullet):(V^\alg)^{\otimes n} \rightarrow \overline{V^\alg}
\end{align*}
by
\begin{align}
Y_{x_1,\dots,x_n}(a_1\otimes \cdots \otimes a_n) = D^\alg(R)\rho_{T_{R^{-1} x_1, R^{-1} r_1},\dots,T_{R^{-1} x_n,R^{-1}r_n}}(D^\alg(r_1^{-1}) a_1\otimes\cdots \otimes  D^\alg(r_n^{-1}) a_n).
\label{add_Y_def}
\end{align}

\begin{lem}\label{lem_Y_well}
The linear map $Y_{x_1,\dots,x_n}(\bullet):(V^\alg)^{\otimes n} \rightarrow \overline{V^\alg}$ is independent of the choice of $R,r_1,\dots,r_n$ satisfying (C1).
\end{lem}
\begin{proof}
Let $1>t_1,\dots,t_n >0$. Since $(T_{R^{-1} x_1,R^{-1}r_1},\dots,T_{R^{-1}x_n,R^{-1} r_n}) \in \CE(n)$ and $D(t_i) \in \CE(1)$, we have
\begin{align*}
&\rho_{T_{R^{-1} x_1,R^{-1}t_1r_1},\dots,T_{R^{-1} x_n,R^{-1}t_nr_n}}(D^\alg((t_1 r_1)^{-1}) a_1\otimes\cdots \otimes  D^\alg((t_nr_n)^{-1}) a_n)\\
&=\rho_{T_{R^{-1} x_1,R^{-1}r_1},\dots,T_{R^{-1}x_n,R^{-1}r_n}}
\circ (\rho_{D(t_1)} \otimes \cdots \otimes \rho_{D(t_n)} )
(D^\alg((t_1 r_1)^{-1})  a_1\otimes\cdots \otimes  D^\alg((t_nr_n)^{-1}) a_n)\\
&=\rho_{T_{R^{-1} x_1,R^{-1} r_1},\dots,T_{R^{-1}x_n,R^{-1}r_n}}
(D^\alg((r_1)^{-1})  a_1\otimes\cdots \otimes  D^\alg((r_n)^{-1}) a_n).
\end{align*}
Thus replacing each $r_i$ by a smaller positive real number does not change $Y$.
Similarly, replacing $R$ by a larger positive real number does not change $Y$.
Hence, the assertion follows.
\end{proof}

Next, we study the property satisfied by the compositions of the maps $Y$, corresponding to the composition of products in the $\CE$-algebra \eqref{eq_def_of_operad}.
Let $n,m \geq 1$, $(x_1,\dots,x_{n+1}) \in \Conf_{n+1}(\R^d)$, and $(y_1,\dots,y_{m}) \in \Conf_{m}(\R^d)$.
We assume that
\begin{align}
\min_{i=1,\dots,n}|x_i-x_{n+1}| > \max_{j=1,\dots,m}|y_j|.
\label{eq_min_max}
\end{align}
Note that \eqref{eq_min_max} immediately implies $(x_1,\dots,x_n,x_{n+1}+y_1,\dots,x_{n+1}+y_m) \in \Conf_{n+m}(\R^d)$.
We would like to prove
\begin{align}
Y_{x_1,\dots,x_{n+1}}(\bullet) \circ_{n+1} Y_{y_1,\dots,y_{m}}(\bullet) =Y_{x_1,\dots,x_n,x_{n+1}+y_1,\dots,x_{n+1}+y_m}(\bullet),\label{add_naive}
\end{align}
but since the domain of $Y$ is $V^\alg$ and its codomain is $\overline{V^\alg}$, \eqref{add_naive} does not make literal sense.
For $\Delta \geq 0$, let $\pr_\Delta:\overline{V^\alg} \rightarrow V_\Delta$ be the projection.
Then $\pr_\Delta Y_{y_1,\dots,y_{m}}(\bullet)$ takes values in $V^\alg$.
Therefore, we formulate \eqref{add_naive} in the following way: after applying the projections, the resulting sum converges absolutely, and its limit agrees with the right-hand side.

\begin{thm}\label{thm_Y}
Let $(V,\m)$ be a $\CE$-algebra equipped with a Hilbert space filtration $V= \cup_{k=0}^\infty H^k$
and assume that $(V,\m)$ satisfies the condition (D).
Then, the family of maps
\begin{align*}
Y_{\bullet}(\bullet):\Conf_n(\R^d) \rightarrow \mathrm{Hom}_\C((V^\alg)^{\otimes n},\overline{V^\alg})\qquad (n \geq 1)
\end{align*}
given in \eqref{add_Y_def} satisfies the following conditions:
\begin{enumerate}
\item[S1)]
For any $n \geq 2$, $(x_1,\dots,x_n) \in \Conf_n(\R^d)$ and $a_1,\dots,a_{n-1} \in V^\alg$,
\begin{align*}
Y_{x_1,\dots,x_n}(a_1,\dots,a_{n-1},\va)=Y_{x_1,\dots,x_{n-1}}(a_1,\dots,a_{n-1}).
\end{align*}
\item[S2)]
For any $(x_1,\dots,x_n) \in \Conf_n(\R^d)$, $a_1,\dots,a_{n} \in V^\alg$ and any permutation $\si \in S_n$,
\begin{align*}
Y_{x_{\si 1},\dots,x_{\si n}}(a_{\si 1},\dots,a_{\si n})=Y_{x_1,\dots,x_{n}}(a_1,\dots,a_{n}).
\end{align*}
\item[S3)]
Let $n,m \geq 1$, $(x_1,\dots,x_{n+1}) \in \Conf_{n+1}(\R^d)$ and  $(y_1,\dots,y_{m}) \in \Conf_{m}(\R^d)$ satisfying 
\begin{align*}
\min_{i=1,\dots,n}|x_i-x_{n+1}| > \max_{j=1,\dots,m}|y_j|.
\end{align*}
Then, for any $u,a_1,\dots,a_{n+m} \in V^\alg$, the sum
\begin{align}
\sum_{\Delta \geq 0}\left\langle u, Y_{x_1,\dots,x_{n+1}}(a_1,\dots,a_n, \pr_\Delta Y_{y_1,\dots,y_{m}}(a_{n+1},\dots,a_{n+m}))\right \rangle_V
\label{eq_Y_sum}
\end{align}
converges absolutely, and its limit is
\begin{align*}
\left \langle u,Y_{x_1,\dots,x_n,x_{n+1}+y_1,\dots,x_{n+1}+y_m}(a_1,\dots,a_{n+m})\right\rangle_V.
\end{align*}
\end{enumerate}
\end{thm}
\begin{proof}
The conditions (S1) and (S2) follow immediately from the definition of a $\CE$-algebra.
We prove (S3).
Let $P>0$ satisfy  $\min_{i=1,\dots,n}|x_i-x_{n+1}| > P>\max_{j=1,\dots,m}|y_j|$.
Then there exist $r_1,\dots,r_{n},s_1,\dots,s_{m}>0$ such that
\begin{align*}
\overline{B_{r_i}(x_i)} \cap \overline{B_{r_j}(x_j)}=\emptyset,\qquad \overline{B_{r_i}(x_i)} \cap \overline{B_P(x_{n+1})} = \emptyset
\end{align*}
and
\begin{align*}
\overline{B_{s_k}(y_k)} \cap \overline{B_{s_l}(y_l)}=\emptyset,\qquad {B_{s_k}(y_k)} \subset B_P(0).
\end{align*}
Choose sufficiently large $Q>0$ such that $B_{Q^{-1}r_i}(Q^{-1}x_i), B_{Q^{-1}P}(Q^{-1}x_{n+1}) \subset \bD$.
Then $Q,r_1,\dots,r_{n},P$ and $P,s_1,\dots,s_{m}$ satisfy (C1) for $x_1,\dots,x_n,x_{n+1}$ and $y_1,\dots,y_m$, respectively.
In $\CE$, we have
\begin{align*}
&(T_{Q^{-1}x_1,Q^{-1}r_1},\dots, T_{Q^{-1}x_n,Q^{-1}r_n},T_{Q^{-1}x_{n+1},Q^{-1}P}) \circ_{n+1}
(T_{P^{-1}y_1,P^{-1}s_1},\dots,T_{P^{-1}y_m,P^{-1}s_m})\\
&=(T_{Q^{-1}x_1,Q^{-1}r_1},\dots, T_{Q^{-1}x_n,Q^{-1}r_n}, T_{Q^{-1}(x_{n+1}+y_1), Q^{-1}s_1},\dots,  T_{Q^{-1}(x_{n+1}+y_m), Q^{-1}s_m}).
\end{align*}
Therefore, as linear maps on $(V^\alg)^{\otimes (n+m)}$
\begin{align*}
&Y_{x_1,\dots,x_n,x_{n+1}+y_1,\dots,x_{n+1}+y_m}(\bullet) \\
&= D^\alg(Q)
\rho_{(T_{Q^{-1}x_1,Q^{-1}r_1},\dots, T_{Q^{-1}x_n,Q^{-1}r_n}, T_{Q^{-1}(x_{n+1}+y_1), Q^{-1}s_1},\dots,  T_{Q^{-1}(x_{n+1}+y_m), Q^{-1}s_m})}\\
&\circ (D^\alg(r_1^{-1}) \otimes \cdots \otimes D^\alg(r_n^{-1}) \otimes D^\alg(s_1^{-1}) \otimes \cdots \otimes D^\alg(s_m^{-1}) )
\\
&= D^\alg(Q)
\rho_{T_{Q^{-1}x_1,Q^{-1}r_1},\dots, T_{Q^{-1}x_n,Q^{-1}r_n},T_{Q^{-1}x_{n+1},Q^{-1}P}}\circ  (D^\alg(r_1^{-1}) \otimes \cdots \otimes D^\alg(r_n^{-1})\otimes \id)\\
& \circ_{n+1}
\rho_{T_{P^{-1}y_1,P^{-1}s_1},\dots,T_{P^{-1}y_m,P^{-1}s_m}}\circ (D^\alg(s_1^{-1}) \otimes \cdots \otimes D^\alg(s_m^{-1}) ).
\end{align*}
Then, for each $\Delta \geq 0$,
\begin{align*}
& D^\alg(Q)
\rho_{T_{Q^{-1}x_1,Q^{-1}r_1},\dots, T_{Q^{-1}x_n,Q^{-1}r_n},T_{Q^{-1}x_{n+1},Q^{-1} P}}\circ  (D^\alg(r_1^{-1}) \otimes \cdots \otimes D^\alg(r_n^{-1})\otimes \id)\\
& \circ_{n+1} \pr_\Delta
\rho_{T_{P^{-1}y_1,P^{-1}s_1},\dots,T_{P^{-1}y_m,P^{-1}s_m}}\circ (D^\alg(s_1^{-1}) \otimes \cdots \otimes D^\alg(s_m^{-1}) )\\
&= D^\alg(Q)
\rho_{T_{Q^{-1}x_1,Q^{-1}r_1},\dots, T_{Q^{-1}x_n,Q^{-1} r_n},T_{Q^{-1} x_{n+1},Q^{-1} P}}\circ  (D^\alg(r_1^{-1}) \otimes \cdots \otimes D^\alg(r_n^{-1})\otimes D^\alg(P^{-1}))\\
& \circ_{n+1} D^\alg(P) \pr_\Delta \rho_{T_{P^{-1}y_1,P^{-1}s_1},\dots,T_{P^{-1}y_m,P^{-1}s_m}}\circ (D^\alg(s_1^{-1}) \otimes \cdots \otimes D^\alg(s_m^{-1}) )\\
&= Y_{x_1,\dots,x_{n+1}}(\bullet) \circ_{n+1} \pr_\Delta
Y_{y_1,\dots,y_{m}}(\bullet).
\end{align*}

Set $a_i'=D^\alg(r_i^{-1})(a_i)$ and $a_{n+j}'=D^\alg(s_j^{-1})a_{n+j}$, and set
\begin{align*}
A&= \rho_{T_{P^{-1}y_1,P^{-1}s_1},\dots,T_{P^{-1}y_m,P^{-1}s_m}}\left(a_{n+1}',\dots,a_{n+m}'\right),\\
A_\Delta &= \pr_\Delta A.
\end{align*}
Choose $k_i \geq 0$ such that $a_i \in H^{k_i}$.
By the definition of the Hilbert space filtration, there exists $K\geq 0$ such that $A,A_\Delta \in H^K$.
Moreover, there exists $L$ such that
\begin{align}
\rho_{T_{Q^{-1}x_1,Q^{-1}r_1},\dots, T_{Q^{-1}x_n,Q^{-1}r_n},T_{Q^{-1}x_{n+1},Q^{-1} P}}:H^{k_1} \hat{\otimes} \cdots \hat{\otimes} H^{k_n} \hat{\otimes} H^K \rightarrow H^L\label{add_bound_rho}
\end{align}
is a bounded linear map.
Enlarging $L$ if necessary, we may assume that $u \in H^L$.
Since $\sum_{\Delta \geq 0} A_\Delta$ converges to $A$ in the norm of $H^K$, the boundedness of $\rho$ implies that
\begin{align*}
&\rho_{(T_{Q^{-1}x_1,Q^{-1}r_1},\dots, T_{Q^{-1}x_n,Q^{-1}r_n}, T_{Q^{-1}(x_{n+1}+y_1), Q^{-1}s_1},\dots,  T_{Q^{-1}(x_{n+1}+y_m), Q^{-1}s_m})}(a_1',\dots,a_n',a_{n+1}',\dots,a_{n+m}')\\
&=\rho_{T_{Q^{-1}x_1,Q^{-1}r_1},\dots, T_{Q^{-1}x_n,Q^{-1}r_n},T_{Q^{-1}x_{n+1},Q^{-1} P}}
\left(a_1',\dots,a_n', A\right)\\
&=\rho_{T_{Q^{-1}x_1,Q^{-1}r_1},\dots, T_{Q^{-1}x_n,Q^{-1}r_n},T_{Q^{-1}x_{n+1},Q^{-1} P}}
\left(a_1',\dots,a_n', \sum_{\Delta\geq 0} A_\Delta \right)\\
&=\sum_{\Delta\geq 0}\rho_{T_{Q^{-1}x_1,Q^{-1}r_1},\dots, T_{Q^{-1}x_n,Q^{-1}r_n},T_{Q^{-1}x_{n+1},Q^{-1} P}}
\left(a_1',\dots,a_n', A_\Delta \right)
\end{align*}
holds in the norm of $H^L$.
In particular, pairing with $D^\alg(Q) u \in H^L$ gives weak convergence, and hence the desired identity follows.

It remains to prove absolute convergence.
Let
\begin{align*}
\rho_{T_{Q^{-1}x_1,Q^{-1}r_1},\dots, T_{Q^{-1}x_n,Q^{-1}r_n},T_{Q^{-1}x_{n+1},Q^{-1} P}}^*:H^L \rightarrow H^{k_1} \hat{\otimes} \cdots \hat{\otimes} H^{k_n} \hat{\otimes} H^K
\end{align*}
be the adjoint operator of \eqref{add_bound_rho}.
Then
\begin{align*}
&\left| 
\left\langle u, D^\alg(Q) \rho_{T_{Q^{-1}x_1,Q^{-1}r_1},\dots, T_{Q^{-1}x_n,Q^{-1}r_n},T_{Q^{-1}x_{n+1},Q^{-1} P}}
\left(a_1',\dots,a_n', A_\Delta\right)\right\rangle_V\right|\\
&=\left|\left(D^\alg(Q) u, \rho_{T_{Q^{-1}x_1,Q^{-1}r_1},\dots, T_{Q^{-1}x_n,Q^{-1}r_n},T_{Q^{-1}x_{n+1},Q^{-1} P}}
\left(a_1',\dots,a_n', A_\Delta\right)\right)_{H^L}\right|\\
&=
\left|
\left(\rho_{T_{Q^{-1}x_1,Q^{-1}r_1},\dots, T_{Q^{-1}x_n,Q^{-1}r_n},T_{Q^{-1}x_{n+1},Q^{-1} P}}^* \left(D^\alg(Q) u\right), a_1' \otimes \cdots \otimes a_n' \otimes A_\Delta\right)_{H^{k_1} \hat{\otimes} \cdots \hat{\otimes} H^{k_n} \hat{\otimes} H^K}\right|.
\end{align*}
Set $U= \rho_{T_{Q^{-1}x_1,Q^{-1}r_1},\dots, T_{Q^{-1}x_n,Q^{-1}r_n},T_{Q^{-1}x_{n+1},Q^{-1} P}}^* \left(D^\alg(Q) u\right) \in H^{k_1} \hat{\otimes} \cdots \hat{\otimes} H^{k_n} \hat{\otimes} H^K$.
Since $A=\sum_{\Delta \geq 0} A_\Delta \in H^K$ is an orthogonal decomposition, we have
\begin{align*}
\sum_{\Delta \geq 0}\left|
\left(U, a_1' \otimes \cdots \otimes a_n' \otimes A_\Delta\right)\right|
&= \sum_{\Delta \geq 0}
\left|
\left((\id_{H^{k_1}}\otimes \cdots \otimes \id_{H^{k_n}}\otimes \pr_\Delta) U, a_1' \otimes \cdots \otimes a_n' \otimes A_\Delta\right)\right|\\
&\leq \sum_{\Delta \geq 0}
\norm{(\id_{H^{k_1}}\otimes \cdots \otimes \id_{H^{k_n}}\otimes \pr_\Delta) U}
\norm{a_1' \otimes \cdots \otimes a_n' \otimes A_\Delta}.
\end{align*}
Using the Cauchy--Schwarz inequality, we obtain
\begin{align*}
&\leq \left(\sum_{\Delta \geq 0}
\norm{(\id_{H^{k_1}}\otimes \cdots \otimes \id_{H^{k_n}}\otimes \pr_\Delta) U}^2\right)^\ft
\left(\sum_{\Delta \geq 0}\norm{a_1' \otimes \cdots \otimes a_n' \otimes A_\Delta}^2
\right)^\ft\\
&=\norm{U}\norm{a_1'}\dots\norm{a_n'}\norm{A}<\infty.
\end{align*}
Therefore, \eqref{eq_Y_sum} converges absolutely.

\end{proof}

\appendix

\section{Remark on Global conformal group}\label{app_remark}
In Appendix \ref{app_remark}, we provide proofs of several propositions from Section \ref{sec_conf_sphere}.
Set
\begin{align*}
\cC=\left( 
\begin{array}{c|ccc}
I_{d-1}& &0& \\ \hline
& 0 & -\frac{1}{\sqrt{2}} & -\frac{1}{\sqrt{2}}\\
0&\frac{1}{\sqrt{2}}&\ft&-\ft\\
&\frac{1}{\sqrt{2}}&-\ft&\ft
\end{array}
\right),
\end{align*}
which is an element of $\mathrm{O}^+(d+1,1)$, and acts on $\mathbb{P}(\mathcal{N})$ by
\begin{align*}
\cC(\si_0(x_1,\dots,x_d)) = \left(x_d+\frac{1+|x|^2}{2}\right) \si_0\left(\frac{x_1}{x_d+\frac{1+|x|^2}{2}},\dots, \frac{x_{d-1}}{x_d+\frac{1+|x|^2}{2}}, \frac{\frac{-1+|x|^2}{2}}{x_d+\frac{1+|x|^2}{2}}\right).
\end{align*}
The restriction of $\cC$ to $x_d=0$ coincides with the planar projection
$\R^{d-1}\rightarrow S^{d-1} \subset \R^d \setminus \{(0,\dots,0,1)\}$,
\begin{align*}
(x_1,\dots,x_{d-1},0) &\mapsto \left(\frac{2x_1}{1+|x|^2},\dots,\frac{2x_{d-1}}{1+|x|^2},\frac{-1+|x|^2}{1+|x|^2}\right).
\end{align*}
In particular, $\cC$ is a generalization of the \textbf{Cayley transform} that maps the upper half-space
$\bH_d = \{(x_1,\dots,x_d)\in\R^d \mid x_d >0 \}$ into the interior of the unit disk.
Moreover, under $\cC$, the center of the unit disk is sent to $(0,\dots,0,1) \in \bH_d$.
Set
\begin{align*}
\fG_C &=  \{g \in \SO^+(d+1,1) \mid g(\bH_d) =\bH_d \},\\
\fS_C &=  \{g \in \SO^+(d+1,1) \mid g(\bH_d) \subset \bH_d\}.
\end{align*}
The sets $\fG$ and $\fS$ are conjugate to $\fG_C$ and $\fS_C$ via the Cayley transform, and the latter description is often more convenient.
Note that
\begin{align*}
J_C =\cC^{-1}J \cC =\left( 
\begin{array}{c|ccc}
I_{d-1}& &0& \\ \hline
& -1 & 0 &0\\
0&0&1&0\\
&0&0&1
\end{array}
\right).
\end{align*}
A matrix commuting with $J_C$ has all off-diagonal entries in the $d$-component equal to zero.
If $g \in \mathrm{O}^+(d+1,1)$ satisfies $J_C g = g J_C$, then its diagonal entry is $1$ or $-1$;
if it is $1$ (resp. $-1$), then $g$ preserves $\bH_d$ (resp. sends $\bH_d$ to $-\bH_d$).
Hence, we have:
\begin{prop}\label{prop_app_SO}
The following properties hold:
\begin{enumerate}
\item
For $g \in \mathrm{O}^+(d+1,1)$, preserving $\{x_d=0\}$ is equivalent to $gJ_C=J_C g$.
\item
The group $\{g \in \mathrm{O}^+(d+1,1) \mid gJ_C=J_C g \}$ has two connected components and coincides with
$\fG_C \sqcup \fG_C J_C$. 
Moreover, under the embedding $\SO^+(d,1) \hookrightarrow\SO^+(d+1,1)$ that fixes the $d$-component identically,
$\fG_C = \SO^+(d,1)$.
\item
If $g \in G_\C$ restricts to the identity map on $\{x_d=0\}$, then $g = \id$.
\end{enumerate}
\end{prop}
Proposition \ref{prop_fixed_point} follows from this proposition and the Cayley transform.
Set
\begin{align*}
C_r(x_0) = \{x\in \R^d \mid |x-x_0|=r\}.
\end{align*}
\begin{prop}\label{prop_app_circle}
If $C_r(x_0) \subset \bD$, then there exist $g \in \fG$ and $1>R>0$ such that
\begin{align*}
g(C_r(x_0)) = C_R(0).
\end{align*}
\end{prop}
\begin{proof}
It suffices to work on the upper half-space via the Cayley transform.
It is easy to see that
\begin{align*}
\cC^{-1}(C_R(0))= \left\{(x',x_d) \in \bH_d \mid |x'|^2 + \left(x_d- \frac{1+R^2}{1-R^2}\right)^2 = \left(\frac{2R}{1-R^2}\right)^2 \right\}.
\end{align*}
The set $\cC^{-1}(C_r(x_0))$ is a circle in $\bH_d$ with radius $r'$ and center $(y',y_d) \in \bH_d$.
Moreover, by assumption we have $y_d > r'$.
By some translation and dilation $T_{y'},D(\lambda) \in \fG_C$, this circle can be moved to one with center $(0,\lambda(y_d))$
and radius $\lambda r'$.
Choosing $\lambda=\frac{1}{\sqrt{y_d^2-r'^2}}$, this circle coincides with $\cC^{-1}(C_R(0))$ for some $1>R>0$.
\end{proof}

Using Proposition \ref{prop_app_circle}, we can prove Proposition \ref{prop_decomposition}.
\begin{proof}[Proof of Proposition \ref{prop_decomposition}]
Since a M\"{o}bius transformation sends circles to circles, $g(\pa \bD)$ is a circle contained in $\overline{\bD}$.
Let $x_0$ be its center and $r$ its radius; then $T_{x_0}r^D \in \fS$ is clear.
The element $(T_{x_0}r^D)^{-1}g$ preserves $\pa \bD$ and sends the interior of $\bD$ into itself, hence it lies in $\fG$.

Next, assume that $\overline{g(\bD)} \subset \bD$.
Then $g(\pa \bD)$ is a circle in the interior of ${\bD}$.
By Proposition \ref{prop_app_circle}, there exist $h_1 \in \fG$ and $1> r >0$ such that
$h_1 g(\pa \bD) = \{x \in \R^d \mid |x|=r \}$.
Hence $r^{-D}h_1g \in \fG$.
\end{proof}
\begin{rem}\label{rem_circle_intersection}
In Proposition \ref{prop_app_circle}, if $C_r(x_0) \subset \overline{\bD}$ and the circle has a non-empty intersection with $\pa \bD$,
that is, if $r+|x_0|=1$, then the proof of Proposition \ref{prop_app_circle} no longer works since $y_d =r'$.
In particular, if $g \in \fS$ and $\overline{g(\bD)}$ is not contained in $\bD$, then the decomposition in Proposition \ref{prop_decomposition} (2) does not hold.
\end{rem}

We conclude this section by proving Proposition \ref{prop_conf_identity}.
Equation \eqref{eq_identity_all} follows immediately from Lemma \ref{lem_norm_all_inv} and Proposition \ref{prop_Liouville_sphere}.
We will prove \eqref{eq_identity_G}.

%
%

Let $H_C:\bH_d \times \bH_d \rightarrow \R_{\geq 0}$ be a map defined by 
\begin{align}
H_C(x',y') = \frac{|x'-y'|^2}{x_d' y_d'}\qquad \text{ for $x',y' \in \bH_d$}
\label{eq_app_def_HC}
\end{align}
and $H:\bD_d \times \bD_d \rightarrow \R_{\geq 0}$ by $H(x,y) = H_C(\cC^{-1}(x), \cC^{-1}(y))$ for $x,y \in \bD_d$.
Since 
\begin{align*}
\cC^{-1}(x) =\left( \frac{x_1}{-x_d+\frac{|x|^2+1}{2}},\dots,\frac{x_{d-1}}{-x_d+\frac{|x|^2+1}{2}}, \frac{1-|x|^2}{-2x_d+1+|x|^2}\right),
\end{align*}
a straightforward computation shows that
\begin{align}
H(x,y)=\frac{|\cC^{-1} x- \cC^{-1} y|^2}{(\cC^{-1}x)_d(\cC^{-1}y)_d} &=
\frac{4|x-y|}{(1-|x|^2)(1-|y|^2)}.\label{eq_app_d_D}
\end{align}
Note that $H_C$ is related to the hyperbolic distance on the upper half-space, while $H$ is related to the distance on the Poincar\'{e} disk.

\begin{lem}\label{lem_app_identity}
For any $g \in \fG_C$ and $x,y \in \bH_d$, $H_C(gx,gy)=H_C(x,y)$.
\end{lem}
\begin{proof}
Let $g \in \fG_C$. Set $(y',y_d) =g(0,\dots,0,1)\in \bH_d$. Then,
\begin{align*}
D(y_d)^{-1}T_{-y'}g(0,\dots,0,1) = (0,\dots,0,1)
\end{align*}
and $D(y_d)^{-1}T_{-y'}g \in \fG_C$. Since $\cC D(y_d)^{-1}T_{-y'}g \cC^{-1} \in \fG$ send the origin to the origin,
there is $R \in \SO(d)$ such that $R = \cC D(y_d)^{-1}T_{-y'}g \cC^{-1}$ by \cite[Proposition 2.9]{Matsumoto}.
Hence,
\begin{align*}
g = T_{y'} D(y_d) (\cC^{-1}R \cC) \in \fG_C.
\end{align*}
By \eqref{eq_app_def_HC}, $H_C(gx,gy)=H_C(x,y)$ holds for $g = T_{y'}$ or $g= D(y_d)$.
Furthermore, by \eqref{eq_app_d_D}, it also holds for $g = \cC^{-1}R \cC$.
\end{proof}

\begin{proof}[Proof of Proposition \ref{prop_conf_identity}]
By Lemma \ref{lem_app_identity}, for any $g \in \fG$ and $x,y \in \bD$, we have
\begin{align*}
\frac{|x-y|}{(1-|x|^2)(1-|y|^2)} &= \frac{|g x-g y|}{(1-|g x|^2)(1-|g y|^2)}.
\end{align*}
Combining this equality with the well-known identity
\begin{align*}
|x-y|^2 + (1-|x|^2)(1-|y|^2) = 1-2(x,y)+|x|^2|y|^2
\end{align*}
we obtain \eqref{eq_identity_G} from Lemma \ref{lem_norm_all_inv}.
%
%
%
\end{proof}


\section{Category of Ind Hilbert spaces}
\label{app_Hilb}
In this Appendix, we recall the definition of $\mathrm{Ind}(\Hilb)$ and the symmetric monoidal structure on it.
All propositions in this section are standard; for details on the categorical aspects, see \cite[Section 6]{KS}, and for the theory of Hilbert spaces, see \cite[Chapter IV]{Takesaki}.
\begin{dfn}\label{def_ind_category}
Let $\cC$ be a category. The ind-category $\operatorname{Ind}(\cC)$ is a full subcategory of $\mathrm{PSh}(\cC)=[\cC^\op,\mathrm{Set}]$, the category of presheaves on $\cC$, consisting of small filtered colimits of representables, i.e., 
\begin{align*}
X \cong \mathrm{colim}_{i\in I} \mathrm{Hom}_{\cC}(\bullet,F(i)),
\end{align*}
for some small filtered category $I$ with a functor $F:I \rightarrow \cC$.
\end{dfn}

Let $F:I \rightarrow \cC$ and $G:J \rightarrow \cC$ be objects in $\operatorname{Ind}(\cC)$. Then,
\begin{align}
\begin{split}
\Hom_{\operatorname{Ind}(\cC)}(F,G) &\cong \Hom_{\operatorname{Ind}(\cC)}(\colim_{i \in I} F(i), \colim_{j \in J} G(j))\\
& \cong \lim_{i \in I}\Hom_{\operatorname{Ind}(\cC)}(F(i), \colim_{j \in J} G(j))\\
& \cong \lim_{i \in I}\left(\colim_{j \in J} \Hom_{\cC}(F(i), G(j))\right)
\end{split}
\label{app_def_ind_hom}
\end{align}
holds (see \cite[Section 6.1]{KS}).

%

\begin{prop}\cite[Corollary 6.1.18]{KS}\label{prop_colimit}
If $\cC$ has all finite colimits, then $\operatorname{Ind}(\cC)$ has all small colimits.
\end{prop}

Assume that $\cC$ is a symmetric monoidal category with the tensor product $\otimes_\cC$ and the unit $1_\cC$.
Then, one can define a bifunctor, called the Day convolution \cite{Day},
\begin{align*}
\otimes_{\mathrm{Day}}:\mathrm{PSh}(\cC) \times \mathrm{PSh}(\cC) \rightarrow \mathrm{PSh}(\cC)
\end{align*}
by the coend
\begin{align}
(F \otimes_{\mathrm{Day}} G)(C) = \int^{A,B \in \cC}F(A) \times G(B) \times \mathrm{Hom}_\cC(C,A \otimes_\cC B).\label{eq_Day}
\end{align}
If $F = y_X=\mathrm{Hom}_\cC(\bullet,X)$ and $G = y_Y$ for some $X,Y \in \cC$,
then by the property of the coend, $y_X \otimes_{\mathrm{Day}} y_Y\cong y_{X \otimes_\cC Y}$.
Hence, the Day convolution coincides with $-\otimes_\cC-$ on $\cC \rightarrow \mathrm{PSh}(\cC)$,
which preserves all filtered colimits in each variable. In particular, if $F,G \in \mathrm{Ind}(\cC)$, then $F \otimes_{\mathrm{Day}} G \in \mathrm{Ind}(\cC)$. Hence, $-\otimes_{\mathrm{Day}}-$ defines a tensor product on $\mathrm{Ind}(\cC)$.
Note that the inclusion $i:\mathrm{Ind}(\cC) \hookrightarrow\mathrm{Psh}(\cC)$ preserves filtered colimits but does not preserve colimits in general.
Therefore, $-\otimes_{\mathrm{Day}}-$ on $\mathrm{Ind}(\cC)$ does not necessarily commute with small colimits.

Assume that $\cC$ has all finite colimits and that $-\otimes_\cC -$ separately preserves all finite colimits.
Then,
\begin{align*}
\cC \overset{X \otimes_\cC -}{\longrightarrow} \cC \overset{i}{\hookrightarrow} \mathrm{Ind}(\cC)
\end{align*}
preserves all finite colimits.
Then, by \cite[Proposition 5.3.6.2 and Example 5.3.6.8]{Lurie}, $X \otimes_\cC-: \cC \rightarrow \mathrm{Ind}(\cC)$ admits a unique extension to an all-colimit preserving functor $X \otimes -: \mathrm{Ind}(\cC) \rightarrow \mathrm{Ind}(\cC)$, which coincides with $X \otimes_{\mathrm{Day}}-$.
Since colimit preserving functors are stable under colimits, we have:
\begin{prop}\cite[Corollary 4.8.1.14]{Lurie2} \label{prop_Day}
Assume that $\cC$ is a symmetric monoidal category with the tensor product $\otimes_\cC$ and the unit $1_\cC$ such that
\begin{itemize}
\item
$\cC$ has all finite colimits.
\item
$-\otimes_\cC-$ separately preserves all finite colimits.
\end{itemize}
Then, $\mathrm{Ind}(\cC)$ is a symmetric monoidal category whose tensor product $\otimes_{\mathrm{ind}(\cC)}$ is defined by \eqref{eq_Day} and unit is $y_{1_\cC}$ such that:
\begin{itemize}
\item
$-\otimes_{\mathrm{ind}(\cC)}-$ distributes all small colimits separately in each variable.
\end{itemize}
\end{prop}

We apply the above general discussion to $\Hilb$. Let $\Hilb$ be the category of separable Hilbert spaces, i.e.\ the category whose objects are separable Hilbert spaces and whose morphisms are bounded linear operators.
First, $\Hilb$ is a $\C$-linear category with biproducts and a zero object.
For bounded linear maps $f,g:H_1 \rightarrow H_2$ between Hilbert spaces,
\begin{align*}
\ker(f-g)= \{v\in H_1 \mid f(v)=g(v) \}
\end{align*}
is a closed subspace of $H_1$, hence a Hilbert space, and the inclusion $\ker(f-g) \rightarrow H_1$ satisfies the universal property of an equalizer. Moreover,
\begin{align*}
\mathrm{im}(f-g)^\perp = \{w \in H_2 \mid (w, f(v)-g(v))_{H_2}=0 \text{ for any } v \in H_1\}
\end{align*}
is a closed subspace of $H_2$, and $H_2 \cong \overline{\mathrm{im}(f-g)} \oplus  \mathrm{im}(f-g)^\perp$ gives an orthogonal decomposition of the Hilbert space. Let $p:H_2 \rightarrow \mathrm{im}(f-g)^\perp$ be the orthogonal projection.
Given any Hilbert space $H_3$ and a bounded linear map
$h:H_2 \rightarrow H_3$ with $h\circ (f-g)=0$, we have $\ker{h} \supset \overline{\mathrm{im}(f-g)}$, and hence $h$ factors through $p$ as a bounded linear map. Therefore, $\Hilb$ is a $\C$-linear category with finite limits and finite colimits.

Note that $\Hilb$ does not have infinite coproducts, and hence it is not a cocomplete category. Indeed, consider the sequence of one-dimensional Hilbert spaces $\C e_n$ with $(e_n,e_n)=1$ ($n \geq 1$).
Then $A_n: e_n \mapsto n e_n$ is a bounded operator, but $\oplus_{n \geq 0} A_n:\bigoplus_{n \geq 0}\C e_n \rightarrow \bigoplus_{n \geq 0}\C e_n$ is clearly unbounded.

%

Let $H_i$ be Hilbert spaces with the inner product $(-,-)_i$ ($i=1,2$).
Denote by $H_1 {\otimes} H_2$ the algebraic tensor product of $\C$-vector spaces.
Then, $H_1 {\otimes} H_2$ inherits an inner product by
\begin{align}
(v_1\otimes v_2,w_1\otimes w_2) =(v_1,w_1)_1 (v_2,w_2)_2\label{eq_app_tensor_inner}
\end{align}
for $v_1,w_1\in H_1$ and $v_2,w_2\in H_2$.
The Hilbert space completion of $H_1 {\otimes} H_2$ is denoted by $H_1 {\hotimes} H_2$.
For bounded linear maps $T_i:H_i \rightarrow K_i$ between Hilbert spaces,
the linear map between algebraic tensor products
$T_1 \otimes T_2: H_1 \otimes H_2 \rightarrow K_1 \otimes K_2$
admits a unique extension
$T_1 \hotimes T_2: H_1 \hotimes H_2 \rightarrow K_1 \hotimes K_2$,
and its operator norm satisfies
$\norm{T_1 \hotimes T_2}\leq \norm{T_1}\norm{T_2}$
(see for example \cite[Chapter IV.1]{Takesaki}).
With respect to this tensor product $\hotimes$, the category of Hilbert spaces $\Hilb$
carries the structure of a symmetric monoidal category.
%
\begin{rem}
\label{rem_subtle_bilinear}
As discussed in Section \ref{sec_ind_Hilb}, the relationship between bilinear maps on Hilbert spaces and the tensor product $\hotimes$ is subtle.
Let $l^2(\N)$ be the real Hilbert space with orthonormal basis $\{e_i\}_{i\in\N}$.
Define a bilinear map $B: l^2(\N) \times l^2(\N) \rightarrow \R$ by
\begin{align*}
B\left(\sum_{n\geq 0}a_n e_n, \sum_{m \geq 0}b_m e_m\right) = \sum_{n \geq 0} a_n b_n .
\end{align*}
Since the arguments are elements of $l^2(\N)$, we have
$\sum_{n \geq 0} |a_n b_n| \leq (\sum_{n \geq 0} |a_n|^2)^{\ft}(\sum_{n \geq 0} |b_n|^2)^{\ft}$,
thus, $B$ is well-defined, and it induces a linear map from the algebraic tensor product
${B}: l^2(\N) \otimes l^2(\N) \rightarrow \R$.
As we have already seen, ${B}$ satisfies $|{B}(a \otimes b)| \leq |a|\,|b|$ for any $a,b \in l^2(\N)$.
However, for the vectors satisfying
$\norm{\frac{1}{\sqrt{N}}\sum_{i=1}^N e_i \otimes e_i}_{l^2(\N)\otimes l^2(\N)}=1$,
we have
\begin{align*}
B\left(\frac{1}{\sqrt{N}} \sum_{i=1}^N e_i \otimes e_i\right) = \sqrt{N}.
\end{align*}
Therefore, $B$ does not admit an extension to a bounded linear operator between Hilbert spaces
$\tilde{B}: l^2(\N) \hotimes l^2(\N) \rightarrow \R$.
%
\end{rem}

We will show that $H \hotimes -: \Hilb \rightarrow \Hilb$ preserves all finite colimits.
For any $V_i \in \Hilb$ ($i=1,\dots,N$), since the canonical isomorphism of $\C$-vector spaces
\begin{align*}
 \bigoplus_{i=1}^N H {\otimes} V_i \rightarrow H {\otimes} \bigoplus_{i=1}^N V_i  
\end{align*}
is isometric, it induces an isometric isomorphism of Hilbert spaces. Hence, $H \hotimes -$ preserves all finite coproducts by $H\hotimes 0 \cong 0$.
Let $f:V \rightarrow W$ be a map in $\Hilb$
and let $p:W \rightarrow (\mathrm{Im}f)^\perp$ be the orthogonal projection with respect to the decomposition $W \cong (\mathrm{Im}f)^\perp \oplus \overline{\mathrm{Im}f}$.
Then, since $\mathrm{Im}(f\hotimes \id_H)^\perp = \mathrm{Im}(f)^\perp \hotimes H$, it follows that $H \hotimes -$ preserves coequalizers.
Thus, by Proposition \ref{prop_colimit} and Proposition \ref{prop_Day}, we have:
\begin{cor}\label{cor_ind_Hilb}
The category $\Ind$ is a cocomplete symmetric monoidal category whose tensor product distributes all small colimits separately.
\end{cor}

Since $\Vect$ is cocomplete, we have:
\begin{prop}\cite[Proposition 6.1.9]{KS}\label{prop_}
The forgetful functor $U_0:\Hilb \rightarrow \Vect$ induces a functor $U:\Ind \rightarrow \Vect$ that preserves filtered colimits.
\end{prop}

Let $\N$ be the category whose objects are natural numbers and such that there is a unique morphism from $n$ to $m$ if and only if $n\leq m$.
Let $H^1 \subset H^2 \subset \dots$ be inclusions of separable Hilbert spaces that preserve the inner products. This can be regarded as an object of $\Ind$ via the functor
$H^\bullet:\N \rightarrow \Hilb,\ k\mapsto H^k$.
Then,  by \eqref{app_def_ind_hom},
\begin{align*}
\Hom_{\Ind}((H^\bullet)^{\otimes_{\mathrm{Day}} n},H^\bullet) \cong
\lim_{k_1,\dots,k_n \in \N^n} \colim_{K \in \N}\Hom_{\Hilb}(H^{k_1}\hotimes \cdots \hotimes H^{k_n}, H^K).
\end{align*}
Therefore, given a $\CE$-algebra in $\Vect$ equipped with a Hilbert space filtration $V=\cup_{k\geq 0} H^k$, the object $H^\bullet$ determines a $\CE$-algebra in $\Ind$.
Hence, by Proposition \ref{prop_Left_Kan}, we obtain a symmetric monoidal functor $\Embc \rightarrow \Ind$.
Note also that Proposition \ref{prop_universal_character_ind} follows from
\begin{align}
\mathrm{Hom}_{\Ind}(V^{\otimes n},\C)\cong \lim_{k_1,\dots,k_n \in \N^n} \mathrm{Hom}_{\Hilb}(H^{k_1}\hotimes \cdots \hotimes H^{k_n} ,\C). \label{eq_ind_hom_dual}
\end{align}

In particular, by the Riesz representation theorem, a compatible family $\{v_n \in H^n\}$ corresponds to an element of
$\mathrm{Hom}_{\Ind}(V,\C)$.

\begin{rem}\label{rem_shift_zero}
Note that for $W \in \Ind$, it may happen that
$\mathrm{Hom}_{\Ind}(W,\C)=0$ even if $W \neq 0$.
Let $S:l^2(\N) \rightarrow l^2(\N)$ be the bounded linear map defined by
$(x_0,x_1,x_2,\dots) \mapsto (x_1,x_2,\dots)$.
Then its adjoint $S^*$ is given by
$(x_0,x_1,x_2,\dots) \mapsto (0,x_0,x_1,x_2,\dots)$.
Let $W$ be the ind Hilbert space defined by the sequence
\begin{align*}
W=\left(l^2(\N) \overset{S}{\rightarrow} l^2(\N) \overset{S}{\rightarrow} l^2(\N) \overset{S}{\rightarrow} \cdots\right).
\end{align*}
Then $W \neq 0$, but $\mathrm{Hom}_{\Ind}(W,\C)=0$.
In particular, the fact that no continuous state exists does not necessarily imply that
$\mathrm{Lan}_A(M)=0$.
However, in the case of ind Hilbert spaces given by an isometric increasing sequence as in
Definition \ref{def_ind_CF}, there exist plenty of elements in the dual.
\end{rem}

\noindent
\begin{center}
{\bf Acknowledgements}
\end{center}

I express my gratitude to Yoh Tanimoto and Maria Stella Adamo for collaborative work
and for discussions on quantum field theory from an analytic perspective,
to Mayuko Yamashita for discussions on topological quantum field theory,
to Masahito Yamazaki and Slava Rychkov for discussions on physical aspects,
and to Hiro-Lee Tanaka, Takumi Maegawa, and Vladimir Sosnilo for discussions related to category theory.
I also wish to express my gratitude to Yuji Tachikawa, 
Masaki Natori, Naruki Masuda, Tomohiro Asano, Masahiro Futaki,
Hiroshi Ooguri, Tomoyuki Arakawa, Shintaro Yanagida, Atsushi Matsuo, Hiroshi Yamauchi, Yoshihisa Saito, Toshiro Kuwabara, Yasuyuki Kawahigashi for valuable discussions and comments. This work is supported by Grant-in Aid for Early-Career Scientists (24K16911).


\begin{thebibliography}{100}

\bibitem[At]{Atiyah}
M.~Atiyah, Topological quantum field theory, Publications Mathematiques de l'IHES, 68, 1988, 175--186.


\bibitem[ABR]{ABR}
{S.~Axler, P.~Bourdon and W.~Ramey},
{Harmonic Function Theory}, 2nd Edition,
 {Graduate Texts in Mathematics},
{\bf137}, {(2)}, {Springer} {2001}.



 \bibitem[AF1]{AF1}
 {D.~Ayala and J.~Francis},
  {Factorization homology of topological manifolds},
{Journal of Topology}, {\bf 8}, {(4)},
 {1045--1084}, {2015}.

\bibitem[AF2]{AF2}
{D.~Ayala and J.~Francis},
The cobordism hypothesis, arXiv:{1705.02240}.

%
%

\bibitem[AFR]{AFR}
{D.~Ayala, J.~Francis and N.~Rozenblyum}
{Factorization homology {I}: Higher categories},
{Advances in Mathematics}, {\bf333},
 {1042--1177}, {2018}.

\bibitem[AFT]{AFT}
{D.~Ayala, J.~Francis and H.L.~Tanaka}
{Factorization homology of stratified spaces},
{Selecta Mathematica. New Series}, {\bf23},
 {(1)}, {293--362}, {2017}.
 
\bibitem[AGT]{AGT} 
M. S.~Adamo, L.~ Giorgetti, and Y.~Tanimoto, Wightman fields for two dimensional conformal field theories with pointed representation category. Comm. Math. Phys., 404(3):1231–1273, 2023.

\bibitem[AMT]{AMT}
M. S.~Adamo, Y.~Moriwaki and Y. Tanimoto, Osterwalder-Schrader axioms for unitary full vertex operator algebras, arXiv:2407.18222.




%

\bibitem[Arai]{Arai}
{A.~Arai}, {Analysis on Fock Spaces and Mathematical Theory of Quantum Fields:
               An Introduction to Mathematical Analysis of Quantum Fields},
{\bf2}, {World Scientific}, {2025}.


%
%
%
%
\bibitem[Bo]{B1} R.E.~Borcherds,
Vertex algebras, {K}ac-{M}oody algebras, and the {M}onster,
Proc. Nat. Acad. Sci. U.S.A.,
{\bf83}, 1986, {(10)},
{3068--3071}.
%
%

\bibitem[BD]{BD}
{A.~Beilinson and V.~Drinfeld},{Chiral Algebras},
 {American Mathematical Society},
{AMS Colloquium Publications},
 {51}, {2004}.









\bibitem[CS]{CS}
{D.~Calaque and C.~Scheimbauer}, 
{A note on the $(\infty,n)$-category of cobordisms}, {Algebraic \& Geometric Topology}, {\bf19}, {2019}, {533--655}.

\bibitem[Ch]{Che55}
Shiing-Shen Chern, An elementary proof of the existence of isothermal parameters on
a surface, Proc. Amer. Math. Soc. 6 (1955), 771--782.


\bibitem[CG1]{CG1}
{K.~Costello and O.~Gwilliam},
{Factorization Algebras in Quantum Field Theory. Volume 1},
{New Mathematical Monographs},
{\bf31}, {Cambridge University Press}, {2016}.

\bibitem[CG2]{CG2}
{K.~Costello and O.~Gwilliam},
{Factorization Algebras in Quantum Field Theory. Volume 2},
{New Mathematical Monographs},
{\bf41}, {Cambridge University Press}, {2021}.




\bibitem[CKLW]{CKLW}
S.~Carpi, Y.~Kawahigashi, R.~Longo, and M.~Weiner. From vertex operator algebras to conformal nets and back. Mem. Amer. Math. Soc., 254(1213): vi+85, 2018.


%
%
%
%
%

%
%



 
 
\bibitem[Day]{Day}
B.~Day, On closed categories of functors, Reports of the Midwest Category Seminar IV. Lecture Notes in Mathematics, {\bf137}, Springer, Berlin, 1970.

\bibitem[De]{De}
{D.W.~DeTemple},
{An Area Method for Systems of Univalent Functions Whose Ranges do not Overlap},
{Mathematische Zeitschrift}, {\bf128},
 {23--33}, {1972}.


\bibitem[Dy]{Dy}
{F. J.~Dyson},
{Divergence of Perturbation Theory in Quantum Electrodynamics},
{Phys. Rev.}, {\bf85},
 {(4)}, {631--632}, {1952}.



%
 

\bibitem[DL]{DL}
C.~Dong and X.~Lin, Unitary vertex operator algebras. J. Algebra, 397, 252-277, 2014.

\bibitem[DX]{DX}
F.~Dai and Y.~Xu, Approximation Theory and Harmonic Analysis on Spheres and Balls, Springer
Monographs in Mathematics, Springer, 2013.
%






\bibitem[FLM]{FLM}
I. ~Frenkel, J. ~Lepowsky, and A. ~Meurman, 
{Vertex operator algebras and the {M}onster}, {Pure and Applied Mathematics},
{\bf134}, {Academic Press, Inc., Boston, MA}, {1988}.



%
\bibitem[FHL]{FHL} I. ~Frenkel, Y. ~Huang and J. ~Lepowsky,
On axiomatic approaches to vertex operator algebras and modules, Mem. Amer. Math. Soc., {\bf104}, 1993, (494).

\bibitem[FMS]{FMS}
{P.~Di Francesco, P.~ Mathieu and D.~S\'{e}n\'{e}chal},
{Conformal field theory}, {Graduate Texts in Contemporary Physics}, {Springer-Verlag, New York}, {1997}.
%
%
%
%





 

%
%
%
\bibitem[GKRV]{GKRV2}
C.~Guillarmou, A.~Kupiainen, R.~Rhodes, and V.~Vargas,
{Segal's axioms and bootstrap for Liouville Theory},
{arXiv:2112.14859 [math-ph]}.
%


\bibitem[GW]{GW}
L.~G{\aa}rding and A.S.~Wightman,
{Fields as Operator-valued Distributions in Relativistic Quantum Theory}, {Arkiv f{\"o}r Fysik},
{\bf28}, {129--189}, {1964}.


\bibitem[GJ]{Glimm-Jaffe}
{J.~Glimm and A.~Jaffe},
{Quantum Physics: A Functional Integral Point of View}, {\bf2}, {Springer-Verlag}, {New York}, {1987}.



%





\bibitem[Hum]{Hum}
J. E.~Humphreys,
Introduction to Lie algebras and representation theory, Graduate texts in mathematics, {\bf9}, Springer-Verlag, 1973.
%
%
\bibitem[HK]{HK1}
{Y.-Z.~Huang, L.~Kong},
{Full field algebras}, {Comm. Math. Phys.},
{\bf272}, {2007}, {(2)}, {345--396}.
%
%

%



%



\bibitem[Jaf]{Jaffe}
{A.~Jaffe}, {Divergence of perturbation theory for bosons},
 {Communications in Mathematical Physics},
{\bf1}, {(2)}, {127--149}, {1965}.


\bibitem[Jan]{Janson}
{S.~Janson},
 {Gaussian Hilbert Spaces}, {129},
{Cambridge University Press},
 {1997}.

\bibitem[Ko]{Kobayashi}
{S.~Kobayashi}, {Transformation Groups in Differential Geometry}, {Classics in Mathematics},
 {Springer}, {Berlin, Heidelberg}, {1972}.

%
%
%
%




\bibitem[Kui1]{Kuiper}
N. H.~Kuiper, On conformally flat spaces in the large, Ann. of Math., {\bf 50}, 1949, 916--924.


\bibitem[Kui2]{Kuiper2}
N. H. ~Kuiper, On Compact Conformally Euclidean Spaces of Dimension $> 2$, Ann. of Math.,
{\bf 52}, (2), 1950, 478--490.



\bibitem[KS]{KS}
M.~Kashiwara and P.~Schapira, Categories and Sheaves, Grundlehren der Mathematischen Wissenschaften 332, 
Springer (2006).

%
%

%
%
%
%
%
%
\bibitem[Li]{Li}
H. ~Li, {Symmetric invariant bilinear forms on vertex operator algebras}, {J. Pure Appl. Algebra}, {\bf96}, {1994}, {(3)}, {279--297}.


%



\bibitem[Lee]{Lee}
{J. M.~Lee},
 {Introduction to Smooth Manifolds}, {2nd Edition}, {2013}, {Springer}.

\bibitem[Lio]{Liouville}
J.~Liouville, Extension au cas des trois dimensions de la questions du trace geographic, in Application de l'Anlyse a la Geometrie, Paris, 1850, 609--616. 

\bibitem[LV]{LoV}
J-L.~Loday and B.~Vallette, Algebraic operads, 346, Grundlehren der mathematischen Wissenschaften, Springer, Heidelberg, 2012, xxiv+634.

\bibitem[Lu1]{Lurie}
 {J.~ Lurie}, {Higher Topos Theory}, {Annals of Mathematics Studies},
 {\bf170}, {Princeton University Press}, {2009}.

\bibitem[Lu2]{Lurie2}
 {J.~Lurie}, {Higher Algebra}, {online version}, {September 18, 2017}.

\bibitem[Lu3]{Lurie3}
 {J.~Lurie}, {On the Classification of Topological Field Theories},
{Current Developments in Mathematics, 2008}, {129--280},
 {2009}, {International Press of Boston}.


%



%
%
%
%


\bibitem[Mat]{Matsumoto}
S.~Matsumoto, Foundations of Flat Conformal Structure,
Adv. Stud. Pure Math., {\bf 20}, Aspects of Low Dimensional Manifolds, 1992, 167--261.




\bibitem[Mi]{Milnor}
{J. ~Milnor}, {Microbundles. I},
 {Topology}, {1964}, {\bf 3},
 {Suppl. 1}, {53--80}.


\bibitem[Mo1]{M1} 
Y.~Moriwaki, Two-dimensional conformal field theory, full vertex algebra and current-current deformation, Adv. Math, 427, (2023).

\bibitem[Mo2]{MSwiss}
Y.~Moriwaki, Convergence and operadic compatibility of bulk and boundary OPEs in two-dimensional conformal field theory, arXiv:2410.02648.

\bibitem[Mo3]{MBergman}
Y.~Moriwaki, 
Bergman space, Conformally flat 2-disk operads and affine Heisenberg vertex algebra, arXiv:2603.06491.

\bibitem[Mo4]{Mfactorization}
Y.~Moriwaki, 
Prefactorization algebras for the conformal Laplacian: Central charge and Hilbert Fock space, arXiv:2602.17549.

\bibitem[Mo5]{Mvertex}
Y.~Moriwaki, System of conformally covariant correlation functions and operator product expansions in higher dimensions, in preparation.



%
%
%
%




%
\bibitem[OS1]{OS1}
 {K.~Osterwalder and R.~Schrader},
{Axioms for {E}uclidean {G}reen's functions},
{Comm. Math. Phys.}, {\bf 31}, {1973}, {83--112}.
\bibitem[OS2]{OS2}
 {K.~Osterwalder and R.~Schrader},
Axioms for Euclidean Green’s functions. II. Comm. Math. Phys., {\bf42}, 1975, 281--305.




%




%



\bibitem[RS]{RS}
{M.~Reed and B.~Simon},
 {Methods of Modern Mathematical Physics, Vol. I: Functional Analysis},
 {Academic Press},
{New York},{1980}

%
%
 \bibitem[St]{Steh2}
L.~Stehouwer, The Spin-Statistics Theorem for Topological Quantum Field Theories, Commun. Math. Phys, 405, 253 (2024).

\bibitem[Sche]{Sc}
{C.I.~Scheimbauer},
{Factorization Homology as a Fully Extended Topological Field Theory},
{2014}, {ETH Z{\"u}rich}, {Doctoral Thesis (Dr.\ sc.\ ETH Zurich)}.




\bibitem[Se]{Segal}
G.~Segal, The definition of conformal field theory. Topology, geometry and quantum field theory, 421--577, London Math. Soc. Lecture Note Ser., 308, Cambridge Univ. Press, Cambridge, 2004.


\bibitem[SW]{SW}
R. F.~Streater and A. S.~Wightman, PCT, spin and statistics, and all that, Princeton Landmarks in Physics. Princeton University Press, Princeton, NJ, 2000.


\bibitem[ST]{ST}
S.~Stolz and P.~Teichner, Supersymmetric field theories and generalized cohomology, Mathematical foundations of quantum field theory and perturbative string theory., {\bf83}, Proc. Sympos. Pure Math. Amer. Math. Soc., Providence, RI, 2011, 279--340.


\bibitem[Ta]{Takesaki}
{M.~Takesaki},
{Theory of Operator Algebras I},
{Encyclopaedia of Mathematical Sciences},
{\bf124}, {Springer-Verlag}, {1979}.


\bibitem[Te]{Te}
{J. E.~Tener},
{Construction of the Unitary Free Fermion Segal CFT},
{Communications in Mathematical Physics}, {355},
 {(2)}, {463--518}, {2017}.





\bibitem[TUY]{TUY}
{A.~Tsuchiya and K.~Ueno and Y.~Yamada},
{Conformal field theory on universal family of stable curves with gauge symmetries},
{Advanced Studies in Pure Mathematics}, {\bf19},
 {459--566}, {1989}.


%

%
%
%
%

%
%
\bibitem[Zh]{Zh}
{Y.~Zhu}, {Modular invariance of characters of vertex operator algebras},
{J. Amer. Math. Soc.}, {\bf9}, {1996}, {(1)}, {237--302}.
\end{thebibliography}
\end{document}